\documentclass[11pt,a4paper]{article}
\usepackage[margin=2.8cm,includefoot,heightrounded]{geometry}

\usepackage{url} 
\usepackage{hyperref} 
\hypersetup{
	pdftitle={},
    pdfauthor={Pavel Drozdov},
    colorlinks,
    citecolor=blue,
    filecolor=black,
    linkcolor=blue,
    urlcolor=blue,
    final=true
} 
\usepackage[dvipsnames]{xcolor} 
\usepackage{booktabs} 
\usepackage{tikz}
\usetikzlibrary{arrows,matrix,positioning,arrows.meta}
\usetikzlibrary{decorations.pathreplacing, calc,fit}
\usepackage{tikz-cd}
\usepackage[normalem]{ulem}
\usepackage{tikzsymbols} 

\usepackage{blkarray}

\usepackage{amsmath}
\usepackage{amsthm}
\usepackage{slashed} 
\usepackage{tensor} 
\usepackage{mathtools}
\usepackage{physics}
\usepackage{braket}
\usepackage{stmaryrd}
\usepackage{array}
\newcolumntype{C}{>{\hfil$}p{1.1cm}<{$\hfil}} 
\newcolumntype{M}{>{\hfil$}p{0.6cm}<{$\hfil}} 

\renewcommand{\imath}{\mathrm{i}}

\usepackage{kpfonts}

\usepackage{authblk}

\title{Explicit isomorphisms for the symmetry algebras of continuous and discrete isotropic oscillators}

\author[1]{Pavel Drozdov} 
\author[2]{Giorgio Gubbiotti}
\author[2]{Danilo Latini} 

\makeatletter
\renewcommand\AB@affilsepx{\\\vspace{0.5em}}
\makeatother

\affil[1]{Dipartimento di Scienze Matematiche, Informatiche e Fisiche, 
  Universit\`a degli Studi di Udine, 33100 Udine, Italy  \&   INFN Sezione di Trieste, 34127 Trieste, Italy  } 
\affil[2]{Dipartimento di Matematica ``Federigo Enriques'',  
Universit\`a degli Studi di Milano \& INFN Sezione di Milano, 20133 Milan, Italy} 
\bgroup

\affil[ ]{\footnotesize \ttfamily \emph{e-mails:}
     \emph{\href{mailto:drozdov.pavel@spes.uniud.it}{drozdov.pavel@spes.uniud.it}, 
     \href{mailto:giorgio.gubbiotti@unimi.it}{giorgio.gubbiotti@unimi.it},   
       \href{mailto:danilo.latini@unimi.it}{danilo.latini@unimi.it}
      } }
\egroup

\newcommand{\R}{\mathbb{R}}

\DeclareMathOperator{\ud}{d}
\DeclareMathOperator{\Id}{Id}
\DeclareMathOperator{\Span}{span_{\R}}
\DeclareMathOperator{\Mat}{Mat}
\DeclareMathOperator{\diag}{diag}
\DeclareMathOperator{\Der}{Der}

\DeclareMathOperator{\sgn}{sgn}

\DeclareMathOperator{\ad}{ad}

\renewcommand{\vec}{\boldsymbol} 
\newcommand{\varomega}{\tilde{\omega}}

 \usepackage{nicematrix}
\NiceMatrixOptions{nullify-dots}
\newcommand{\FradPlain}{\mathfrak{A}}
\newcommand{\FradZero}{\mathfrak{A}^{(0)}}

\newcommand{\Frad}{\mathfrak{A}^{+}} 
\newcommand{\frad}{\mathfrak{a}^+}

\newcommand{\F}{F}

\newcommand{\mFrad}{\mathfrak{A}^{-}}
\newcommand{\mfrad}{\mathfrak{a}^-} % with minus 

\newcommand{\so}{\mathfrak{so}}

\newcommand{\su}{\mathfrak{su}}
\renewcommand{\u}{\mathfrak{u}}
\newcommand{\gl}{\mathfrak{gl}}
\renewcommand{\sl}{\mathfrak{sl}}

\newcommand{\Z}{\mathbb{Z}}

\DeclareMathOperator{\Dom}{Dom} 

\usepackage{bbold} 

\newcommand{\hyphen}{\hbox{--}}
\newcommand{\ie}{\textit{i.e.} }
\newcommand{\mm}{\textit{mutatis mutandis}} 
\newcommand{\eg}{\textit{e.g.} }

\usepackage[shortlabels]{enumitem}

\usepackage[tiny]{titlesec}
\usepackage{titletoc}  
\titleformat{\section}
  {\normalfont \bfseries \scshape }{\thesection. }{0.1em}{}

\titlecontents{section}[0pt]{\addvspace{1ex}}{\thecontentslabel. \normalfont \bfseries \scshape}
{\Large\bfseries}{\hspace{1em plus 1fill}\large\bfseries \contentspage}

\titleformat{\subsection}
  {\normalfont \bfseries}{\thesubsection. }{0.1em}{}
  \titleformat{\subsubsection}
  {\normalfont \small \bfseries}{\thesubsubsection. }{0.1em}{}
  
  \renewcommand{\thesection}{\Roman{section}}

\numberwithin{equation}{section}

\usepackage{amsthm}
\usepackage[capitalize]{cleveref}

\theoremstyle{theorem}
\newtheorem{theorem}{Theorem}[section]

\newtheorem*{theorem*}{Theorem}
\newtheorem*{main}{Main Theorem}
\newtheorem{corollary}[theorem]{Corollary}
\newtheorem{lemma}[theorem]{Lemma}
\newtheorem{prop}[theorem]{Proposition}

\newtheorem*{problem}{Problem}

\theoremstyle{remark}
\newtheorem{remark}[theorem]{Remark}

\theoremstyle{definition}
\newtheorem{definition}[theorem]{Definition}
\newtheorem{example}[theorem]{Example} 

\makeatletter
\renewenvironment{proof}[1][\relax]{\par
  \pushQED{\qed}%
  \normalfont \topsep6\p@\@plus6\p@\relax
  \trivlist
  \item[\hskip\labelsep\itshape
    \ifx#1\relax \proofname\else\proofname{} of 
    #1\fi\@addpunct{.}]\ignorespaces
}{%
  \popQED\endtrivlist\@endpefalse
}
\makeatother

\usepackage[backend=biber,sorting=nyt,maxnames=99,giveninits=true,style=alphabetic,doi=false,url=false]{biblatex}
\allowdisplaybreaks

\begin{document}

\maketitle

\begin{abstract}
    \noindent
    We present a detailed study of a parametric Lie algebra encompassing the
    symmetry algebras of various models, 
    both continuous and discrete. 
    This algebraic structure characterizes the isotropic oscillator (with positive, purely imaginary, and zero frequency) and one of its possible nonlinear deformations.
    We demonstrate a novel occurrence of this Lie algebra in the framework of maximally superintegrable discretizations of the isotropic
    harmonic oscillator. 
    In particular, we also show that the continuous model and one of its discretizations admit a Nambu--Hamiltonian structure.
    Through an in-depth analysis of the properties
    characterizing the Lie algebra  in the abstract
    setting, for different values of the parameter, we find explicit expressions of the Killing forms and construct explicit isomorphism maps to $\u_N$,
    $\gl_N(\mathbb{R}) $, and a semidirect sum of $\so_N (\R)$ with $\R^{N(N+1)/2}$. 
    Notably, due to the above isomorphisms, our formulas hold true for $\su_N$ and $\sl_N(\R) $ and are valid for arbitrary $N$.

\end{abstract} 

\tableofcontents

\section{Introduction and origins of the algebra}
\label{sec:intro}

In classical and quantum mechanics, the isotropic harmonic oscillator
(IHO)~\cite{Arnold1978.Mathematical_methods_classical_mechanics,
faddeevLecturesQuantumMechanics2009} is a well-studied and illustrative example
of a maximally superintegrable (MS) system\ \cite{MillerPostWinternitz2013R}.
With the suitable choice of units, its Hamiltonian can be written as:  
\begin{equation}\label{eq:cont-Hamiltonian}
    H = \frac{1}{2}  \sum_{k=1}^N (p_k^2 + \alpha  q_k^2),
    \qquad \alpha\in\R.
\end{equation}
When the parameter $\alpha$ is positive, it is 
related to the frequency of the IHO.

In this setup, a Demkov--Fradkin tensor is a quantity that encompasses enough
integrals of motion to make the IHO in $N$ degrees of freedom maximally
superintegrable~\cite{Demkov1958, fradkinExistenceDynamicSymmetries1967}. Its
expression in terms of Darboux variables $(q_i, p_i)$ is given by:
\begin{equation}
    F_{i,j} = p_i p_j + \alpha  q_i q_j. 
\end{equation}
Together with the angular momenta, \ie\!:
\begin{equation}\label{eq:discr-angular-momentum-in-stmplectic-realization}
    L_{i,j} = q_i p_j - q_j p_i,
\end{equation}  
the components of the Demkov--Fradkin tensor form a Lie
algebra with respect to the canonical Poisson bracket $\{q_i, p_j \} =
\delta_{i,j}$ for all values of the energy:
\begin{subequations}\label{eq:commrelfrad} 
    \begin{align}
        \bigl\{L_{i,j},L_{k,l}  \bigr\} &=
        L_{j,l}\delta_{i,k}+L_{k,j}\delta_{l,i}
        +L_{l,i}\delta_{j,k}+L_{i,k}\delta_{l,j},
        \\
        \bigl\{L_{i,j},F_{k,l}  \bigr\} &=
        F_{j,l}\delta_{i,k}+ F_{k,j}\delta_{i,l}
        -F_{i,l}\delta_{j,k} 
        -F_{i,k}\delta_{j,l},
        \\
        \bigl\{ F_{i,j},F_{k,l}  \bigr\} &=
        \alpha \left(  
            L_{j,l}\delta_{i,k}+L_{j,k}\delta_{i,l}
            +L_{i,l}\delta_{j,k}+L_{i,k}\delta_{j,l}
        \right), 
    \end{align}%
\end{subequations}
where $i,j,k,l = 1, \dots N $. 

Throughout the years the Demkov--Fradkin tensor has been extended also to many
nonlinear MS
systems~\cite{Ballesteros_et_al2008PhysD,Gonera_etal2021,Kuru_etal2025}.  In
particular, we observe that in~\cite{Gonera_etal2021} an analogue of the
Demkov--Fradkin tensor appeared for the following deformed nonlinear
oscillator:
\begin{equation}
    H_{\lambda} \coloneqq 
    \frac{1}{2}  
    \dfrac{\sum_{k=1}^N (p_k^2 + \alpha  q_k^2)}{%
        1+\lambda \sum_{k=1}^{N}q_{k}^{2}}.
    \label{eq:defiho}
\end{equation}
The Hamiltonian~\eqref{eq:defiho} is the so-called Darboux III
system, which was extensively studied in the
literature~\cite{Ballesteros_et_al2008PhysD,Ballesteros_etal2011InternatJTheorPhys,Ballesteros_etal2011PhysLettA,Ballesteros_et_al2011AnnPhys,Kalninsetal2003,Latini_etal2016}.
Then, there is an associated deformed Demkov--Fradkin tensor  \emph{with frequency depending on the Hamiltonian}:
\begin{equation}
    F_{i,j}^{(\lambda)} = p_i p_j + \alpha(H_\lambda)  q_i q_j,
    \qquad \alpha(H_\lambda)\coloneqq \alpha - 2\lambda H_\lambda.
\end{equation}
Following the results of~\cite{Gonera_etal2021}, we have that the deformed
Demkov--Fradkin tensor and the angular momenta
form the same Lie algebra as the non-deformed one~\eqref{eq:commrelfrad} with
$\alpha$ replaced by $\alpha(H_\lambda)$. We emphasize that the sign of
$\alpha(H_{\lambda})$ depends on the value of the energy.
Later in this paper, we will show that this is not an isolated
    case, since a similar algebraic construction will arise from the problem of
    finding a MS discretization of the IHO. In this case, the sign of the
    parameter $\alpha$ is related to a discretization parameter, which
     disappears in the continuum limit.

Motivated by the appearance of the Lie algebra~\eqref{eq:commrelfrad} in
different problems related to oscillators, we can consider it abstractly, \ie
without referring to any particular realization in terms of Darboux variables.
To this end, let us introduce a parametric family of real  Lie algebras:
 \begin{equation}
  	 	 \FradPlain_{N}(\alpha)  \coloneqq 
    	 \Span 
    	 \bigl\{ 
    	 L_{i,j},  F_{i,j}
    	 \bigr\}_{1 \leq i,j \leq N}
    	 \big/ \mathfrak{I}, 
       % \notag \\  =& 
%       =\Span\Set{\Set{L_{i,j}}_{1\leq i<j\leq N},
%        \Set{F_{i,j}}_{1\leq i\leq j\leq N}}, 
        \label{eq:Fradgen}
  \end{equation}
    with Lie brackets: 
\begin{subequations}\label{eq:pb-Fradkin}
    \begin{align}
        \bigl[L_{i,j},L_{k,l}  \bigr] &=
        L_{j,l}\delta_{i,k}+L_{k,j}\delta_{l,i}
        +L_{l,i}\delta_{j,k}+L_{i,k}\delta_{l,j},
        \label{eq:LijLlm}
        \\
        \bigl[L_{i,j},F_{k,l}  \bigr] &=
        F_{j,l}\delta_{i,k}+ F_{k,j}\delta_{i,l}
        -F_{i,l}\delta_{j,k} 
        -F_{i,k}\delta_{j,l},
        \label{eq:LiFLlm}
        \\
        \bigl[F_{i,j},F_{k,l}  \bigr] &=
        \alpha \left(  
            L_{j,l}\delta_{i,k}+L_{j,k}\delta_{i,l}
            +L_{i,l}\delta_{j,k}+L_{i,k}\delta_{j,l}
        \right). 
        \label{eq:FijFlm}
    \end{align}%
\end{subequations}
Here, the symmetry and antisymmetry conditions are provided by an ideal: 
\begin{equation}\label{eq:ideal}
	\mathfrak{I} \coloneqq \left\{ 
	L_{j,i} = - L_{i,j};
	\quad 
	L_{i,i} =0;
	\quad 
	F_{i,j}=F_{j,i}
    \right\},
\end{equation} 
meaning that $\dim \FradPlain_N(\alpha)=N^2$.

Since we allow the real parameter $\alpha$ to be either positive, negative,
or zero, we consider the following three cases:
\begin{enumerate}[(i)]
    \item \textbf{Case $ \alpha>0$:} there exists an $\omega>0$ such that
        $\alpha = \omega^{2}$, and we will write
        $\Frad_{N}\coloneqq\FradPlain_{N}(\omega^{2})$.\label{item:case0-intro}
        %and we will address to this Lie algebra as an \emph{algebra of
    \item \textbf{Case $\alpha<0$:} there exists an $\omega>0$ such that $\alpha= -\omega^{2}$, and we will write
        $\mFrad_{N}\coloneqq\FradPlain_{N}(-\omega^{2})$.\label{case-minus}
    \item \textbf{Case $\alpha =0 $:} we will write
        $\FradZero_{N}\coloneqq\FradPlain_{N}(0)$.\label{item:case3}
\end{enumerate}

In this paper, we present a detailed study of the parametric Lie
algebra $\FradPlain_{N}(\alpha)$~\eqref{eq:Fradgen} in arbitrary dimension. Our
main result is:
\begin{main}\label{th:main}
    The following \textbf{explicit} Lie algebra isomorphisms hold:
    \begin{enumerate}[(i)]
        \item $\Frad_{N} 
        \cong \u_N $, \label{th:main:1} 
        \item $\mFrad_{N}  
        \cong \mathfrak{gl}_N (\R) $, \label{th:main:2} 
        \item $\FradZero_{N}  \cong \so_N (\R) \niplus  
           \mathbb{R}^\mathcal{N}$, where $\mathcal{N} \coloneqq N(N+1)/2$, \label{th:main:3}
    \end{enumerate}
    where $\mathfrak{u}_{N}$ denotes the unitary algebra,
        $\mathfrak{gl}_N(\R)$ the general linear algebra,
        $\mathfrak{so}_N(\R)$ the special orthogonal algebra, 
        $\niplus$ denotes the semidirect sum of Lie algebras, and we identify the one-dimensional abelian Lie algebra with the field of real numbers $\R$.  The explicit
    isomorphism maps are provided in the proofs of \Cref{th:iso-suN,th:iso-slN,th:Frad0-iso},
respectively.
\end{main}

\begin{remark}
    We remark that our statement is consistent with the results
    known in the literature for $N=2,3$. Indeed, for instance, 
    in~\cite{Gonera_etal2021} the following Lie algebra
    isomorphisms are presented  for $N=2$:
    \begin{equation}   \label{eq:frad2} 
        \FradPlain_{2}(\alpha) \cong
        \begin{cases}
            \mathfrak{u}_{2}, & \alpha>0,
            \\ 
            \mathfrak{su}_{1,1}\oplus \R , & \alpha<0, 
            \\ 
             \mathfrak{e}_{2}\oplus \R,
             & \alpha =0, 
        \end{cases}
    \end{equation}
    where $\su_{1,1}$ is the generalized special unitary algebra of signature
    $(1,1)$, and $\mathfrak{e}_2=\mathfrak{so}_2(\R)\niplus \R^2$ is the
    so-called Euclidean algebra. This result is equivalent to ours taking into
    account the well-known isomorphisms of low-dimensional Lie algebras, see
    \eg\!\cite{gilmoreLieGroupsLie2005}:
    $\su_{1,1}\oplus \R \cong  \sl_{2} (\R)\oplus \R\cong \gl_N(\R)$, and
    $\mathfrak{e}_{2}\oplus \R\cong\mathfrak{so}_2(\R)\niplus \R^3$. In the
    last isomorphism we used the fact that a direct sum can be seen as a trivial
    semi-direct sum. Moreover, from~\cite{fradkinExistenceDynamicSymmetries1967}
    we recall that for $N=3$ one has  $\Frad_3\cong\u_{3}$, which
    certainly agrees with the Main Theorem. Finally, for general $N$ we refer the reader to~\cite{Gonera_etal2021}, which again agrees with our Main Theorem.
    \label{rem:knownstuff}
\end{remark}

The paper is organized in the following way. In \Cref{sec:discretization} we
show how the Lie algebra~$\FradPlain_N (\alpha)$ arises also in the problem of
finding a MS discretization of the IHO. Then,
\Cref{sec:Frad-plus,sec:Frad-minus,sec:Frad-zero} are devoted to the detailed
study of Lie algebras $\Frad_N$, $\mFrad_N$ and $\FradZero_N$ respectively, resulting in
proofs of points \ref{th:main:1}, \ref{th:main:2}~\ref{th:main:3} of the Main
Theorem. Finally, in~\Cref{sec:concl} we conclude by discussing possible
future directions and applications of the outcomes of our investigation.

\section{The Lie algebra $\mathfrak{A}_{N}(\alpha)$ and the discretization
of the IHO}
\label{sec:discretization}

In this Section, we show how the Lie algebra $\mathfrak{A}_{N}(\alpha)$ also
appears if one considers the discretization of the isotropic harmonic oscillator
in $N$ degrees of freedom. The  discretization of the simple harmonic
oscillator (IHO in 1 degree of freedom) has been extensively studied in the
literature, see \eg\cite{Cieslinski2011}. Some possible discretizations of the
isotropic harmonic oscillator in $N$ degrees of freedom appeared more recently
in~\cite{GLT_coalgebra,DrozGub_h6} within  the context of the coalgebra
symmetry approach~\cite{Ballesteros_et_al_1996,BallesterosRagnisco1998} applied
to discrete-time systems. We will show how to obtain the discretization
starting from a very general one, and we will compare it with some well-known discretization procedures, such as the celebrated Kahan--Hirota--Kimura (KHK)
discretization~\cite{Kahan1993,KahanLi1997,HirotaKimura2000}. 
We will also discuss how this discretization is related to
Nambu--Hamiltonian mechanics~\cite{Nambu1973,Takhtajan1994} and to the
skew gradient discretization method proposed by Quispel and Capel
in~\cite{QuispelCapel1996}. To keep the material self-contained, as several
different theories are involved in the discussion,  before proceeding to our construction, we will first review the
basic definitions we need.

\subsection{Preliminary definitions}

Consider a system of first-order difference equations (O$\Delta$Es):
\begin{equation}
    \vec{x}' = \vec{f}(\vec{x}),
    \quad
    \vec{x}'=\vec{x}(t+h), 
    \label{eq:xpF}
\end{equation}
for an unknown vector valued sequence $\Set{ \vec{x}(t)}_{t\in
h\Z}\subset\R^{M}$, where $h>0$ is a fixed parameter.  It is well known,
see for instance~\cite{GJTV_class}, that fixed an initial condition
$\vec{x}(t_{0})=\vec{x}_{0}\in D = \Dom \vec{f} \subset \R^{M}$, such a system
of O$\Delta$Es is equivalent to the iteration of a map:
\begin{equation}
    \begin{tikzcd}[row sep=tiny]
        \vec{\Phi} \colon  D \arrow{r} & \R^{M}
        \\
        \vec{x} \arrow[mapsto]{r} & \vec{\Phi}(\vec{x}) = \vec{f}(\vec{x}),
    \end{tikzcd}
    \label{eq:phimap}
\end{equation}
\ie $\vec{x}(t_{0}+nh) = \vec{\Phi}^{n}(\vec{x}(t_{0}))$. In the case
when the function $\vec{F}$ is rational, with a rational inverse, the two
pictures are completely equivalent.

We are interested in the integrability of systems/maps of the
form~\eqref{eq:xpF}. This theory is carried out in complete analogy with
the continuous theory. In particular, a \emph{first integral}, or
\emph{invariant}, is a well-defined scalar function
$S\colon D \to \R$ such that its pullback through the map $\vec{\Phi}$ is
constant:
\begin{equation}
    \vec{\Phi}^{*}S(\vec{x}) = S(\vec{\Phi}(\vec{x})) = S(\vec{x}).
    \label{eq:invariant}
\end{equation}
Maps with no additional structure need a total of $M-1$ invariants, so it
is usual to restrict to classes of maps with additional properties. The most
relevant classes to this study are \emph{symplectic maps}. The theory of
symplectic integrable maps was developed in
\cite{Veselov1991,Bruschietal1991,Maeda1987}. For a complete overview
on the subject, we refer to
\cite[Chap. 6]{HietarintaJoshiNijhoff2016}, and the thesis \cite[Chap. 1]{TranPhDThesis}.

We recall the following:
\begin{definition}
    Let us assume we are given a smooth map $\vec{\Phi}\colon D \to \R^{2N}$,
    $D\subset \R^{2N}$ open and dense, and a closed two-form $\eta \in \Omega^{2}D$.
    Then, we say that the map $\vec{\Phi}$ is a \emph{symplectic map} with
    respect to the form $\eta$ if
    its pullback preserves the form $\eta$: $\vec{\Phi}^{*}\eta=\eta$. 
    \label{def:poisson}
\end{definition}

More explicitly we can write the condition  $\vec{\Phi}^{*}\eta=\eta$
in matrix form as:
\begin{equation}
    \frac{\partial \vec{\Phi}}{\partial\vec{x}}  
    \eta({\vec{\Phi}}(\vec{x}))  
    \left(\frac{\partial \vec{\Phi}}{\partial\vec{x}}\right)^{T}
    =\eta(\vec{x}),
    \label{eq:Poisson1}
\end{equation}
where $\partial \vec{\Phi}/\partial\vec{x}$ is the  Jacobian matrix of
$\vec{\Phi}$, and with a slight abuse of notation we denoted again by
$\eta$ the associated skew-symmetric matrix.

We recall that a closed two-form $\eta \in \Omega^{2}D$ induces a bilinear
operation on the module of smooth functions by:
\begin{equation}
    \left\{ S,R \right\}_{\eta} = 
    {\grad S(\vec{x})}^{T}\eta^{-1}(\vec{x})\grad R(\vec{x}) .
    \label{eq:pbgen}
\end{equation}
This is what is called a \emph{Poisson bracket}. If the bracket of
two smooth functions vanishes, then the functions are said to be in
involution.

\begin{remark} \label{rem:poisspres} 
    We observe that the condition~\eqref{eq:Poisson1} can be rewritten
    in terms of Poisson brackets as:
    \begin{equation}
        \vec{\Phi}^{*}\left\{ S,R \right\}_{\eta} = 
        \left\{ \vec{\Phi}^{*}S,\vec{\Phi}^{*}R \right\}_{\eta}.
        \label{eq:Poisson2}
    \end{equation}
\end{remark}
Then we have the following characterization of integrability for
symplectic maps:
\begin{definition}[\cite{Veselov1991,Bruschietal1991,Maeda1987}]
    A symplectic map $\vec{\Phi}\colon D \to \R^{2N}$, preserving
    the form $\eta \in \Omega^{2}D$, is \emph{Liou\-vil\-le integrable} if
    it possesses $N$ functionally independent invariants
    in involution with respect to the associated Poisson bracket.
    \label{def:lpint}
\end{definition}

If in the limit $h\to0^{+}$, possibly up to some rescaling, a system of
first-order difference equations~\eqref{eq:xpF} goes into a system of
First-order ordinary differential equations (ODEs):
\begin{equation}
    \vec{\dot{X}} = \vec{F}({\vec{x}}),
    \label{eq:xdotF}
\end{equation}
we say that~\eqref{eq:xpF} is a \emph{discretization} of~\eqref{eq:xdotF}.
Note that, in general, the problem of discretization is ill-posed, in the
sense that it does not admit a unique solution. Moreover, different
discretizations can have different properties, both from the analytical and
the numerical points of view, see for instance the discussion made
in~\cite{LeviMartinaWinternitz2015}.

In particular, if the system of ODEs~\eqref{eq:xdotF} is (super)integrable,
and so is the system of O$\Delta$Es~\eqref{eq:xpF}, we will say that the
latter is a \emph{(super)integrable discretization}, see for
instance~\cite{Suris2003book}. 

\begin{remark}
    We remark that a system of O$\Delta$Es can be an integrable
    discretization of a superintegrable system of ODEs. This was
    highlighted in the case of the rational anisotropic caged
    oscillator~\cite{Evans2008} in~\cite{GubLat_sl2}.
    \label{rem:supintdiscr}
\end{remark}

As stated above, we will also need some additional definitions, namely that
of \emph{Nambu--Poisson bracket} and \emph{discrete gradient}.

\begin{definition}
    Consider a smooth manifold $M$ of dimension $n$ and its
    ring of smooth functions $A=\mathcal{C}^{\infty}(M)$. A Nambu--Poisson
    bracket of order $n$ is a map:
    \begin{equation}
        \begin{tikzcd}[row sep=tiny]
            \left\{ ,\ldots, \right\} \colon A^{\otimes n}  \arrow{r} & A
            \\
            f_{1}\otimes \cdots \otimes f_{n}\arrow[mapsto]{r} & \left\{ f_{1},\ldots, f_{n}\right\},
        \end{tikzcd}
        \label{eq:nambup}
    \end{equation}
    such that the following three conditions are met:
    \begin{itemize}
        \item skew-symmetry with respect to any permutation of indices:
            given $\sigma\in\mathcal{S}_{n}$, symmetric group of $n$ elements, then:
            \begin{equation}
                \left\{ f_{\sigma(1)},\ldots,f_{\sigma(n)} \right\}
                =
                (-1)^{\sgn \sigma}\left\{ f_{1},\ldots,f_{n} \right\},
                \label{eq:skewnambu}
            \end{equation}
            for all $f_{1},\ldots,f_{n}\in A$, where $\sgn \sigma$ is the
            sign of the permutation $\sigma$;
        \item Leibnitz rule:
            \begin{equation}
                \left\{ f_{1}f_{2},f_{3},\ldots,f_{n+1} \right\}=
                f_{1}\left\{ f_{2},f_{3},\ldots,f_{n+1} \right\}+
                f_{2}\left\{ f_{1},f_{3},\ldots,f_{n+1} \right\},
                \label{eq:leibnambu}
            \end{equation}
            for all $f_{1},\ldots,f_{n+1}\in A$;
        \item the fundamental identity:
            \begin{equation} 
                \begin{gathered}
                    \left\{\left\{ f_{1},\ldots,f_{n} \right\},f_{n+1},\ldots,f_{2n-1}\right\}+
                    \left\{f_{n},\left\{ f_{1},\ldots,f_{n-1},f_{n+1} \right\},f_{n+2},\ldots,f_{2n-1}\right\}
                    \\
                    +\ldots+
                    \left\{f_{n},\ldots,f_{2n-2},\left\{ f_{1},\ldots,f_{n-1},f_{2n-1} \right\}\right\}
                    =
                    \left\{f_{1},\ldots,f_{n-1},
                    \left\{ f_{n},\ldots,f_{2n-1} \right\}\right\},
                \end{gathered}
                \label{eq:fundid}
            \end{equation}
            for all $f_{1},\ldots,f_{2n-1}\in A$.
    \end{itemize}
    \label{def:nambu}
\end{definition}

\noindent The Nambu--Poisson bracket was introduced in the case $n=3$ by Nambu in
1973~\cite{Nambu1973} while he was looking for a way to consider ternary interactions with the recently discovered quarks in mind. A systematic
study of the Nambu--Poisson brackets was later made by
Takhtajan~\cite{Takhtajan1994}. For our purposes we will use the following result, from~\cite[\S 8.3]{LaurentGengoux2018}, adapted to our notations:

\begin{prop}
    Given a smooth manifold $M$ of dimension $n$ and its
    ring of smooth functions $A=\mathcal{C}^{\infty}(M)$, then
    the bracket:
    \begin{equation}
        \displaystyle\left\{ f_{1},\ldots, f_{n}\right\}=
        \mu\det\frac{\partial (f_{1},\ldots, f_{n})}{\partial(x_{1},\ldots,x_{n})},
        \label{eq:nambupg}
    \end{equation}
    where $x_{1},\ldots,x_{n}$ are local coordinates, and $\mu\in A$ is fixed,
    is a Nambu--Poisson bracket of order $n$.
    \label{prop:nambu}
\end{prop}

\noindent It is clear that a Nambu--Poisson bracket of order $n$ defines a chain of lower-order Nambu--Poisson brackets obtained fixing up to $n-2$ functions
$f_{1}$, \ldots $f_{n-2}$ in $A$. In particular, the bracket obtained by fixing $n-2$ functions is a ``regular'' Poisson bracket, see~\cite{Takhtajan1994}.  Moreover, the dynamics of a Nambu--Poisson system is defined in local coordinates by fixing $n-1$
\emph{Nambu--Hamiltonian functions} $H_{1}$, \ldots $H_{n-1}$ in $A$ and then taking:
\begin{equation}
    \dot{x}_{i} = \left\{ x_{i},H_{1},\ldots,H_{n-1} \right\}.
    \label{eq:nambusys}
\end{equation}
We will refer to a system of this kind as a \emph{Nambu--Poisson system}.
Such a system naturally admits the $n-1$ Nambu--Hamiltonian functions as
first integrals. Such systems are particular instances of
\emph{skew gradient systems}, \ie systems that can be written as:
\begin{equation}
    \vec{\dot{x}} = S(\vec{x})\lrcorner dH_{1}\lrcorner \ldots \lrcorner 
    dH_{m},
    \label{eq:gradsys}
\end{equation}
where $S$ is a completely skew-symmetric $m+1$-tensor field, $1 \leq m \leq n-1$,  and we denoted by
$\lrcorner$ the tensor contraction.

To connect this kind of systems with the discrete setting, following the ideas
and the results of~\cite{QuispelCapel1996}, we need to introduce also the
definition below:

\begin{definition}[\cite{Gonzalez1996}]
    Given a smooth function $F\in \mathcal{C}^{\infty}(U)$,
    $U\subset\R^{n}$ open set, then a \emph{discrete gradient} for $F$ is
    any vector $\Delta F/\Delta \vec{x}\in \R^{n}$ such that the following
    identity hold:
    \begin{equation} 
        (\vec{x'}-\vec{x})\cdot \frac{\Delta F}{\Delta \vec{x}} 
        =
        F(\vec{x'})-F(\vec{x}),
        \quad
        \forall\vec{x},\vec{x'}\in U,
        \label{eq:discrgrad}
    \end{equation}
    called \emph{the discrete gradient identity}.
    \label{def:discrgrad}
\end{definition}

As a discrete gradient, one can take any vector of the following form:
\begin{equation}
    \frac{\Delta F}{\Delta \vec{x}}(\vec{x'},\vec{x})
    = \vec{a}(\vec{x'},\vec{x}) - 
    (\vec{x'}-\vec{x})\frac{(\vec{x'}-\vec{x})\cdot
    \vec{a}(\vec{x'},\vec{x})+F(\vec{x'})-F(\vec{x})}{(\vec{x'}-\vec{x})^{2}},
    \label{eq:DFDx}
\end{equation}
for a suitable choice of the vector $\vec{a}$, together with the consistency
condition $\Delta F/\Delta \vec{x} \to \grad F(\vec{x})$ as
$\vec{x'}\to\vec{x}$, \ie $\vec{a}(\vec{x'},\vec{x}) \to \grad F(\vec{x})$ as
$\vec{x'}\to\vec{x}$. For instance, a possible choice is
$\vec{a}(\vec{x'},\vec{x})\equiv \grad F(\vec{x})$. A choice that will be more
relevant to the discussion we will make in the following is:
\begin{equation}
    \frac{\Delta F}{\Delta \vec{x}}(\vec{x'},\vec{x})
    =
    \grad F\left( \frac{\vec{x'}+\vec{x}}{2} \right),
    \label{eq:dgradav}
\end{equation}
\ie the gradient of $F$ evaluated on the mid-point
$(\vec{x}+\vec{x'})/2$.

We recall the following result:

\begin{theorem}[\cite{QuispelCapel1996}]
    Given a gradient system~\eqref{eq:gradsys}, then for any choice of
    discrete gradient the discrete system:
    \begin{equation} 
        \frac{\vec{x'}-\vec{x}}{h} = 
        \widetilde{S}(\vec{x'},\vec{x})\lrcorner 
        \frac{\Delta H_{1}}{\Delta \vec{x}}(\vec{x'},\vec{x}) 
        \lrcorner \ldots \lrcorner 
        \frac{\Delta H_{m}}{\Delta \vec{x}}(\vec{x'},\vec{x}),
        \label{eq:dgradsys}
    \end{equation}
    where $\widetilde{S}(\vec{x'},\vec{x})$ is a completely skew-symmetric
    $m+1$-tensor field, such that $\widetilde{S}(\vec{x'},\vec{x}) \to S(\vec{x})$
    for $h\to0^{+}$, is an integrable discretization of the gradient
    system~\eqref{eq:gradsys} preserving the invariants $H_{1}$, \ldots,
    $H_{m}$.
    \label{thm:dqs}
\end{theorem}

The problem with such a discretization is that, in general, it might not
be rationally invertible. Indeed, the method is \emph{implicit}, \ie it
requires the use of some local inversion algorithm to set up the evolution,
for instance, one or more iterations of Newton's method. So, such a
discretization, albeit very interesting from a theoretical point of view,
it is not always suitable, especially if one is looking for explicit
methods.

\subsection{Discretization}

Let us consider the equations of the IHO in Hamiltonian form, \ie the Hamilton
equations of motion of~\eqref{eq:cont-Hamiltonian} with
$\alpha=\tilde{\omega}^2$:
\begin{subequations}
    \begin{align}
        \dot{q}_{i} &= p_{i},
        \label{eq:ihoq}
        \\
        \dot{p}_{i} &= -\tilde{\omega}^{2} q_{i}.
        \label{eq:ihop}
    \end{align}
    \label{eq:iho}
\end{subequations}

We ask ourselves the following problem:

\begin{problem}
    Find a linearly covariant \emph{maximally superintegrable
    discretization} of the Hamilton equations of the IHO~\eqref{eq:iho},
     preserving the standard symplectic structure of $\R^{2}$.
\end{problem}

\noindent We start building a linearly covariant discretization as
in~\cite{CelledoniMcLachlanOwrenQuispel2013} with the following rules:
\begin{description}
    \item[Rule 1:] substitute the derivative with the discrete derivative 
        operator:
        \begin{equation}
            \frac{\ud}{\ud t} \longrightarrow \frac{1}{h}\left( T_{h}-\Id \right),
            \quad
            T_{h} x(t) = x(t+h).
            \label{eq:dder}
        \end{equation}
    \item[Rule 2:] substitute any field $x$ with a convex combination:
        \begin{equation}
            x \longrightarrow c x + (1-c) x'.
            \label{eq:convex}
        \end{equation}
\end{description}
As above, $h>0$ is a discretization parameter. We observe that, in
general, one can assume $c=c(h)$. With these two rules, we arrive at the following discretization of the
system~\eqref{eq:iho}:
\begin{subequations}
    \begin{align}
        \frac{q_{i}'-q_{i}}{h} &= a_{i} p_{i} + (1-a_{i})p_{i}',
        \label{eq:dihoq}
        \\
        \frac{p_{i}'-p_{i}}{h} &= -\tilde{\omega}^{2} 
        \left[b_{i} q_{i}'+(1-b_{i})q_{i}\right].
        \label{eq:dihop}
    \end{align}
    \label{eq:diho}
\end{subequations}
We can write the time evolution explicitly as:
\begin{subequations}
    \begin{align}
        q_{i}' &=
        \frac{(1-h^2 \tilde{\omega}^2(1-a_i)(1-b_i))q_i+h p_i}{\Delta_{i}},
        \label{eq:dihoqex}
        \\
        p_{i}' &=
        -\frac{h\tilde{\omega}^2(h a_i b_i  p_i+q_i)-p_i}{\Delta_{i}},
        \label{eq:dihopex}
    \end{align}
    \label{eq:dihoex}
\end{subequations}
where:
\begin{equation}
    \Delta_{i} \coloneqq 
    1
    + h^2\tilde{\omega}^2(1-a_i)b_i.
    \label{eq:Deltai}
\end{equation}
That is, the system of O$\Delta$Es~\eqref{eq:diho} can be put into map form as follows:
\begin{equation}
    \begin{tikzcd}[row sep=tiny]
        \vec{\Phi} \colon  \R^{2N} \arrow{r} & \R^{2N}
        \\
        (q_{1},p_{1},\ldots,q_{N},p_{N}) \arrow[mapsto]{r} &
        (q_{1}',p_{1}',\ldots,q_{N}',p_{N}').
    \end{tikzcd}
    \label{eq:phimapdiho}
\end{equation}

\noindent Then, we have the following:
\begin{lemma}
    The system of O$\Delta$Es preserves the canonical symplectic structure
    if and only if $b_{i} =a_{i}$ for all $i=1,\ldots,N$.
    \label{lem:symplectic}
\end{lemma}

\begin{proof}
    To prove this statement we compute the Poisson bracket of
    $q_{i}'$, $p_{i}'$. We have:
    \begin{align}
        \big\{ p_{i}',p_{j}' \big\} &=
        \left\{  
        \frac{h\tilde{\omega}^2(h a_i b_i  p_i+q_i)-p_i}{\Delta_{i}},
        \frac{h\tilde{\omega}^2(h a_j b_j  p_j+q_j)-p_j}{\Delta_{j}}
        \right\}
        \notag   \\ 
        &=
        \frac{h\tilde{\omega}^2}{\Delta_{i}\Delta_{j}}
        \left[  
        (h^2 \tilde{\omega}^2a_i b_i  -1)\left\{  
        p_i,
        q_j
        \right\}
        +
        (h^2 \tilde{\omega}^2a_j b_j  -1)\left\{ q_i,
        p_j
        \right\}
        \right]
        \notag  \\
        &=
        \frac{h\tilde{\omega}^2\delta_{i,j}}{\Delta_{i}\Delta_{j}}
        \left[  
           - (h^2 \tilde{\omega}^2a_i b_i  -1)
        +
        (h^2 \tilde{\omega}^2a_j b_j  -1)
        \right].
    	 \label{eq:qipqjp}
    \end{align}
    Because of the appearance of the Kronecker delta symbol, this quantity can be non-zero only if $i=j$, but in such a case the two terms in the square bracket cancel each other. 
     So, $  \big\{ p_{i}',p_{j}' \big\} =0 $ and similarly, we get $\{ q_{i}',q_{j}' \}=0$. These two conditions do not pose any restriction on the discretization. 

    On the other hand, let us
    compute:
    \begin{align}
    	  \left\{ q_{i}',p_{j}' \right\}& =
        \left\{ 
        \frac{(1-h^2\tilde{\omega}^2(1-a_i)(1-b_i))q_i+h p_i}{\Delta_{i}},
        -\frac{h\tilde{\omega}^2(ha_j b_j  p_j+q_j)-p_j}{\Delta_{j}}
        \right\}
        \notag  \\
        &=
        \frac{(1-h^2 \tilde{\omega}^2(1-a_i)(1-b_i))(1-{h^2}\tilde{\omega}^2 a_j b_j )
        +h^{2}\tilde{\omega}^{2}}{\Delta_{i}\Delta_{j}}\delta_{i,j}.
          \label{eq:qippjp} 
    \end{align}
     So, this last condition is not satisfied only if $i=j$, in which case
    becomes:
    \begin{equation}\label{eq:cdab} 
        \frac{1+h^{2}\tilde{\omega}^{2}a_{i}(1-b_{i})}{1+h^{2}\tilde{\omega}^{2}b_{i}(1-a_{i})}=1.
    \end{equation}
    That is, we have $b_{i}=a_{i}$, and this concludes the proof of the lemma.
\end{proof}

\noindent At this point, following \Cref{lem:symplectic}, we consider the discretization:
\begin{subequations}
    \begin{align}
        \frac{q_{i}'-q_{i}}{h} &= a_{i} p_{i} + (1-a_{i})p_{i}',
        \label{eq:dihoq2}
        \\
        \frac{p_{i}'-p_{i}}{h} &= -\tilde{\omega}^{2} 
        \left[a_{i} q_{i}'+(1-a_{i})q_{i}\right],
        \label{eq:dihop2}
    \end{align}
    \label{eq:diho2}%
\end{subequations}
or, explicitly:
\begin{subequations}
    \begin{align}
        q_{i}' &=
        \frac{(1-h^2\tilde{\omega}^2(1-a_i)^2)q_i+h p_i}{\Delta_{i}},
        \label{eq:dihoqex2}
        \\
        p_{i}' &=
        -\frac{h\tilde{\omega}^2(h a_i^2  p_i+q_i)-p_i}{\Delta_{i}},
        \label{eq:dihopex2}
    \end{align}
    \label{eq:dihoex2}%
\end{subequations}
where:
\begin{equation}
    \Delta_{i} \coloneqq  
    1
    + h^2\tilde{\omega}^2(1-a_i)a_i.
    \label{eq:Deltai2}
\end{equation}

\noindent In general, such a discretization is integrable but not superintegrable, as
highlighted in the following lemma:

\begin{lemma}
    For all values of the constants $a_{i}$ the discretization of the
    IHO~\eqref{eq:diho2} is integrable, but not maximally superintegrable,
    admitting the following $N$ commuting invariants:
    \begin{equation}
        a_{i}^{(h,a_{i})}
        = 
        p_{i}^{2} +h\tilde{\omega}^{2} (2a_{i}-1)p_{i}q_{i} +
        \tilde{\omega}^{2}q_{i}^{2}.
        \label{eq:Fiigen}
    \end{equation}
    \label{prop:generic}
\end{lemma}

\begin{proof}
    The invariants~\eqref{eq:Fiigen} can be found through a direct analysis
    of the system with a quadratic polynomial ansatz. The fact that, in
    general, the system is not maximally superintegrable follows from an
    analysis of the solutions. Solving the system \eqref{eq:diho2} with
    respect to $(p_{i},p_{i}')$, and then using the two as compatibility
    conditions, we can rewrite it as a system of $N$ second-order
    O$\Delta$Es:
    \begin{equation}
        q_i(t+h)-\frac{2-h^2\tilde{\omega}^2(2a_i^2-2 a_i +1)}{1 + h^2\tilde{\omega}^2(1-a_i)a_i}q_i(t)
        +q_i(t-h)=0.
        \label{eq:diho2nd}
    \end{equation}
    This immediately yields the solution:
    \begin{equation}
        q_{i}(t) = q_{i,0}\sin\left(\omega_{i} t +\varphi_i \right),
        \label{eq:qisol}
    \end{equation}
    where $q_{i,0}$ and $\varphi_i$ are arbitrary constants, and:
    \begin{equation}
        \omega_{i} = \tilde{\omega} \sqrt{1-h^{2}\tilde{\omega}^{2}\left(
        a_{i}-\frac{1}{2} \right)^{2}}.
        \label{eq:Omegadef}
    \end{equation} 
    So, for generic values of the parameters $a_{i}$ the trajectories
    described by the solutions will not be closed, indicating that the system will not
    be maximally superintegrable.
\end{proof}

\begin{remark}
    We remark that:
    \begin{equation}
        \lim_{h\to 0^+} a_{i}^{(h,{a_i})} = 2H_{i},
        \label{eq:limitsdigamma}
    \end{equation}
    where $H_{i}$ is the single particle Hamiltonian of a harmonic oscillator.
    \label{rem:singlepart}
\end{remark}

Now, we observe that we can find a maximally superintegrable subcase by imposing
additional properties, as underlined by the following result:

\begin{lemma}
    The components of the angular momentum 
    $L_{i,j}$ are conserved by
    the discretization of the IHO~\eqref{eq:diho2} if and only if:
            \begin{equation}
                a_{i} = a_{j} \eqqcolon a,
                \label{eq:alphaij}
            \end{equation}
    and in such a case also the quantities:
    \begin{equation}
        F_{i,j}^{(h,a)} = p_{i}p_{j} 
        +\frac{2a-1}{2}h\tilde{\omega}^{2}(q_{i}p_{j}+q_{j}p_{i})
        +\tilde{\omega}^{2} q_{i}q_{j},
        \label{eq:Fijt}
    \end{equation}
    are invariant.
    \label{prop:fint}
\end{lemma}

\begin{proof}
    The proof of this statement follows again from a direct computation
    \emph{imposing} the conservation of the components of the angular
    momentum. Indeed, we have:
    \begin{align}
      L_{i,j}' =& q_{i}'p_{j}'-q_{j}'p_{i}' 
          \notag  \\
            =&
            \frac{h^3\tilde{\omega}^4(a_i+a_j-2)(a_i-a_j)}{\Delta_{i}\Delta_{j}}q_iq_j
            +\frac{h^3 \tilde{\omega}^2(a_i-a_j)(a_i+a_j)}{\Delta_{i}\Delta_{j}}p_ip_j       
           \notag  \\
            &+\frac{1-h^2 \varomega^2 (a_i^2+a_j^2-2a_i)+h^4 \varomega^4 a_j^2(a_i-1)^2 }{\Delta_{i}\Delta_{j}} q_i p_j  
           \notag   \\
            &-
            \frac{1-h^2 \varomega^2 (a_i^2+a_j^2-2a_j)+h^4 \varomega^4
            a_i^2(a_j-1)^2}{\Delta_{i}\Delta_{j}}
            q_jp_i,
    	 \label{eq:Lijprime}
    \end{align}
    which by assumption must be equal to $L_{i,j}$.  Taking
        the coefficients of $q_{i}q_{j}$ and $p_{i}p_{j}$ we see that the only
        compatible solution in the space of parameters is $a_{i}=a_{j}$.  Since
        this implies $\Delta_i=\Delta_j$, we have a factorization and
        $L_{i,j}'=L_{i,j}$ becomes an identity. Moreover, since the condition
        $a_{i}=a_{j}$ must hold for any pair of indices, we have that all parameters
        $a_{i}$ must be equal to a single parameter $a$.

    To conclude the proof, we note that the invariants~\eqref{eq:Fijt} can
    be found through a direct computation using a quadratic
    polynomial ansatz in the variables $(q_{i},p_{i})$ and
    $(q_{j},p_{j})$.
\end{proof}

Now, observe that:
\begin{equation}
    \lim_{h\to 0^+} F_{i,j}^{(h,a)} = F_{i,j},
    \label{eq:limits}
\end{equation}
\ie the quadratic polynomials~\eqref{eq:Fijt} form a \emph{discrete analogue of
the Demkov--Fradkin tensor}. For this reason, we call the two-index
object $F_{i,j}^{(h,a)}$ \emph{the discrete Demkov--Fradkin tensor}.

\begin{remark}   \label{rem:diag}
    We observe that, as expected, $a_{i}^{(h,a)}=F_{i,i}^{(h,a)}$,
    being the diagonal elements of the Demkov--Fradkin tensor the single
    particle Hamiltonians.
\end{remark}
So, our maximally superintegrable discretization of the IHO is the
following:
\begin{subequations}
    \begin{align}
        \frac{q_{i}'-q_{i}}{h} &= a p_{i} + (1-a)p_{i}',
        \label{eq:dihoq3}
        \\
        \frac{p_{i}'-p_{i}}{h} &= -\tilde{\omega}^{2} 
        \left[a q_{i}'+(1-a)q_{i}\right].
        \label{eq:dihop3}
    \end{align}
    \label{eq:diho3}%
\end{subequations}
or explicitly:
\begin{subequations}
    \begin{align}
        q_{i}' &=
        \frac{(1-h^2\tilde{\omega}^2(1-a)^2)q_i+h p_i}{\Delta},
        \label{eq:dihoqex3}
        \\
        p_{i}' &=
        -\frac{h\tilde{\omega}^2(ha^2  p_i+q_i)-p_i}{\Delta},
        \label{eq:dihopex3}
    \end{align}
    \label{eq:dihoex3}%
\end{subequations}
where:
\begin{equation}
    \Delta \coloneqq  
    1
    + h^2\tilde{\omega}^2(1-a)a.
    \label{eq:Deltai3}
\end{equation}
We close this Subsection by computing the symmetry algebra generated
by the invariants of the system~\eqref{eq:diho3}:
\begin{prop}
    The invariants $\set{L_{i,j},F_{i,j}^{(h,a)}}$ form a Lie algebra
    under Poisson bracket with commutation relations:
    \begin{subequations}
        \begin{align}
            \bigl\{L_{i,j},L_{k,l}  \bigr\} &=
            L_{j,l}\delta_{i,k}+L_{k,j}\delta_{l,i}
            +L_{l,i}\delta_{j,k}+L_{i,k}\delta_{l,j},
            \label{eq:LijLlmdiho}
            \\
            \bigl\{L_{i,j},F_{k,l}^{(h,a)}  \bigr\} &=
            F_{j,l}^{(h,a)}\delta_{i,k}
            + F_{k,j}^{(h,a)}\delta_{i,l}
            - F_{i,l}^{(h,a)}\delta_{j,k} 
            - F_{i,k}^{(h,a)}\delta_{j,l},
            \label{eq:LijFLlmdiho}
            \\
            \bigl\{F_{i,j}^{(h,a)},F_{k,l}^{(h,a)}  \bigr\} &=
            \kappa(h,a)\left( 
                L_{j,l}\delta_{i,k}+L_{j,k}\delta_{i,l}
                +L_{i,l}\delta_{j,k}+L_{i,k}\delta_{j,l}
            \right),%
            \label{eq:FijFlmdiho}
        \end{align}
        \label{eq:commrel}%
    \end{subequations}
    where:
    \begin{equation}
        \kappa(h,a) \coloneqq  
        \tilde{\omega}^2\left[1-h^2\tilde{\omega}^2\left(a-\frac{1}{2}\right)^2\right].
        \label{eq:OmegaDef}
    \end{equation}
    That is, they form a Lie algebra $\mathfrak{A}_{N}(\kappa(h,a))$.
    \label{prop:algdiho}
\end{prop}
\begin{proof}
    The commutation relations~\eqref{eq:LijLlmdiho} are the usual ones.
    Here we give a sketch of the proof of the
    relations~\eqref{eq:FijFlmdiho}, while~\eqref{eq:LijFLlmdiho} can be
    proved in an analogous way.

    To show it, let us consider the following derivatives:
    \begin{subequations}
        \begin{align}
            \frac{\partial F_{i,j}^{(h,a)}}{\partial q_{s}} &=
            \frac{2a-1}{2}h\tilde{\omega}^{2}(\delta_{i,s}p_{j}+\delta_{j,s}p_{i})
            +\tilde{\omega}^{2} (\delta_{i,s}q_{j}+q_{i}\delta_{j,s}),
            \label{eq:Fijqk}
            \\
            \frac{\partial F_{i,j}^{(h,a)}}{\partial p_{s}} &=
            \delta_{i,s}p_{j} + p_{i}\delta_{j,s} 
            +\frac{2a-1}{2}h\tilde{\omega}^{2}(q_{i}\delta_{j,s}+q_{j}\delta_{i,s}).
            \label{eq:Fijpk}
        \end{align}
        \label{eq:Fijder}
    \end{subequations}
    So, we have:
    \begin{equation}\label{eq:discr-pb-by-def}
    	\bigl\{F_{i,j}^{(h,a)},F_{k,l}^{(h,a)}  \bigr\} 
    	= \sum_s %\left( 
    	 \frac{\partial F_{i,j}^{(h,a)}}{\partial q_{s}}
    	 \frac{\partial F_{k,l}^{(h,a)}}{\partial p_{s}}
    	 - 
    	 \frac{\partial F_{k,l}^{(h,a)}}{\partial q_{s}}
    	  \frac{\partial F_{i,j}^{(h,a)}}{\partial p_{s}}.
    \end{equation}
    Combining \eqref{eq:discr-pb-by-def} with \eqref{eq:Fijder} and performing some calculations, we obtain: 
    \begin{align}
    	\bigl\{F_{i,j}^{(h,a)},F_{k,l}^{(h,a)}  \bigr\} 
    	= - \tilde{\omega}^2 \left(1 - h^2 \tilde{\omega}^2 \left( a - \frac{1}{2} \right)^2    \right) \sum_s 
    	\bigl[ &
    	\left(
    	(q_{l} p_{j}-q_{j} p_{l}) \delta_{k,s} 
    	   +\delta_{l,s} 
    	   (q_{k} p_{j} -q_{j} p_{k} ) \right) \delta_{i,s}  
    	\notag
    	\\
    	 &-\delta_{j,s} 
    	   \bigl( ( q_{i} p_{l} - q_{l} p_{i} ) \delta_{k,s} 
    	   +\delta_{l,s}  
    	   (q_{i} p_{k}-q_{k} p_{i}  ) \bigr) 
    	\bigr].
    \end{align} 
    Therefore, using \eqref{eq:OmegaDef}, \eqref{eq:discr-angular-momentum-in-stmplectic-realization}
    and taking a sum over $s$, we arrive at:
    \begin{align}
    	\bigl\{F_{i,j}^{(h,a)},F_{k,l}^{(h,a)}  \bigr\} 
    	=  \kappa (h,a) \left( 
    	-\delta_{i,k} L_{l,j} 
    	- \delta_{i,l} L_{k,j} 
    +\delta_{j,k} L_{i,l}
    	+  \delta_{j,l} L_{i,k}
    	\right).  
    \end{align}
    Using then the antisymmetry of $L_{i,j}$, we obtain 
    \eqref{eq:FijFlmdiho}.     
\end{proof}

Now, the crucial observation is that the quantity $\kappa(h,a)$ defined
in~\eqref{eq:OmegaDef} can be \emph{positive, negative, or zero}. This
underlines a significant difference between the continuous and the discrete
setting. In particular, it shows that one can check whether or not a
discretization can be thought to be a ``good one''~\cite{HietarintaBook}
purely from the algebraic point of view: the isomorphism class
of the algebra will show whether or not the properties of the continuous MS
system are preserved.

\subsection{Comparison with other discretizations}

In this Subsection, we compare the discretization we obtained with
different known discretizations. We address a quite popular discretization
rule, the \emph{KHK
discretization}~\cite{Kahan1993,KahanLi1997,HirotaKimura2000}, which over
the years has obtained great attention from the discrete integrable system
community~\cite{PetreraSuris2010,PetreraPfadlerSuris2009,
    PetreraPfadlerSuris2011,PetreraZander2017,PSS2019,GubNahm,GJ_biquadratic,
GMcLQ_dummies,CelledoniMcLachlanOwrenQuispel2013,
CelledoniMcLachlanOwrenQuispel2014,CelledoniMcLachlanMcLarenOwrenQuispel2017,
VanDerKampCelledoniMcLachlanMcLarenOwrenQuispel2019}.  For a system of
first-order ODEs~\eqref{eq:xdotF} the KHK discretization rule is the
following:
\begin{equation}
    \frac{\vec{x'}-\vec{x}}{h} =
    - \frac{1}{2} \vec{f}(\vec{x'})
    +2\vec{f}\left(\frac{\vec{x}+\vec{x'}}{2}\right)
    - \frac{1}{2}\vec{f}(\vec{x}).
    \label{eq:khkgen}
\end{equation}
In fact, the KHK method is a member of the class of \emph{Runge--Kutta (RK)
methods}, see~\cite{CelledoniMcLachlanOwrenQuispel2013}:
\begin{equation}
    \frac{\vec{x'}-\vec{x}}{h} =
    \lambda\vec{f}(\vec{x'})+(1-2\lambda)\vec{f}\left(\frac{\vec{x}+\vec{x'}}{2}\right) 
    + \lambda\vec{f}(\vec{x}),
    \label{eq:rkgen}
\end{equation}
for an arbitrary $\lambda\in\R$. It is trivial to see that if the function
$\vec{f}$ is linear, \ie there exists a matrix $A$ such that
$\vec{f}(\vec{x})=A\vec{x}$, the KHK method and the RK method collapse to the
same method, namely:
\begin{equation}
    \frac{\vec{x'}-\vec{x}}{h} =
    \frac{A}{2}(\vec{x}+\vec{x'}),
    \label{eq:ellmeth}
\end{equation}
which is the energy conservation scheme~\cite{SanzCalvo2018}. 

In our case of interest, in the Hamiltonian equations~\eqref{eq:iho} the
matrix $A$ can be written as a block diagonal matrix as:
\begin{equation}
    A = 
    \begin{pmatrix}
        A_{\tilde{\omega}} & & & &0
        \\
        & A_{\tilde{\omega}}
        \\
        & & \ddots
        \\
        0 & & & & A_{\tilde{\omega}}
    \end{pmatrix},
    \qquad
    A_{\tilde{\omega}} = 
    \begin{pmatrix}
        0 & 1
        \\
        -\tilde{\omega}^{2} & 0
    \end{pmatrix}.
    \label{eq:Amat}
\end{equation}
So, writing $\vec{x}=(q_{1},p_{1},\ldots,q_{N},p_{N})$ we obtain the
discretization:
\begin{subequations}
    \begin{align}
        \frac{q_{i}'-q_{i}}{h} &= \frac{1}{2}( p_{i} + p_{i}'),
        \label{eq:dihoq4}
        \\
        \frac{p_{i}'-p_{i}}{h} &= -\frac{\tilde{\omega}^{2}}{2}
        (q_{i}'+q_{i}),
        \label{eq:dihop4}
    \end{align}
    \label{eq:diho4}%
\end{subequations}
which clearly coincides with our final discretization~\eqref{eq:dihop3} for
$a=1/2$. That is, we have that the three methods KHK, RK, and energy
conservation scheme corresponds to our chosen discretization for the specific
value $1/2$ of the free parameter $a$. Moreover, from \Cref{prop:fint} we have
that these three classes of discretization \emph{preserve exactly the canonical
Poisson bracket and the Demkov--Fradkin tensor itself}, since
$F_{i,j}^{(h,1/2)}=F_{i,j}$. That is, we have that, this system is so simple that it can preserve both the energy (Hamiltonian) and the
Poisson bracket at the same time. At first, this might seem to contradict the
Ge--Marsden theorem~\cite{ZhongMarsden1988}.  However, it is well known from
the literature that linear and linearisable discrete systems are special, in the sense that they don't share all properties of integrable ones,
see~\cite{Takenawaatel2003,HayHoweseNakazonoShi2015}. So, in this light it is
not completely surprising that this apparent contradiction arises. 

Now, we show how the KHK/RK discretizations are linked with the discretization
procedure presented in~\Cref{thm:dqs}. To this end, let us first prove that the
IHO is in fact a Nambu--Hamiltonian system. The first step to prove this is  the following technical result:

\begin{lemma}\label{lem:fipart}
    The functions $\mathcal{S}=\Set{F_{1,1},\dots,F_{1,N},F_{2,2},\dots,F_{2,N}}$
    are functionally independent.
\end{lemma}

\begin{proof}
    By direct computation, we have, for $ \dfrac{\partial(F_{1,1},\dots,F_{1,N},F_{2,2},\dots,F_{2,N})}{\partial(q_1,\dots,q_N,p_1,\dots,p_N)}$ :
    \begin{align}\label{eq:discretization-jacobian}
    	\begin{pNiceMatrix}
		2 \varomega^{2} q_{1} & 2 p_{1} 
		& 0 & 0 & \Cdots[shorten=15pt] & & & & & 0 & 0 &  0  & 0
		\\ 
		\varomega^{2} q_{2} & p_2 & 
\Block[borders={bottom,left}]{1-2}{}  
 \varomega^2 q_1 & p_1 &   & \Cdots[shorten=15pt] & &  & &0  & 0  & 0 & 0 
		\\ 
		\varomega^{2} q_{3} & p_3 & 0 & 0 & \Block[borders={bottom,left}]{1-2}{}  \varomega^2 q_1 & p_1 &   &  & & 0 & 0 & 0 & 0 		\\ 
		\Vdots & \Vdots &\Vdots    & \Vdots   &   & & & & & & & \Vdots & \Vdots 
		\\
		\varomega^2 q_{N-1} & p_{N-1} 
		& 0 & 0 & 0 & 0 & & & &   \Block[borders={bottom,left}]{1-2}{}   \varomega^2 q_1 & p_1 & 0 & 0 
		\\
		\varomega^2 q_N & p_N 
		& 0 & 0 & 0 &  0 & \Cdots  &   &  & 0 & 0 &   \Block[borders={left, bottom}]{1-2}{}    \varomega^2 q_1 & p_1
        \\ %SECOND PART
        	0 & 0
		& 2\varomega^2 q_2  & 2p_2 & 0 & 0 & \Cdots & & & 0 & 0 &  0  & 0
		\\ 
		0 & 0 &  \varomega^2 q_3 & p_3 & \Block[borders={bottom,left}]{1-2}{}  \varomega^2 q_2 & p_2 &   &  & & 0 & 0 & 0 & 0 		\\ 
		\Vdots & \Vdots &\Vdots    & \Vdots   &   & & & & & & & \Vdots & \Vdots 
		\\
		0 & 0 & \varomega^2 q_{N-1}  & p_{N-1}  & 0 & 0 & & & &   \Block[borders={bottom,left}]{1-2}{}   \varomega^2 q_2 & p_2 & 0 & 0 
		\\
		0 & 0 & \varomega^2 q_N & p_N & 0 &    0 & \Cdots &   &  & 0 & 0 &   \Block[borders={left}]{1-2}{}    \varomega^2 q_2 & p_2
		 \CodeAfter  % added <<<<<<<<<<<<
        \line[shorten=10pt]{3-6}{5-10}
         \CodeAfter  % added <<<<<<<<<<<<
        \line[shorten=10pt]{8-6}{10-10}
	\end{pNiceMatrix}
    \end{align}
    One can estimate the lower bound of the rank of~\eqref{eq:discretization-jacobian}
    by evaluating it on the point $\vec{q^*}=(0,1,0,\dots,0)$ and
    $\vec{p^*}=(1,0,\dots,0)$. Then, the echelon form of~\eqref{eq:discretization-jacobian} is: 
    \begin{align}\label{eq:discretization-jac-echelon}
        \begin{pmatrix}
                1 & 0 & 0 & 0 & \cdots & 0 
                \\ 
                0 & 1 & \tilde{\omega} & 0 & \cdots & 0 
                \\ 
                0 & 0 & 0 & 1 & \cdots & 0 
                \\ 
                \vdots  & \vdots  & \vdots  & \vdots  & \ddots  & \vdots 
                \\
                0 & 0 & 0 & 0 & \cdots & 1 
        \end{pmatrix}. 
    \end{align}
    The rank of~\eqref{eq:discretization-jac-echelon} is clearly maximal.
    Therefore, the same is true for the
    Jacobian~\eqref{eq:discretization-jacobian} given the continuity of rank
    function, except possibly for a set of points of measure zero. 
\end{proof}

In the following proposition, we prove that the IHO is actually a
Nambu--Hamiltonian with the first integrals in the set $\mathcal{S}$ with
respect to the standard Nambu--Poisson bracket~\eqref{eq:nambupg}, using an
appropriate $\mu$ factor depending on $N$:

\begin{prop}\label{th:Nambu-mechanics-f1-fN}
    For all $N>1$ the IHO admits a Nambu--Hamiltonian formulation with respect
    to the Nambu bracket:
    \begin{equation}
        \Set{f_{1},\dots,f_{N}} = 
        \mu_{2N}(\vec{x})\det \left( \grad f_{1},\dots,\grad f_{N} \right),
        \label{eq:nambum}
    \end{equation}
    with:
     \begin{equation}
        \mu_{2N}(\vec{x})
        =
        \frac{1}{4 \varomega^{2(N-1)} L_{1,2}^{N-1}},
        \label{eq:muN}
    \end{equation}
    and the functions of the set $\mathcal{S}$ as Nambu--Hamiltonian. 
    \label{prop:nambuiho}
\end{prop}
\begin{proof}
    We will prove that the following Nambu--Hamiltonian equations hold:
  \begin{subequations}
        \begin{align}
            \dot{q}_{i} &=
            \Set{q_{i}, F_{1,1},\dots,F_{1,N},F_{2,2},\dots,F_{2,N} },
            \label{eq:dqinambu}
            \\
            \dot{p}_{i} &=
            \Set{p_{i}, F_{1,1},\dots,F_{1,N},F_{2,2},\dots,F_{2,N}
            }.
            \label{eq:dpinambu}%
        \end{align}
        \label{eq:dqidpinambu}%
    \end{subequations}
    That is, more explicitly we will prove that:
 \begin{subequations}
        \begin{align}
            \dot{q}_{i} =
            \mu_{2N}(\vec{x})\det\frac{\partial(q_i,F_{1,1},\dots,F_{1,N},F_{2,2},\dots,F_{2,N})}{\partial(q_1,\dots,q_N,p_1,\dots,p_N)},
            \label{eq:dqinambu2}
            \\
            \dot{p}_{i} =
            \mu_{2N}(\vec{x})\det 
            \frac{\partial(q_i,F_{1,1},\dots,F_{1,N},F_{2,2},\dots,F_{2,N})}{\partial(q_1,\dots,q_N,p_1,\dots,p_N)}.
            \label{eq:dpinambu2}%
        \end{align}
        \label{eq:dqidpinambu2}%
    \end{subequations}
    We observe that for technical reasons we reordered the dynamical variables.
    Note that, \Cref{lem:fipart} allows us to use the family of invariants
    $\mathcal{S}$.

    We can write 
\begin{equation}
	\dot q_i = \mu_{2N} (\vec x) \cdot \det  J_i, 
	\qquad 
	1 \leq i\leq N,
\end{equation}
where 
\begin{equation}\label{eq:Nambu-first-matrix} 
J_i = \begin{pNiceMatrix}
		0 & 0 & 0 & \Cdots & 1 & \Cdots & 0  & 0 & 0  & 0 & \Cdots & 0 & \Cdots  & 0 
		\\
		2 \varomega^2 q_1 & 0 & 0 & \Cdots & 0 & \Cdots &  0 & 2 p_1 & 0 & 0 & \Cdots & 0 &  \Cdots &  0 
		\\ 
		\varomega^2 q_2 & \varomega^2 q_1 & 0 & \Cdots & 0 & \Cdots & 0 & p_2  & p_1 & 0 & \Cdots & 0 & \Cdots  & 0  
		\\
		\varomega^2 q_3 & 0 &  \varomega^2 q_1 & \Cdots & 0 & \Cdots & 0 & p_3  & 0 & p_1 & \Cdots & 0 & \Cdots  & 0  
		\\
		\Vdots & \Vdots & \Vdots & \Ddots & \Vdots & & \Vdots & \Vdots & \Vdots & \Vdots & \Ddots & & &   \Vdots 
		\\
		\varomega^2 q_i & 0 & 0 & \Cdots & 
		\varomega^2 q_1 & \Cdots & 0 & p_i & 0 & 0 & \Cdots & p_1 & \Cdots & 0    
		\\
		\Vdots & \Vdots  & \Vdots & & \Vdots & \Ddots & \Vdots & \Vdots & \Vdots & \Vdots & &   & %\Ddots[shorten=30pt]
		\Ddots  & \Vdots 
		\\
		\varomega^2 q_N & 0  & 0 & \Cdots & 0 & \Cdots & \varomega^2 q_1 & p_N & 0 & 0 & \Cdots &  0 & \Cdots & p_1    
		\\ 
		0 & 2 \varomega^2 q_2 & 0 & \Cdots & 0 & \Cdots & 0 & 0 & 2 p_2 & 0 & \Cdots &  0 & \Cdots   &  0
		\\ 
		0 & \varomega^2 q_3 & \varomega^2 q_2 &\Cdots & 0 & \Cdots & 0 & 0 & p_3 & p_2 & \Cdots & 0 &\Cdots & 0 
		\\ 
		\Vdots &  \Vdots & \Vdots & \Ddots & \Vdots & & \Vdots  & \Vdots & \Vdots  & \Vdots
		& \Ddots & \Vdots & & \Vdots   
		\\ 
		0 & \varomega^2 q_i  & 0  &  & \varomega^2 q_2 & &  0 & 0  & p_i & 0 &  \Cdots & p_2 & \Cdots & 0 
		\\
		\Vdots & \Vdots & \Vdots &  & \Vdots & \Ddots & \Vdots & \Vdots & \Vdots & \Vdots & &\Vdots  & \Ddots & \Vdots 
		\\
		0 & \varomega^2 q_N &  0 & \Cdots & 0 &\Cdots & \varomega^2 q_2 & 0 & p_N & 0
		& \Cdots & 0 & \Cdots & p_2 
	\end{pNiceMatrix} 
\end{equation}
Let us move the first row of~\eqref{eq:Nambu-first-matrix} to the middle of the
matrix using a sequence of $N$ row exchanges. Therefore, $\det J_{i} = (-1)^N
\det \tilde{J}_{i} $, where 
\begin{equation}
	\tilde{J}_i = \begin{pNiceMatrix}
	\Block[borders={bottom,right}]{7-7}{} 
		2 \varomega^2 q_1 & 0 & 0 & \Cdots & 0 & \Cdots &  0 
		& 2 p_1 \Block[borders={bottom}]{7-7}{}  & 0 & 0 & \Cdots & 0 &  \Cdots &  0 
		\\ 
		\varomega^2 q_2 & \varomega^2 q_1 & 0 & \Cdots & 0 & \Cdots & 0 & p_2  & p_1 & 0 & \Cdots & 0 & \Cdots  & 0  
		\\
		\varomega^2 q_3 & 0 &  \varomega^2 q_1 & \Cdots & 0 & \Cdots & 0 & p_3  & 0 & p_1 & \Cdots & 0 & \Cdots  & 0  
		\\
		\Vdots & \Vdots & \Vdots & \Ddots & \Vdots & & \Vdots & \Vdots & \Vdots & \Vdots & \Ddots & & &   \Vdots 
		\\
		\varomega^2 q_i & 0 & 0 & \Cdots & 
		\varomega^2 q_1 & \Cdots & 0 & p_i & 0 & 0 & \Cdots & p_1 & \Cdots & 0    
		\\
		\Vdots & \Vdots  & \Vdots & & \Vdots & \Ddots & \Vdots & \Vdots & \Vdots & \Vdots & &   & %\Ddots[shorten=30pt]
		\Ddots  & \Vdots 
		\\
		\varomega^2 q_N & 0  & 0 & \Cdots & 0 & \Cdots & \varomega^2 q_1 & p_N & 0 & 0 & \Cdots &  0 & \Cdots & p_1   
		\\
			0 \Block[borders={right}]{7-7}{}  & 0 & 0 & \Cdots & 1 & \Cdots & 0  & 0 & 0  & 0 & \Cdots & 0 & \Cdots  & 0  
		\\ 
		0 & 2 \varomega^2 q_2 & 0 & \Cdots & 0 & \Cdots & 0 & 0 & 2 p_2 & 0 & \Cdots &  0 & \Cdots   &  0
		\\ 
		0 & \varomega^2 q_3 & \varomega^2 q_2 &\Cdots & 0 & \Cdots & 0 & 0 & p_3 & p_2 & \Cdots & 0 &\Cdots & 0 
		\\ 
		\Vdots &  \Vdots & \Vdots & \Ddots & \Vdots & & \Vdots  & \Vdots & \Vdots  & \Vdots
		& \Ddots & \Vdots & & \Vdots   
		\\ 
		0 & \varomega^2 q_i  & 0  &  & \varomega^2 q_2 & &  0 & 0  & p_i & 0 &  \Cdots & p_2 & \Cdots & 0 
		\\
		\Vdots & \Vdots & \Vdots &  & \Vdots & \Ddots & \Vdots & \Vdots & \Vdots & \Vdots & &\Vdots  & \Ddots & \Vdots 
		\\
		0 & \varomega^2 q_N &  0 & \Cdots & 0 &\Cdots & \varomega^2 q_2 & 0 & p_N & 0
		& \Cdots & 0 & \Cdots & p_2 
	\end{pNiceMatrix}. 
\end{equation}
The latter determinant can be calculated using the general formula for block
matrices~\eqref{eq:block-matrices-formula-general-app}, and adopting a notation
of \Cref{sec:appendix-matrices}, in particular \eqref{app:eq:A-B-def},
\eqref{app:eq-breveBi-def}, we have: 
\begin{subequations}\label{eq:Nambu-ABCD}
\begin{align}
	A & =
	\mathcal{A} \left(2 \varomega^2 q_1,  \varomega^2 q_1, 
	 \varomega^2  [	q_2,  \cdots,  q_{N}]
\right)
	 =  \varomega^2  \mathcal{A} \left(2q_1, q_1, 
	%\begin{bmatrix}
	 [	q_2,  \cdots,  q_{N}]
%	\end{bmatrix} 
\right),
	\\
	B & = \mathcal{A} \left( 2 p_1, p_1, 
		 [p_2, \cdots, p_N]
	\right),
	\\
	C & =\breve{\mathcal{B}}_i \left( 2 \varomega^2 q_2, \varomega^2 q_2, \varomega^2  [q_3, \cdots, q_{N}] \right), 
	\\ 
	D & = \mathcal{B} \left( 2 p_2, p_2, [p_3,\cdots, p_{N}] \right).  
\end{align}	
\end{subequations} 
Hence:
\begin{multline}
		\det J_i = (-1)^N \det \left[ \mathcal{A} \left(2 \varomega^2 q_1,  \varomega^2 q_1, 
	 \varomega^2 [	q_2,  \cdots,  q_{N}]
\right) \right] \cdot   \det \bigl[ 
	 \mathcal{B} \left( 2 p_2, p_2, [p_3,\cdots, p_{N}] \right) 
	 \\
	 - 
	 \breve{\mathcal{B}}_i \left( 2 \varomega^2 q_2, \varomega^2 q_2, \varomega^2   [q_3, \cdots, q_{N}] \right) 
	 \mathcal{A}^{-1} \left(2 \varomega^2 q_1,  \varomega^2 q_1, 
	 \varomega^2 [	q_2,  \cdots,  q_{N}]
\right) \mathcal{A} \left( 2 p_1, p_1, 
		 [p_2, \cdots, p_N]
	\right)
	 \bigr].  
\end{multline}
Using \Cref{app:lemma-dets-summarized} and  \Cref{app:lemm:final-det}, one
arrive at the following result:
\begin{equation}
    \det J_i = -4 {(-1)}^N  \cdot 
    \varomega^{2N} q_1^{N} \cdot {(-1)}^N 
        \frac{L_{1,2}^{N-1} }{\varomega^2 q_1^N} p_i
         = - 4 \varomega^{2(N-1)} L_{1,2}^{N-1} p_i.  
    \label{eq:detJi}
\end{equation}
For $\dot p_i $, one can set instead in \eqref{eq:Nambu-ABCD}:
\begin{subequations}
\begin{align}
	C & =\mathcal{B} \left( 2 \varomega^2 q_2, \varomega^2 q_2, \varomega^2   [q_3, \cdots, q_{N}] \right), 
	\\
	D & = \breve{\mathcal{B}}_{N+i} \left( 2 p_2, p_2, [p_3,\cdots, p_{N}] \right),
\end{align}	
\end{subequations} 
and proceed with the same argument, \mm. 
\end{proof}

\begin{prop}
    The discretization of the IHO~\eqref{eq:diho3} is equivalent to the discrete
    gradient system:
    \begin{subequations}
        \begin{align}
            \frac{q_{i}'-q_{i}}{h} = \widetilde{\mu}_{N}^{(h,a)} \det
            \left(\frac{\Delta q_{i}}{\Delta \vec{x}}, \frac{\Delta
                F_{1,1}^{(h,a)}}{\Delta \vec{x}}, \ldots, \frac{\Delta
                F_{1,N}^{(h,a)}}{\Delta \vec{x}}, \frac{\Delta F_{2,2}^{(h,a)}}{\Delta
                \vec{x}}, \ldots,\frac{\Delta F_{2,N}^{(h,a)}}{\Delta \vec{x}}
            \right),
            \label{eq:delqinambu2}
            \\
            \frac{p_{i}'-p_{i}}{h} = \widetilde{\mu}_{N}^{(h,a)} \det
            \left(\frac{\Delta p_{i}}{\Delta \vec{x}}, \frac{\Delta
                F_{1,1}^{(h,a)}}{\Delta \vec{x}}, \ldots, \frac{\Delta
                F_{1,N}^{(h,a)}}{\Delta \vec{x}}, \frac{\Delta F_{2,2}^{(h,a)}}{\Delta
                \vec{x}}, \ldots,\frac{\Delta F_{2,N}^{(h,a)}}{\Delta \vec{x}}
            \right),
            \label{eq:delpinambu2}
        \end{align}%
        \label{eq:delqidelpinambu2}%
    \end{subequations}
    where the discrete gradient is defined in equation~\eqref{eq:dgradav}, and
    $\widetilde{\mu}_{2N}=\mu_{2N}((\vec{x}+\vec{x'})/2)$ with $\mu_{2N}$
    defined by equation~\eqref{eq:muN}. 
    \label{prop:dnambuiho2}
\end{prop}

\begin{proof}
    The argument is identical to the one made for the proof of
    \Cref{th:Nambu-mechanics-f1-fN}, involving the statements of
    \Cref{sec:appendix-matrices}. Indeed, let us note that the Jacobian matrices
    involved can be obtained from the proof of \Cref{th:Nambu-mechanics-f1-fN}
    by substituting $\vec{x}$ with $(\vec{x}+\vec{x'})/2$.
\end{proof}

In addition to that, we observe that the discrete gradient system construction
applies also for the invariants $\mathcal{S}^{(h,a)}\coloneqq\set{F_{1,1}^{(h,a)},F_{1,2}^{(h,a)},\dots,F_{1,N}^{(h,a)},F_{2,2}^{(h,a)},\dots,F_{2,N}^{(h,a)}}$. However, using the
discrete gradient~\eqref{eq:dgradav}, or a more general version with the choice:
\begin{equation}
    \vec{a}(\vec{x},\vec{x'})=\grad F\left( a\vec{q'}+(1-a)\vec{q},a\vec{p}+(1-a)\vec{p'} \right), 
    \label{eq:dgradexteded}
\end{equation}
is not equivalent to the system we derived from our discretization rules.
Nevertheless, the system is again a discrete MS gradient system, but it is
nonlinear and the discretization is implicit. We omit the explicit formulas
since they are too cumbersome to show. We will discuss in the conclusions how
similar considerations suggest several possible generalizations to other MS
systems.

\section{Properties and isomorphism for the Lie algebra $\Frad_{N}$}
\label{sec:Frad-plus}

In this Section we consider the $N^{2}$-dimensional Lie algebra $\Frad_N=
\FradPlain_N (\omega^2)$. We observe that in such a case the commutation
relations~\eqref{eq:pb-Fradkin} take the form: 
\begin{subequations}\label{eq:pb-Fradkin-plus}
    \begin{align}
        \bigl[L_{i,j},L_{k,l}  \bigr] &=
        L_{j,l}\delta_{i,k}+L_{k,j}\delta_{l,i}
        +L_{l,i}\delta_{j,k}+L_{i,k}\delta_{l,j},
        \label{eq:LijLlm-Fradkin-plus}
        \\
        \bigl[L_{i,j},F_{k,l}  \bigr] &=
        F_{j,l}\delta_{i,k}+ F_{k,j}\delta_{i,l}
        -F_{i,l}\delta_{j,k} 
        -F_{i,k}\delta_{j,l},
        \label{eq:LiFLlm-Fradkin-plus}
        \\
        \bigl[F_{i,j},F_{k,l}  \bigr] &=
        \omega^2 \left(  
            L_{j,l}\delta_{i,k}+L_{j,k}\delta_{i,l}
            +L_{i,l}\delta_{j,k}+L_{i,k}\delta_{j,l}
        \right). 
        \label{eq:FijFlm-Fradkin-plus}
    \end{align} 
\end{subequations}
In particular, we present its Levi decomposition in an appropriate basis, its
classification using the invariance properties of its Killing form of the
semisimple Levi factor. Finally, we present its explicit isomorphism with
$\mathfrak{u}_N$, thus proving the point~\ref{item:case0-intro} of the Main
Theorem.

\subsection{Change of the basis and Levi decomposition of $\Frad_N$}

Let us now introduce a set of elements of $\Frad_N$ by:
\begin{equation}\label{eq:prop:more-general-transformation}
    f_{i,j} \coloneqq   \F_{i,j} + \omega L_{i,j}
    -    \delta_{i,j} \F_{N,N}, 
    \quad
    1\leq i,j\leq N,
\end{equation}
or more explicitly:
\begin{subequations}\label{eq:rotation}
	\begin{align}
            f_{i,j} & = \F_{i,j} + \omega L_{i,j}, \quad i\neq j, 
	\label{eq:rotation-non-diag}
	\\
	f_{j,j} &= \F_{j,j} - \F_{N,N},
	\label{eq:rotation-diag}
	\end{align}
\end{subequations}
where we used that $L_{i,i}\equiv0$. Noting that $f_{N,N}=F_{N,N}-F_{N,N}=0$,
we have that there are $N^{2}-1$ linearly independent elements $f_{i,j}$.
In particular, we see that these $N^2-1$ elements form a Lie subalgebra
of $\Frad_N$, as underlined by the following lemma:

\begin{lemma}\label{lem:fradN-semisimple}
    The vector space: 
    \begin{equation}\label{eq:fradN-pm-basis}
		\frad_N \coloneqq 
		 \Span \bigl( 
	 \{ f_{i,i}  \}_{i=1}^{N-1} 
	 \cup  
	 \{ f_{i,j}  \}_{ {i,j=1},    {i \ne j}  }^{N} 
	\bigr), 
    \end{equation}
    is closed under the Lie bracket of $\Frad_N$, \ie it is a Lie subalgebra
    with commutation relations:
    \begin{align}
		 	[f_{i,j} , f_{k,l} ]  =  
		\omega \,
		\big(
	 \delta_{i,k}   	 
	 ( f_{j,l}  - f_{l,j}  )
	 +\delta_{j,l}   
	 ( f_{i,k}  - f_{k,i}  )
	 +\delta_{i,l}    
	 ( f_{j,k}  + f_{k,j}  )
	 -\delta_{j,k}    
	 ( f_{i,l}  + f_{l,i}  & ) 
	 \notag \\ 
	 \label{eq:frad-N-closed-Lie-bracket} 
	 - 2    \delta_{k,l} 
	 ( f_{j,N}  \delta_{i,N} - f_{N,i}  \delta_{j,N} )
	 + 2    \delta_{i,j} 
     ( f_{l,N}  \delta_{k,N} - f_{N,k}  \delta_{l,N} & ) \big).
    \end{align}
    In particular, it is a semisimple $N^2-1$-dimensional Lie algebra.
\end{lemma}

\begin{proof}
    The Lie bracket of $\frad_N$ can be written in terms of the $\Frad_N $
    generators, using \eqref{eq:prop:more-general-transformation}:
    \begin{align}
            [f_{i,j}, f_{k,l}]_{\frad_N} &=
            \left[\F_{i,j} + \omega L_{i,j}- \delta_{i,j} F_{N,N} , 
            \F_{k,l} +\omega L_{k,l} - \delta_{k,l} F_{N,N} \right]_{\Frad_N}
            \notag
            \\ 
            &= \left[ \F_{i,j}, \F_{k,l} \right]_{\Frad_N}
            +\omega  \left[\F_{i,j}, L_{k,l} \right]_{\Frad_N} 
            -\delta_{k,l} \left[ F_{i,j}, F_{N,N} \right]_{\Frad_N}
             + \omega \left[ L_{i,j}, \F_{k,l} \right]_{\Frad_N}
             \notag \\ 
            &\phantom{=} +\omega^2   \left[ L_{i,j}, L_{k,l}\right]_{\Frad_N}
            - \omega \delta_{k,l}  \left[ L_{i,j}, F_{N,N} \right]_{\Frad_N}
            -\delta_{i,j}  \left[ F_{N,N}, F_{k,l} \right]_{\Frad_N}
            -\omega \delta_{i,j}  \left[ F_{N,N}, L_{k,l} \right]_{\Frad_N}.
    \end{align} 
    Performing direct computations using~\eqref{eq:pb-Fradkin}
    and~\eqref{eq:ideal}, one arrives at the
    expression~\eqref{eq:frad-N-closed-Lie-bracket}, showing that the
    commutation relation of $\frad_N$ are closed with respect to the inherited
    Lie bracket. Therefore, $\frad_N$ is a Lie subalgebra of $\Frad_N $. 

    To prove semisimplicity, let us now calculate the Killing form and show it
    is non-degenerate. We observe that we can
    rewrite~\eqref{eq:frad-N-closed-Lie-bracket} in the form: 
    \begin{equation}
            [f_{i,j}, f_{k,l}] = \sum_{\alpha,\beta=1}^N 
            c^{(\alpha \beta) }_{(ij)(kl)} f_{\alpha,\beta},
    \end{equation} 
    with structure constants:
\begin{align}\label{eq:structure-constant-in-fp}
	c^{(\alpha\beta)}_{(ij)(kl)} = \omega 	\, & \bigl[   
	%1 
	\delta_{i,k}   
	\left( \delta\indices{^\alpha _j} \delta\indices{^\beta _l}
	- \delta\indices{^\alpha _l}
	\delta\indices{^\beta _j}
	 \right) 
	 %2 
	 +\delta_{j,l}   
	 \left(\delta\indices{^\alpha _i}
	 \delta\indices{^\beta _k}
	 -\delta\indices{^\alpha _k}
	 \delta\indices{^\beta _i}
	 \right)
	 \notag \\
	 %3 
	 &+\delta_{i,l}  
	 \left(
	 \delta\indices{^\alpha _j}
	 \delta\indices{^\beta _k}
	 +\delta\indices{^\alpha _k}
	 \delta\indices{^\beta _j}
	 \right)
	 -\delta_{j,k}  
	 \left( 
	 \delta\indices{^\alpha _i}
	 \delta\indices{^\beta _l}
	 +\delta\indices{^\alpha _l}
	 \delta\indices{^\beta _i}
	 \right)
	 \notag \\ 
	 &- 2 \delta_{k,l}  
	 \left( 
	 \delta_{i,N}
	 \delta\indices{^\alpha _j}
	 \delta\indices{^\beta _N}
	 -\delta_{j,N}
	 \delta\indices{^\alpha _N}
	 \delta\indices{^\beta _i}
	 \right)
	 \notag \\ 
	 & +2 \delta_{i,j}   
	 \left( 
	 \delta_{k,N} 
	 \delta\indices{^\alpha _l} 
	 \delta\indices{^\beta _N}
	 -\delta_{l,N}
	 \delta\indices{^\alpha _N}
	 \delta\indices{^\beta _k}
	 \right) \bigr], 
    \end{align}
    where we put some indices on top to enhance readability,
    \ie$\delta\indices{^\alpha _\beta } \equiv \delta_{\alpha,\beta} $.  

    A matrix element of the Killing form can be calculated using the
    structure constants of the Lie algebra:
    \begin{equation}
            K_{(ij) (kl)} = \sum_{\alpha, \beta, \gamma, \delta =1}^N 
            c^{(\gamma\delta)}_{(ij)(\alpha \beta)}
            c^{(\alpha \beta)}_{(kl)(\gamma \delta)}. 
    \end{equation}
    After some calculations with the use
    of~\eqref{eq:structure-constant-in-fp}, one arrives at the general
    expression for the matrix element of the Killing form for $\frad_N$ for
    arbitrary $N$: 
    \begin{equation}\label{eq:Killing-form-generalN-matrix-el}
            K_{(ij) (kl)} = 8 N \omega^2 \bigl( 
            (\delta_{k,N} \delta_{l,N} - \delta_{k,l})
            \delta_{i,j} 
            + \delta_{i,N}\delta_{j,N}\delta_{k,l}
            -\delta_{i,k}\delta_{j,l}
            \bigr). 
    \end{equation}
    It is easy to see that~\eqref{eq:Killing-form-generalN-matrix-el} gives the
    non-zero entries only if either $i=j$ and $k=l$ or $i=k$ and $j=l$. Let us
    denote the corresponding matrix elements as follows:  
    \begin{subequations}\label{eq:block-matrix-elts}
            \begin{align}
                    \kappa_{i,j} &\coloneqq K_{(ii)(jj)}= 
                    8N\omega^2 \left(\delta_{N,i}^2 +\delta_{N,j}^2 - \delta_{i,j}^2 -1      \right) = - 8N\omega^2 (  \delta_{i,j} +1),
                    \label{eq:block-matrix-elts-kappa}
                    \\ 
                    d_{i,j} & \coloneqq K_{(ij)(ij)} = -8N \omega^2, 
                    \label{eq:block-matrix-elts-d}
            \end{align}
    \end{subequations}
    where in~\eqref{eq:block-matrix-elts-kappa} we took into account that there
    is no element with indices $(N,N)$ and that $\delta_{i,j}^2 = \delta_{i,j}$
    for any $i,j $.  Here $\kappa \coloneqq (\kappa_{i,j})$ is a square
    $(N-1)$-dimensional matrix, and $d \coloneqq (	d_{i,j} ) $ is a
    diagonal $N(N-1)$ matrix. The Killing form itself can be written as:
    \begin{equation}\label{eq:Killing-form-generalN-dual-basis}
            K= \sum_{i,j=1}^{N-1} \phi_{i,i}  \phi_{j,j}  
            + \sum_{i,j=1 \atop i\neq j}^{N}  \phi_{i,j}  \phi_{i,j} , 
    \end{equation}
    where $ \Span  \{ \phi_{i,j}  \}_{1\leq i,j \leq N} $ is the dual basis of
    ${(\frad_N)}^*$,  \ie $ \phi_{i,j}   (f_{i,j}  )=1$. Given the ordering 
    \begin{equation}
            \phi_{1,1} , \dots, \phi_{(N-1),(N-1)} ,\dots, \phi_{1,2} ,\dots, \phi_{(N-1),N} , 
    \end{equation}
    the matrix associated to \eqref{eq:Killing-form-generalN-dual-basis} can be
    written in a block form:
    \begin{equation}\label{eq:Killing-form-matrix}
        \hat{K} =    
                \left[\begin{array}{  c | c   }
        \kappa  & \mathbb{0}   \\ \hline 
         \mathbb{0}   & d 
       \end{array}\right]_{N^2-1}
       =-8 N \omega^2  \left[\begin{array}{  c | c   }
        \hat{\kappa} & \mathbb{0}   \\ \hline 
         \mathbb{0}   &  \mathbb{1}_{N(N-1)}   
       \end{array}\right], 
    \end{equation}
    where the first block  has the following shape:
    \begin{equation}\label{eq:hat-kappa-def}
        \hat{\kappa} \coloneqq %(  \delta_{i,j} +1)   =  
        \begin{bmatrix}
                    2 
                     &  1  
                       &  \dots 
                       & 1  
                     \\ 
                        1  
                        &
                          2  
                           &
                          \dots   
                          & 
                          1 
                         \\ 
                          \vdots 
                        &
                          \vdots   
                           &
                         \ddots   
                         & \vdots 
                         \\
                         1     &     1  &\dots & 2 
      \end{bmatrix}_{(N-1)\times (N-1)} 
    \end{equation}
    The determinant of $\hat {\kappa}$  can be calculated
    using~\cite[Lemma~3.2]{DrozGub_h6} for $\mu=2 $, $\nu=1$ and $n=N-1$. Thus 
    \begin{equation}\label{eq:hatkappa-determinant-compitation}
            \det \hat \kappa = (\mu - \nu)^{n-1} \left( \mu + (n-1) \nu  \right)
            = 2+ N-2 = N. 
    \end{equation}
    The determinant of $\hat{K} $, in turn, could be computed using the property of
    the determinants of the block matrices: 
    \begin{align}
        \det ( \hat{K} ) &=
            (-8N \omega^2)^{N^2 -1} \cdot 
             \det \hat \kappa \cdot \det \mathbb{1}_{N(N-1)}  
             \notag
             \\ 
             & = (-8N \omega^2)^{N^2 -1} \cdot  N \cdot 1
             \notag 
             \\ & = 
            (-8 \omega^2)^{N^2 -1} \cdot N^{N^2}. 
            \label{eq:Killing-determinant}
    \end{align}
    From \eqref{eq:Killing-determinant}, one can clearly see that for $\omega \ne 0 $, the Killing form \eqref{eq:Killing-form-generalN-matrix-el} is non-degenerate. Therefore, by Cartan's criterion, the Lie algebra $\frad_N $ is semisimple. 
\end{proof}

\begin{remark}
 We remark that another natural choice of the basis of $\frad_N$ can be the set: 
 \begin{equation}\label{eq:ladder-operators}
		 f_{i,j}^\pm \coloneqq   \F_{i,j} \pm \omega L_{i,j}
    -    \delta_{i,j} \F_{N,N}, 
    \quad
    1\leq i \leq  j\leq N. 
   % \nonumber 
	\end{equation}
 In the notation of \eqref{eq:prop:more-general-transformation}, $f_{i,j}^+ \equiv f_{i,j} $. The set $\{f_{i,j}^+, f_{i,j}^- \}_{1 \leq i \leq j \leq N}$  indeed forms a basis of $\frad_N $ with the additional condition for the indices $i\leq j $. However, from \eqref{eq:ladder-operators} it follows that  $ f_{i,j}^- = f_{j,i}^+$ (rotational symmetry), therefore instead of a ``mixed'' basis we can generate the Lie algebra $\frad_N$ using the set $\{ f_{i,j}^+  \}_{1 \leq i,j \leq N} $, allowing now the indices to run over full range. 
 \end{remark}

We can now decompose the Lie algebra $\Frad_N$ into a direct sum: 
\begin{theorem}\label{th:LeviDecomp}
    The (trivial) Levi decomposition of $\Frad_N$ is given by 
    \begin{equation}\label{eq:th-levi-direct-sum}
            \Frad_N= \frad_N  \oplus   \R, 
    \end{equation}  
    where $\frad_N$ is semisimple, and the $1$-dimensional radical of $\Frad_N$ is spanned by $r_N \coloneqq  \F_{1,1} + \ldots + \F_{N,N} $.
\end{theorem}

\begin{proof}
    The Lie algebra $\frad_N $ is a semisimple Lie subalgebra of $\Frad_N$ by
    \Cref{lem:fradN-semisimple} of dimension $N^{2}-1$. We complete to a basis
    of $\Frad_N$ adding the element $r_{N}$. Then, observe that $r_N$ commutes
    with every element of $\Frad_N$, because for any $i,j$: 
    \begin{align}
            [f_{i,j}, r_N]_{\Frad_N} &= \sum_{s=1}^N 
            \left( \left[\F_{i,j}, \F_{s,s} \right]_{\Frad_N} 
            + \omega  \left[L_{i,j}, \F_{s,s} \right]_{\Frad_N} 
            +\delta_{i,j}  \left[\F_{s,s}, \F_{N,N} \right]_{\Frad_N} 
            \right)
            \notag \\
            &= 2 \omega^2 (L_{j,i} +L_{i,j})
            + 2 \omega (F_{j,i} - F_{i,j}) 
            + 4 \omega^2 \delta_{i,j} L_{N,N}
            \notag \\ 
            &=0,
    \end{align}
    where we used \eqref{eq:pb-Fradkin} and \eqref{eq:ideal}. Being
    $1$-dimensional, $\Span(r_N)$  is not only a solvable but also an abelian
    ideal of $\Frad_N$, so it coincides with the radical of $\Frad_N $. Summing
    up, we have $[\frad_N,\frad_N]\subset \frad_N$, $[\frad_N,r_N] =0$, and
    $[r_N,r_N]=0$, that is the Lie algebra  $\Frad_N$ can be decomposed into
    the direct sum~\eqref{eq:th-levi-direct-sum}. This proves the statement. 
\end{proof}

\subsection{Properties of the Lie algebra $\frad_N$  }

In this Subsection, we determine the isomorphism class of the semisimple Lie
algebra $\frad_N$. Let us start with the following proposition:

\begin{prop}\label{th:eigenvalues-Killing}
    The eigenvalues of the matrix of the Killing
    form~\eqref{eq:Killing-form-matrix} are given by 
    \begin{equation}
        \lambda_1 = -8N^2 \omega^2,
        \qquad
        \lambda_2 = \ldots=\lambda_{N^2-1}= - 8 N \omega^2.
        \label{eq:lambdasu}
    \end{equation}
\end{prop}

\begin{proof}
    The set of eigenvalues of the block matrix $\hat{K} $
    \eqref{eq:Killing-form-matrix} consists of the sets of eigenvalues of the
    matrices $\kappa $ and $d$. The matrix $d$ is diagonal, hence its
    eigenvalues are: 
    \begin{equation}
            \lambda_N \coloneqq \dots = \lambda_{N^2-1}
            = - 8 N \omega^2. 
    \end{equation}
    The matrix $\kappa$ can be written as: 
    \begin{equation}\label{eq:mat-kappa-aux}
            \kappa = - 8 N \omega^2 \hat \kappa, 
    \end{equation}
    where $\hat \kappa $ is defined by \eqref{eq:hat-kappa-def}. 
    Then, the characteristic equation for $\hat \kappa$ is:
    \begin{equation}
            \det \left(\hat\kappa - \hat{\lambda} \mathbb{1}\right) = 
            \det \begin{bmatrix}
            2 - \hat\lambda 
             &  1  
               &  \cdots 
               & 1  
             \\ 
                1  
                &
                  2 - \hat\lambda  
                   &
                  \cdots   
                  & 
                  1 
                 \\ 
                  \vdots 
                &
                  \vdots   
                   &
                 \ddots   
                 & \vdots 
                 \\
                 1     &     1  &\cdots & 2 - \hat\lambda 
\end{bmatrix}_{(N-1)\times (N-1)} = 0 .
    \end{equation}
    This determinant can be calculated in the same way as in
    \eqref{eq:hatkappa-determinant-compitation}: 
    \begin{equation}
            \det (\hat\kappa - \hat{\lambda} \mathbb{1}) = 
            (1- \hat\lambda )^{N-2} 
            (N - \hat\lambda  )
            = 0. 
    \end{equation}
    Thus, we obtain 
    \begin{equation}
            \hat\lambda_1 \coloneqq N,
            \qquad
            \hat\lambda_2 \coloneqq \dots = \hat\lambda_{N-1} =1.
        \label{eq:lambdared}
    \end{equation}
    Scaling back by the factor $-8N\omega^2$, the eigenvalues of the matrix 
    $\kappa$ \eqref{eq:mat-kappa-aux} are: 
    \begin{equation}
        \lambda_1 \coloneqq -8N^2\omega^2 , 
        \qquad
        \lambda_2 \coloneqq \dots = \lambda_{N-1} =-8N \omega^2,
        \label{eq:lambdaunred}
    \end{equation}
    respectively. This proves the statement. 
\end{proof}

\begin{corollary}
    The semisimple Lie algebra $\frad_N$ is  compact.
    \label{cor:compact}
\end{corollary}
\begin{proof}
    Given $\omega^2 >0 $, the statement follows from the definiteness of
    the Killing form, using \Cref{th:eigenvalues-Killing}. 
\end{proof}

We now recall the following fact about the real forms of complex semisimple Lie
algebras, which is a consequence of Sylvester's theorem:

\begin{prop}[{\cite[\S 4.1]{SnoblWinternitz2017book}}] \label{prop:sign}
    The signature of the Killing form 
    \begin{equation}
            \#   \coloneqq  (p,n,z),
    \end{equation}
    where $p,n,z$ are the numbers of its positive, negative, and zero
    eigenvalues respectively, is an invariant of a real form of a
    semisimple complex Lie algebra.
\end{prop}

So, since $\frad_N$ is a real Lie algebra, from \Cref{prop:sign}
and \Cref{cor:compact}, we have that $\frad_N$ is a
$N^2-1$ dimensional semisimple real compact form. Hence, it is in the
isomorphism class of the compact form of $\sl_{N}(\mathbb{C})$, \ie $\su_N$. In
the next Subsection we prove this by showing an explicit isomorphism.

\subsection{Explicit isomorphism $\frad_N \cong \su_N $}

We concluded the previous Subsection by noting that $\frad_N$ is in the
isomorphism class of $\su_N$. Before presenting the general result, we show how
to prove it in the simplest case, \ie $N=2$, in the following example:

\begin{example}[$\frad_2$]\label{example:frad2-su2}
    Let us first consider a simple example of $\Frad_2 $. This algebra is
    4-dimensional, generated by $f_{1,1} , f_{1,2} , f_{2,1}$ and
    $r_2$. We have $[r_2, \hyphen] =0 $, hence: 
    \begin{equation}
            \Frad_2 = \frad_2 \oplus \Span(r_2).
    \end{equation}
    Following \eqref{eq:frad-N-closed-Lie-bracket}, we may write
    down the commutation  relations  for $ \frad_2$ (see~\Cref{table:frad2-comm-table}).  

\begin{table}[h!]
	\[
		  \begin{array}{c||c|c|c}
          %  \toprule
           \frad_2 &  f_{1,1}  & f_{1,2}  & f_{2,1} 
            \\
             \midrule \midrule 
           f_{1,1}  & 0 & - 4 \omega  f_{2,1}  &  4 \omega f_{1,2} 
            \\    \midrule 
           f_{1,2}  & 4 \omega f_{2,1}   & 0 & -2 \omega f_{1,1} 
            \\   \midrule 
            f_{2,1}  & - 4 \omega f_{1,2}   & 2 \omega f_{1,1}  &  0 
            \end{array}
        \]
\caption{$ \frad_2$ commutation table for $\omega >0$. }	
	\label{table:frad2-comm-table} 
\end{table}

    This algebra is isomorphic to $\su_2 =  \Span ( t_1, t_2, t_3 )$ such that $ [t_j, t_k]= 2\sum_l \epsilon_{jkl} t_l $,
    where $\epsilon_{jkl} $ is the Levi-Civita symbol.   Hence, the commutation
    relations may be written as in~\cref{table:su2-comm-table}.   
    
\begin{table}[h!]
		\[
		  \begin{array}{c||c|c|c}
          %  \toprule
           \su_2 &  t_1 & t_2 & t_3
            \\
            \midrule  \midrule 
         t_1 & 0 & 2t_3 &   -2t_2
            \\    \midrule 
          t_2 & -2t_3  & 0 & 2t_1
            \\   \midrule 
           t_3 & 2t_2  & -2 t_1 &  0 
            \end{array}
        \]
	\caption{$\su_2$ commutation table}
	\label{table:su2-comm-table}
\end{table}
The elements might be represented as $2\times 2$-matrices:  
\begin{equation}
	 	\Mat (t_1) =- \imath \sigma_3 = 
	\begin{pmatrix}
		- \imath & 0
		\\ 
		0 &  \imath  
	\end{pmatrix},
	\quad
	\Mat (t_2 ) = - \imath \sigma_2 = 
	\begin{pmatrix}
		0 & -1 
		\\ 
		1 & 0 
	\end{pmatrix},
	\quad 
	\Mat (t_3) = \imath \sigma_1  = 
	\begin{pmatrix}
		0 &  \imath 
		\\ 
		 \imath & 0 
	\end{pmatrix},   
\end{equation}
where $\sigma_1, \sigma_2 $ and $\sigma_3 $ are Pauli matrices. 
The chosen ordering of the basis is such that the Cartan element appears first. 
From Tables~\ref{table:frad2-comm-table}~and~\ref{table:su2-comm-table}, we
see that these Lie algebras are isomorphic. The explicit isomorphism is given by:
 
\begin{equation}
	\rho_2: \  \frad_2  \to \su_2
\end{equation} 
such that: 
\begin{equation}
	\rho_2(f_{1,1} )= 2 \omega t_1,
	\qquad
	\rho_2 (f_{1,2} )=-\sqrt{2}\omega t_2, 
	\qquad
	\rho_2 (f_{2,1} )=  \sqrt{2}\omega t_3. 
\end{equation}	
Therefore, an arbitrary element $X \in \frad_N $ can be represented as the following $2\times 2$-matrix: 
\begin{equation}\label{eq:frad-N=2example-arb-elem}
	\Mat (X) = \Mat ( a_{1,1} f_{1,1} 
	+ a_{1,2} f_{1,2} 
	+ a_{2,1} f_{2,1} 
	) 
	= \omega \begin{pmatrix}
		- 2 \imath a_{1,1} & \sqrt{2} ( a_{1,2} + \imath a_{2,1} )
		\\ 
		-\sqrt{2} (  a_{1,2} - \imath a_{2,1} ) & 2 \imath a_{1,1}
	\end{pmatrix},
\end{equation}
where $ a_{1,1}, a_{1,2} $ and $a_{2,1} $ are real coefficients. 

\end{example}
Then, this is the main result of this Section:
\begin{theorem}\label{th:iso-suN}
    The Lie algebra isomorphism $ \frad_N  \cong \su_N $ holds. 
\end{theorem}

\begin{proof}
    We already determined that the isomorphism class of $\frad_N$ is $\su_N$,
    see \Cref{cor:compact}.  Let us now construct the isomorphism explicitly. 
	
    First, let us introduce the basis for the matrix representation of $\su_N$. It can be given as the following set of $N \times N$ matrices:	 
	 \begin{equation}\label{eq:set-suN-basis}
	 	 \mathcal{H} \cup 
	 	 \mathcal{M}^{2} \cup \dots \cup 
	 	 \mathcal{M}^{N}   
	 	 \cup \mathscr{M}^2 \cup \dots \cup  \mathscr{M}^N . 
	 \end{equation}
	 Here $\mathcal{H}$ denotes the set of diagonal matrices and generates the Cartan subalgebra of $\su_N$. We may assume the shape of them being:
	 \begin{equation}\label{eq:cartan-subset}
	 	\mathcal{H} \coloneqq \left\{
	 	\begin{pmatrix}
  	       -\imath 
  	         &  0  
  	           &  \cdots 
  	           & 0  
  	         \\ 
  	            0  
  	            &
  	              0  
  	               &
  	              \cdots   
  	              & 
  	              0 
  	             \\ 
  	              \vdots 
  	            &
  	              \vdots   
  	               &
  	             \ddots   
  	             & \vdots 
  	             \\
  	             0     &     0  & \cdots & 
  	             \imath 
  \end{pmatrix}, 
      \begin{pmatrix}
      	0 & 0 & \cdots & 0 \\ 
      	0 & -\imath & \cdots & 0 \\
      	\vdots & \vdots & \ddots & \vdots \\
      	0 & 0 & \cdots & \imath
      \end{pmatrix}, 
      	 	\cdots,  
      	 	\begin{pmatrix}
      	0 & \cdots  &  0 & 0 \\ 
          	\vdots & \ddots & \vdots   & \vdots \\
      	0 & \cdots  & -\imath & 0  \\ 
      	0 & \cdots  & 0 & \imath
      \end{pmatrix} 
	 	\right\} \eqqcolon \left\{ H_1, \dots, H_{N-1} \right\}. 
	 \end{equation}
The number of elements of the set $\mathcal{H}$ is $|\mathcal{H}| = N-1 $. 

Each of the sets $\mathcal{M}^{\alpha}$, $\alpha=2,\dots,N $ consists of the  block-diagonal matrices $N \times N $ of the form: 
\begin{equation}
	M\indices{^\alpha _k} = \diag( \mathbb{0}, B^\alpha , \mathbb{0} ),
	\qquad 
	k= 1, \dots, |  \mathcal{M}^{\alpha}|,
\end{equation}
 where we insert the $j\times j $ block:
\begin{equation}
	B^\alpha \coloneqq 
	\begin{pmatrix}
		0 & \cdots & -1 \\
		\vdots & \ddots & \vdots \\ 
		1 & \cdots & 0
	\end{pmatrix}, 
	\qquad
	\begin{aligned}
	  (B^\alpha)_{1,\alpha }&=-1,  
	  \\ 
	  (B^\alpha)_{\alpha ,1}&=1,  	
	\end{aligned}
\end{equation}
in different positions on the diagonal, \ie 
\begin{subequations}\label{eq:M-subsets}
\begin{align}
	\mathcal{M}^2 & \coloneqq \left\{  
	%1st 
	\begin{pmatrix}
		\color{red} 0 & \color{red} -1 & 0 & \cdots  & 0 \\
		\color{red} 1 &  \color{red}  0 & 0& \cdots & 0 \\
		0 & 0 & 0 & \cdots & 0 \\ 
		\vdots & \vdots & \vdots & \ddots & \vdots \\
		0 & 0 & 0 & \cdots & 0 
	\end{pmatrix}, 
	%2nd 
	\begin{pmatrix}
		0 & 0 & 0 & \cdots  & 0 \\
		0 & \color{red} 0 & \color{red} -1& \cdots & 0 \\
		0 & \color{red} 1 & \color{red} 0 & \cdots & 0 \\ 
		\vdots & \vdots & \vdots & \ddots & \vdots \\
		0 & 0 & 0 & \cdots & 0 
	\end{pmatrix},  
	\dots, 
	\begin{pmatrix}
		0 & \cdots    & 0 & 0    & 0 \\
		\vdots  & \ddots      & \vdots & \vdots  & \vdots  \\
		0 & \cdots    & 0 & 0 & 0   \\
		0 & \cdots   & 
		0  & \color{red} 0  & \color{red} -1  
		\\
		0 & \cdots    & 0 &\color{red} 1 & \color{red}  0 
	\end{pmatrix} 
	\right\}  
	= \left\{M\indices{^2 _1}, M\indices{^2 _2},
	\dots, M\indices{^2 _{N-1} }  \right\}; 
	\\
	\mathcal{M}^3 & \coloneqq \left\{ 
	%1st 
	\begin{pmatrix}
		\color{red} 0 & \color{red} 0 &\color{red} -1 &0  &\cdots  & 0 \\
		\color{red} 0 &  \color{red}  0 & \color{red} 0& 0& \cdots & 0 \\
		\color{red} 1 & \color{red} 0 & \color{red} 0& 0 & \cdots & 0 \\ 
		0 & 0 & 0 & 0 & \cdots & 0 \\
		\vdots & \vdots & \vdots & \vdots & \ddots & \vdots  \\
		0 & 0 & 0 & 0 & \cdots  & 0
	\end{pmatrix}, 
	%2nd 
	\begin{pmatrix}
		0 &  0 &0 &0  &\cdots  & 0 \\
		0 &  \color{red}  0 & \color{red} 0&\color{red} -1& \cdots & 0 \\
		0 & \color{red} 0 & \color{red} 0& \color{red} 0 & \cdots & 0 \\ 
		0 & \color{red} 1 & \color{red} 0 & \color{red} 0 & \cdots & 0 \\
		\vdots & \vdots & \vdots & \vdots & \ddots & \vdots  \\
		0 & 0 & 0 & 0 & \cdots  & 0
	\end{pmatrix},   
	\dots, 
	\begin{pmatrix}
		0 & \cdots    & 0 & 0    & 0 &0 \\
		\vdots  & \ddots      & \vdots & \vdots  & \vdots & \vdots   \\
		0 &\cdots  & 0 & 0 & 0 & 0
		\\
		0 & \cdots  & 0  &\color{red} 0 &\color{red} 0 & \color{red} -1   \\
		0 & \cdots &0  & \color{red}	0  & \color{red} 0  & \color{red} 0  
		\\
		0 & \cdots  & 0  & \color{red} 1 &\color{red} 0 & \color{red}  0 
	\end{pmatrix} 
	\right\}  
	= \left\{M\indices{^3 _1}, M\indices{^3 _2},
	\dots, M\indices{^3 _{N-2} }  \right\};
	\\ 
	& \qquad \makebox[\linewidth-3cm]
	{ \leavevmode\xleaders\hbox{\vdots \quad \quad }\hfill\kern0pt  }
	\notag 
	\\
	\mathcal{M}^N & \coloneqq 
	\left\{ 
	\begin{pmatrix}
		\color{red} 0 & \color{red} \cdots & \color{red} -1 
		\\ 
		\color{red} \vdots &\color{red} \ddots & \color{red} \vdots 
		\\ 
		\color{red} 1 & \color{red} \cdots & \color{red} 0 
	\end{pmatrix}
	\right\}  = \left\{M\indices{^N_1}  \right\}.
\end{align}
\end{subequations}
The set $\mathcal{M}^N $ contains only one element because there is only one way to fit $N \times N $ block into the matrix of the same dimension. In general, the number of elements of the  sets $M^\alpha$ can be calculated as 
\begin{equation}
	| \mathcal{M}^\alpha| = N+1 - \alpha, 
	\qquad 
	\alpha=2,\dots, N.
\end{equation}
In a similar fashion, we define a family of sets $\mathscr{M}^\beta  $, where the elements are given by 
\begin{equation}
	\mu\indices{^\beta _k} = \diag( \mathbb{0},  \widetilde{B}^\beta , \mathbb{0} ),
	\qquad 
	k= 1, \dots, |  \mathscr{M}^{\beta}|,
\end{equation}
but now with the block of the following shape: 
\begin{equation}
	\widetilde{B}^\beta \coloneqq 
	\begin{pmatrix}
		0 & \cdots & \imath \\
		\vdots & \ddots & \vdots \\ 
		\imath & \cdots & 0
	\end{pmatrix}, 
	\qquad
	  (\widetilde{B}^\beta)_{1,\beta} =(\widetilde{B}^\beta)_{\beta,1}=\imath ,  	
\end{equation}
hence:  
\begin{subequations}\label{eq:mu-subsets}
\begin{align}
	\mathscr{M}^2 & \coloneqq \left\{  
	%1st 
	\begin{pmatrix}
		\color{red} 0 & \color{red} \imath & 0 & \cdots  & 0 \\
		\color{red} \imath &  \color{red}  0 & 0& \cdots & 0 \\
		0 & 0 & 0 & \cdots & 0 \\ 
		\vdots & \vdots & \vdots & \ddots & \vdots \\
		0 & 0 & 0 & \cdots & 0 
	\end{pmatrix}, 
	%2nd 
	\begin{pmatrix}
		0 & 0 & 0 & \cdots  & 0 \\
		0 & \color{red} 0 & \color{red} \imath & \cdots & 0 \\
		0 & \color{red} \imath & \color{red} 0 & \cdots & 0 \\ 
		\vdots & \vdots & \vdots & \ddots & \vdots \\
		0 & 0 & 0 & \cdots & 0 
	\end{pmatrix},  
	\dots, 
	\begin{pmatrix}
		0 & \cdots    & 0 & 0    & 0 \\
		\vdots  & \ddots      & \vdots & \vdots  & \vdots  \\
		0 & \cdots    & 0 & 0 & 0   \\
		0 & \cdots   & 
		0  & \color{red} 0  & \color{red} \imath  
		\\
		0 & \cdots    & 0 &\color{red} \imath & \color{red}  0 
	\end{pmatrix} 
	\right\}  
	= \left\{\mu\indices{^2 _1}, 
		\mu\indices{^2 _2},
	\dots, \mu\indices{^2 _{N-1} }  \right\}; 
	\\
	\mathscr{M}^3 & \coloneqq \left\{ 
	%1st 
	\begin{pmatrix}
		\color{red} 0 & \color{red} 0 &\color{red} \imath &0  &\cdots  & 0 \\
		\color{red} 0 &  \color{red}  0 & \color{red} 0& 0& \cdots & 0 \\
		\color{red} \imath & \color{red} 0 & \color{red} 0& 0 & \cdots & 0 \\ 
		0 & 0 & 0 & 0 & \cdots & 0 \\
		\vdots & \vdots & \vdots & \vdots & \ddots & \vdots  \\
		0 & 0 & 0 & 0 & \cdots  & 0
	\end{pmatrix}, 
	%2nd 
	\begin{pmatrix}
		0 &  0 &0 &0  &\dots  & 0 \\
		0 &  \color{red}  0 & \color{red} 0&\color{red} \imath & \cdots & 0 \\
		0 & \color{red} 0 & \color{red} 0& \color{red} 0 & \cdots & 0 \\ 
		0 & \color{red} \imath & \color{red} 0 & \color{red} 0 & \cdots & 0 \\
		\vdots & \vdots & \vdots & \vdots & \ddots & \vdots  \\
		0 & 0 & 0 & 0 & \cdots  & 0
	\end{pmatrix},   
	\dots, 
	\begin{pmatrix}
		0 & \cdots    & 0 & 0    & 0 &0 \\
		\vdots  & \ddots      & \vdots & \vdots  & \vdots & \vdots   \\
		0 &\cdots  & 0 & 0 & 0 & 0
		\\
		0 & \cdots  & 0  &\color{red} 0 &\color{red} 0 & \color{red} \imath   \\
		0 & \cdots &0  & \color{red}	0  & \color{red} 0  & \color{red} 0  
		\\
		0 & \cdots  & 0  & \color{red} \imath &\color{red} 0 & \color{red}  0 
	\end{pmatrix} 
	\right\}  
	= \left\{\mu\indices{^3 _1}, \mu\indices{^3 _2},
	\dots, \mu\indices{^3 _{N-2} }  \right\};
	\\ 
	& %\makebox[\linewidth-3cm]{\dotfill}
	\qquad \makebox[\linewidth-3cm]
	{ \leavevmode\xleaders\hbox{\vdots \quad \quad }\hfill\kern0pt  } 
	\notag 
	\\
	\mathscr{M}^N & \coloneqq 
	\left\{ 
	\begin{pmatrix}
		\color{red} 0 & \color{red} \cdots & \color{red} \imath 
		\\ 
		\color{red} \vdots &\color{red} \ddots & \color{red} \vdots 
		\\ 
		\color{red} \imath & \color{red} \cdots & \color{red} 0 
	\end{pmatrix}
	\right\}  = \left\{ \mu\indices{^N_1}  \right\}.
\end{align}
\end{subequations}
The dimension of each subset is the same as it was in the case of $\mathcal{M}^\alpha $:
\begin{equation}
	| \mathscr{M}^\beta| = N+1 - \beta, 
	\qquad 
	\beta=2,\dots, N.
\end{equation} 
In other words, if we define a basis of $\gl_N (\R) $ (elementary matrices) $\{e\indices{_\alpha _\beta }\}_{\alpha,\beta=1}^N $:
\begin{equation}\label{eq:elementary-matrices-definition}
	(e\indices{_\alpha _\beta })_{i,j}=
	\begin{cases}
		1, & \text{if } i=\alpha,\, j=\beta,   
		\\ 
		0, & \text{otherwise},
	\end{cases}
\end{equation}
we can rewrite the sets \eqref{eq:cartan-subset}, \eqref{eq:M-subsets} and \eqref{eq:mu-subsets}   as
\begin{subequations}\label{eq:suN-sets-summarized}
	\begin{align}
		\mathcal{H} & = \left\{ 
		\imath (e_{NN} - e_{kk } ), 
		\qquad 
		k=1,\dots,N-1
		\right\};
		\\ 
		\mathcal{M}^{\alpha} &=
		\left\{ e_{k+\alpha-1, k} 
		- e_{k,k+\alpha-1},
		\qquad 
		k=1,\dots, \alpha
		\right\}, 
		\qquad 
		\alpha=2,\dots,N;
		\label{eq:suN-sets-summarized-M}
		\\
		\mathscr{M}^\beta &= \left\{ 
		\imath ( e_{k+\beta-1, k} 
		+ e_{k,k+\beta-1} ), \qquad 
		k=1,\dots, \beta 
		\right\}, 
		\qquad
		\beta=2,\dots,N. 
	\end{align}
\end{subequations} 
The sets $\mathcal{H}$, $\mathcal{M^\alpha } $ and $\mathscr{M}^\beta $ constitute a generalization of the elements $\Mat (t_1)= H_1 $, $\Mat (t_2)= M\indices{^2 _1} $ and $\Mat (t_3)=\mu\indices{^2 _1}$ that we studied in \Cref{example:frad2-su2} for $N=2 $. 

We see that all matrices forming the sets \eqref{eq:suN-sets-summarized} are linearly independent and antihermitian. The total dimension is 
\begin{equation}\label{eq:suN-basis-total-dimension}
	|H| + \sum_{\alpha=2}^N |\mathcal{M}^{\alpha}| + 
	\sum_{\beta=2}^N|\mathscr{M}^\beta| 
	= N-1 + 2  \sum_{\alpha=2}^N (N+1-\alpha )
	= N^2 -1. 
\end{equation}
Therefore, the set \eqref{eq:set-suN-basis} forms a basis of $\su_N $. 

Let us  proceed to the isomorphism construction. Now we can define an isomorphism: 
\begin{equation}\label{eq:mat-repr-rho-plus}  
	\rho_N\colon \frad_N \to \su_N,
\end{equation}
for example, in the following way. First, we map the diagonal elements (which belong to the Cartan subalgebra of $\frad_N $) to the Cartan subalgebra of $\su_N $:
\begin{equation}\label{eq:isomorphism-map-cartan}
	f_{i,i}  %=  f_{i,i}^-  
	\mapsto 2 \omega H_i, 
	\qquad
	i=1,\dots, N-1 
\end{equation}
(recall that $\omega >0 $, so we have $\sqrt{\omega^2} = \omega $). 
Then, some basis elements of $\frad_N $ are just proportional to those of $\su_N$: 
\begin{subequations}\label{eq:isomorphism-map-proportional}
	\begin{align}
		f_{1,2}  &\mapsto -\sqrt{2}\omega M\indices{^2 _1}, 
	& 
	f_{2,1}  & \mapsto  \sqrt{2}\omega \mu\indices{^2 _1};
	\\ 
	f_{j,2}  & \mapsto  \sqrt{2}\omega M\indices{^{j-1} _2 }, 
	& 
	f_{2,j}  & \mapsto \sqrt{2}\omega \mu\indices{^{j-1} _2 }, 
	& j=3,\dots, N. 
	\end{align}
\end{subequations}
(here $ j \ne 2 $ because this case was already covered in \eqref{eq:isomorphism-map-cartan}).

Finally, some basis elements can be mapped to the linear combination of $M\indices{^\alpha _k} $ and $\mu\indices{^\beta _l} $: 
\begin{subequations}\label{eq:isomorphism-map-lk}
	\begin{align}
		f _{i,j} &\mapsto 
		     \omega M\indices{^{j-i+1} _i}
		   -  \omega \mu\indices{^{j-i+1} _i},
		   & 1 \leq i , j  \leq N  
		   \ \& \ i,j\ne 2;
	\end{align}
\end{subequations}
The maps \eqref{eq:isomorphism-map-cartan}, \eqref{eq:isomorphism-map-proportional} and \eqref{eq:isomorphism-map-lk} define an $N\times N$-matrix representation of the Lie algebra $ \frad_N$ \eqref{eq:mat-repr-rho-plus}, which can be written in terms of $\gl_N (\R)$ basis using \eqref{eq:suN-sets-summarized}:
\begin{subequations}
\begin{align}
	 \Mat (f_{i,i} )   & =  2 \imath \omega (
    e_{N,N} - e_{i,i}   ) ,
 &    & &
 i&=1, \dots N-1; 
 \\   
 \Mat (f_{1,2}  )  & = -\sqrt{2}\omega ( e_{2, 1} - e_{1, 2} ), \quad   
& \Mat (f_{2,1}  )  & = \imath \sqrt{2}  \omega ( e_{2,1} + e_{1,2} );   \quad 
&  
\\
\Mat (f_{j,2}  )  & = \sqrt{2}\omega  
		( e_{j,2} - e_{2,j} ), 
		 &  
\Mat (f_{2,j}  )   & =\imath \sqrt{2}  \omega (e_{j,2} + e_{2,j}),  
& 
j&=3,\dots,N; 
\end{align}	
\begin{align}
		\Mat (f_{i,j} ) &=  
		     \omega ( e_{j,i} - e_{i,j} ) - \imath \omega (e_{i,j} + e_{j,i} ),
		   & 1 \leq i < j  \leq N  
		   \ \& \ i,j\ne 2. 
	\end{align}
\end{subequations}
\color{black}
or, visually: 
\begin{subequations}
\begin{equation}
\resizebox{1\textwidth}{!}{$
	\Mat (f_{1,1} ) = 
		\begin{pmatrix}
  	       -2 \imath \omega
  	         &  0  
  	           &  \cdots 
  	           & 0  
  	         \\ 
  	            0  
  	            &
  	              0  
  	               &
  	              \cdots   
  	              & 
  	              0 
  	             \\ 
  	              \vdots 
  	            &
  	              \vdots   
  	               &
  	             \ddots   
  	             & \vdots 
  	             \\
  	             0     &     0  & \cdots & 
  	             2 \imath \omega
  \end{pmatrix}, \ 
  \Mat(f_{2,2} )=
      \begin{pmatrix}
      	0 & 0 & \cdots & 0 \\ 
      	0 & -2 \imath \omega & \cdots & 0 \\
      	\vdots & \vdots & \ddots & \vdots \\
      	0 & 0 & \cdots & 2 \imath \omega
      \end{pmatrix}, 
      	 	\dots,  
      	 	\Mat (f_{N-1, N-1} )=
      	 	\begin{pmatrix}
      	0 & \cdots  &  0 & 0 \\ 
          	\vdots & \ddots & \vdots   & \vdots \\
      	0 & \cdots  & -2 \imath \omega & 0  \\ 
      	0 & \cdots  & 0 & 2 \imath \omega
      \end{pmatrix};
      $}
\end{equation}
\begin{equation}
	\Mat(f_{1,2} ) = \begin{pmatrix}
		  0 &   \sqrt{2}\omega & 0 & \cdots  & 0 \\
		  -\sqrt{2}\omega &   0 & 0& \cdots & 0 \\
		0 & 0 & 0 & \dots & 0 \\ 
		\vdots & \vdots & \vdots & \ddots & \vdots \\
		0 & 0 & 0 & \cdots & 0 
	\end{pmatrix},
	\quad 
	\Mat(f_{2,1} )=  \begin{pmatrix}
		 0 &  \imath  \sqrt{2}\omega & 0 & \cdots  & 0 \\
		  \imath \sqrt{2}\omega &   0 & 0& \cdots & 0 \\
		0 & 0 & 0 & \cdots & 0 \\ 
		\vdots & \vdots & \vdots & \ddots & \vdots \\
		0 & 0 & 0 & \cdots & 0 
	\end{pmatrix}; 
\end{equation}
\begin{equation}\label{eq:frad-plus-iso-visual-j2-2j}
	\Mat ( f_{j,2} ) = 
	\begin{pmatrix}
		0 & 0 & \cdots & 0 & \cdots & 0  \\ 
		0 & 0 & \cdots &- \sqrt{2}\omega & \cdots & 0 \\ 
		\vdots & \vdots & \ddots & \vdots 
		& \vdots & \vdots  \\
		0 & \sqrt{2}\omega & \dots & 0 & \cdots & 0  \\ 
		\vdots   &  
		\vdots & \ddots & \vdots  & \vdots & \vdots    \\ 
		0 & 0 & \cdots & 0 & \cdots & 0 
	\end{pmatrix},   
	\quad 
	\Mat ( f_{2,j} ) = 
	\begin{pmatrix}
		0 & 0 & \cdots & 0 & \cdots & 0  \\ 
		0 & 0 & \cdots & \imath \sqrt{2}\omega & \cdots & 0 \\ 
		\vdots & \vdots & \ddots & \vdots 
		& \vdots & \vdots  \\
		0 &  \imath \sqrt{2}\omega & \cdots & 0 & \dots & 0  \\ 
		\vdots   &  
		\vdots & \ddots & \vdots & \vdots & \vdots   \\ 
		0 & 0 & \cdots & 0 & \cdots & 0 
	\end{pmatrix}, 
\end{equation}
where the dimension of the blocks in \eqref{eq:frad-plus-iso-visual-j2-2j} is $(j-1) \times (j-1) $ and 
\begin{equation}\label{eq:fij-matrix1}
\renewcommand{\arraystretch}{1.5} 
\Mat (f _{i,j}) = 	  \left( \begin{array}{  c | c | c   }
      \mathbb{0}     & \mathbb{0}  & \mathbb{0}    \\   \hline 
     \mathbb{0}   & 
    \begin{NiceMatrix}
    	0  & \Cdots & -\omega(1+\imath) \\  
    	\Vdots & \Ddots & \Vdots \\ 
    	\omega(1-\imath) & \Cdots & 0
    \end{NiceMatrix}   
    & \mathbb{0}   \\  \hline 
     \mathbb{0}  & \mathbb{0}  & \mathbb{0}   
   \end{array}\right), 
   \qquad  1 \leq i < j  \leq N  
		   \ \& \ i,j\ne 2, 
\end{equation} 		
\begin{equation}\label{eq:fij-matrix2}
\renewcommand{\arraystretch}{1.5} 
\Mat (f _{i,j}) = 	  \left( \begin{array}{  c | c | c   }
     \mathbb{0}      & \mathbb{0}  & \mathbb{0}    \\   \hline 
     \mathbb{0}   & 
    \begin{NiceMatrix}
    	0 & \Cdots & \omega(1-\imath) \\  
    	\Vdots & \Ddots & \Vdots \\ 
    	-\omega(1+\imath) & \Cdots & 0 
    \end{NiceMatrix}  
     & \mathbb{0}   \\  \hline 
     \mathbb{0}  & \mathbb{0}  & \mathbb{0}    
   \end{array}\right), 
   \qquad  1 \leq j < i  \leq N  
		   \ \& \ i,j\ne 2, 
\end{equation}  
\end{subequations}		    
where the dimension of the central blocks in \eqref{eq:fij-matrix1} and \eqref{eq:fij-matrix2} is $(j-i+1) \times   (j-i+1)$.  	   

Summing up, since we constructed a bijective map from $\frad_N$ to $\su_N$,
seen as a matrix Lie algebra with the matrix commutator as Lie bracket, we can
conclude that we found the desired isomorphism.
\end{proof}
At this point, let us see an example of how to use the construction given in the proof of the
previous theorem:

\begin{example}[$\frad_3$]
    Let us consider $\frad_3$ as an example. The commutation relations are
    defined by \eqref{eq:frad-N-closed-Lie-bracket} and can be summarized as it
    is done in~\Cref{table:frad3-comm-table}. 

\begin{table}[h!]
\resizebox{1\textwidth}{!}{$\begin{array}{c|c|c|c|c|c|c|c|c|}
          %  \toprule
           \frad_3 &  f_{1,1} & f_{2,2} & f_{1,2} & f_{1,3} 
           & f_{2,1} & f_{2,3} 
           & f_{3,1}  & f_{3,2}
            \\
            \midrule
            f_{1,1} & 0 & 0 & 
            -2 \omega f_{2,1} &
            - 4 \omega f_{3,1} &
            2 \omega f_{1,2} &
           - 2 \omega f_{3,2} & 
           4 \omega f_{1,3} &
           2 \omega f_{2,3}
             \\
            \midrule 
		 f_{2,2} & 0 & 0 &
		   2 \omega f_{2,1} &
            - 2 \omega f_{3,1} &
            -2 \omega f_{1,2} &
              - 4 \omega f_{3,2} & 
           2 \omega f_{1,3} &
           4 \omega f_{2,3}
               \\
            \midrule 
            f_{1,2} &  2 \omega f_{2,1} 
            & -2 \omega f_{2,1} & 
            0 & 
             \omega (f_{2,3} -    f_{3,2}) &
             2\omega (   f_{2,2}
           -  f_{1,1} ) 
             & 
            - \omega (f_{1,3} + f_{3,1}) 
            &  \omega (f_{2,3} +   f_{3,2})
            &  \omega (f_{1,3} - f_{3,1})
              \\
            \midrule 
            f_{1,3} &  4 \omega f_{3,1} & 
            2 \omega f_{3,1} &  
             \omega (f_{3,2} - f_{2,3} )& 0 &  
              \omega (f_{2,3} +    f_{3,2}) &  \omega (f_{1,2} - f_{2,1} )
             & - 2 \omega f_{1,1} 
             & - \omega ( f_{1,2} + f_{2,1} ) 
              \\
            \midrule  
            f_{2,1} &  - 2 \omega f_{1,2} & 
            2 \omega f_{1,2} &  
            2 \omega (f_{1,1} - f_{2,2})
            & -  \omega ( f_{2,3} +    f_{3,2}) 
            & 0 &   \omega (f_{1,3} - f_{3,1}) & 
             \omega ( f_{2,3} -f_{3,2}   ) 
             &  \omega (f_{1,3} + f_{3,1}) 
            \\
            \midrule  
            f_{2,3} & 2 \omega f_{3,2} 
            & 4  \omega f_{3,2} & 
             \omega ( f_{1,3} + f_{3,1} )
            &  \omega (f_{2,1} -f_{1,2}  )
            &  \omega ( f_{3,1} -f_{1,3} ) 
            & 0 & 
            -  \omega (f_{1,2} + f_{2,1}  ) & - 2\omega f_{2,2}
            \\
            \midrule   
            f_{3,1} & -4\omega f_{1,3}
            & -2\omega f_{1,3}
            &  - \omega ( f_{2,3} +f_{3,2}   )
            & 2\omega f_{1,1}
            &  \omega ( f_{3,2} - f_{2,3} )
            &  \omega (f_{1,2} + f_{2,1}  )
            & 0 & 
            \omega ( f_{1,2} - f_{2,1} )
              \\
            \midrule    
              f_{3,2} & -2\omega f_{2,3}
              & -4\omega f_{2,3} 
              &   \omega ( f_{3,1} - f_{1,3}  ) 
              &   \omega (f_{1,2} + f_{2,1}  ) 
              & - \omega (f_{1,3} + f_{3,1} )
              & 2\omega f_{2,2}
              &   \omega (f_{2,1} - f_{1,2})
              & 0 
              \end{array}
 $}
	\caption{$ \frad_3$ commutation table. }	
	\label{table:frad3-comm-table} 
\end{table}

\begin{table}[h!]
$$ 
\begin{array}{c|c|c|c|c|c|c|c|c|}
          %  \toprule
           \su_3 & t_1 & t_2 & t_3 & t_4
           & t_5 & t_6 & t_7  &  t_8
            \\
            \midrule
           t_1 & 0 & 0 & t_6 & t_7 & 2t_8 & -t_3 & -t_4 & -2t_5
             \\
            \midrule
           t_2 & 0 & 0 & -t_6 & 2t_7 & t_8 & t_3 & -2t_4 & - t_5
             \\
            \midrule 
            t_3 & -t_6 & t_6 & 0 & -t_5 & t_4 
            & 2(t_1 -t_2) & -t_8 & t_7
              \\
            \midrule  
            t_4 & -t_7 & -2t_7 & t_5 & 0 & -t_3 &t_8 & 2t_2 & -t_6
              \\
            \midrule  
            t_5 & -2t_8 & -t_8 & -t_4 & t_3 & 0 & t_7 & -t_6 & 2 t_1 
              \\
            \midrule  
            t_6 & t_3 & -t_3 & 2(t_2-t_1) & - t_8 & -t_7 & 0 & t_5 & t_4 
              \\
            \midrule 
            t_7 & t_4 & 2t_4 & t_8 & -2t_2 & t_6 & -t_5 & 0 & -t_3 
              \\
            \midrule  
            t_8 & 2t_5 & t_5 & -t_7 & t_6 & -2t_1 & -t_4 & t_3 & 0 
 \end{array}
$$
\caption{$ \su_3$ commutation table. }	
	\label{table:su3-comm-table} 
\end{table} 

Let us now consider the Lie algebra $\su_3$ (see~\Cref{table:su3-comm-table}).  In
a  matrix representation, using the notation
\eqref{eq:suN-sets-summarized}, we have: 
\begin{subequations}
\begin{align}
	\Mat (t_1) = H_1 &= \begin{pmatrix}
		- \imath & 0 & 0 \\ 0 & 0 & 0 \\ 0 & 0 & \imath
	\end{pmatrix},  
	& 
	\Mat (t_2) = H_2 &= \begin{pmatrix}
		0 & 0 & 0 \\ 0 & - \imath & 0 \\ 0 & 0 & \imath 
	\end{pmatrix}; 
	& 
	& 
	\\ 
	\Mat (t_3) = M\indices{^2 _1} &= 
	\begin{pmatrix}
		 0 & -1 & 0 \\ 1 & 0 & 0 \\ 0 & 0 & 0 
	\end{pmatrix},  
	& 
	\Mat (t_4) = M\indices{^2 _2} &= 
	\begin{pmatrix}
		0 & 0 & 0 \\ 0 & 0 & -1 \\ 0 & 1 & 0 
	\end{pmatrix}, 
	& 
	 \Mat (t_5) =M\indices{^3 _1} &= \begin{pmatrix}
		0 & 0 & -1 \\ 0& 0 & 0 \\ 1 & 0&0 
	\end{pmatrix};
	\\ 
	\Mat(t_6) = \mu\indices{^2 _1} &=\begin{pmatrix}
		0 & \imath & 0 \\ \imath & 0 & 0 \\ 0& 0 &0 
	\end{pmatrix},
	& 
	\Mat (t_7) =\mu\indices{^2_2}&=\begin{pmatrix}
		0 & 0 & 0 \\ 0 & 0 & \imath \\ 0 & \imath & 0 
	\end{pmatrix}, 
	& 
	\Mat(t_8)=\mu\indices{^3 _1 }
	&= \begin{pmatrix}
		0 & 0 & \imath \\ 0 & 0 & 0 \\ \imath & 0 & 0 
	\end{pmatrix}. 
\end{align}	
\end{subequations}
Therefore, by  \eqref{eq:isomorphism-map-cartan}, \eqref{eq:isomorphism-map-proportional} and \eqref{eq:isomorphism-map-lk}, the isomorphism 
\begin{equation}
	\rho_3: \ \frad_3 \to \su_3 
\end{equation}
can be defined as: 
\begin{subequations}
\begin{align}
	\Mat (f_{1,1})  & = 2\omega H_1, 
	& 
	\Mat (f_{2,2})  & = 2\omega H_2;  
	\\ 
	\Mat (f_{1,2}) & = -\sqrt{2}\omega M\indices{^2 _1}, 
	& 
	\Mat  (f_{2,1}) & = \sqrt{2}\omega \mu\indices{^2 _1}, 
	 \\ 
	\Mat (f_{1,3})  & = \omega ( M\indices{^3 _1}-  \mu\indices{^3 _1}),
	 &
	\Mat (f_{3,1})   & = - \omega (M\indices{^3 _1} + \mu \indices{^3 _1}), 
	 \\  
	\Mat (f_{2,3})   & = \sqrt{2}\omega \mu\indices{^2 _2}, 
	 &   
	\Mat (f_{3,2}) & = \sqrt{2}\omega M\indices{^2 _2 }. 
	 \end{align}
\end{subequations} 
 
Therefore, similarly to \eqref{eq:frad-N=2example-arb-elem}, an arbitrary element $X \in \frad_3$ can be represented as $3\times 3$-matrix:  
\begin{equation}
\Mat (X) = 
\omega 
\begin{pmatrix}
-2 \imath a_{1,1} & \sqrt{2} (a_{1,2} +\imath a_{2,1}) & -\left(1+ \imath \right) a_{1,3}+\left(1- \imath \right) a_{3,1} 
\\
- \sqrt{2} (a_{1,2} -\imath a_{2,1} ) & -2 \imath a_{2,2} 
 & -\sqrt{2} ( a_{3,2} - \imath a_{2,3} ) 
\\
 \left(1- \imath \right) a_{1,3}-\left(1+ \imath \right) a_{3,1} 
 &  \sqrt{2} ( a_{3,2} + \imath a_{2,3}) 
 & 2 \imath ( a_{1,1}+ a_{2,2} ) 
\end{pmatrix}, 
\end{equation}
where $a_{1,1},\ldots, a_{3,2}$ are real coefficients. 
\end{example}

\subsection{Summary and Lie algebra isomorphism $\Frad_{N} \cong \mathfrak{u}_N $}

In the following theorem, we sum up our main finding of this Section, concluding
the proof of the case~\ref{th:main:1} of the Main Theorem. 
\begin{theorem}    
    The Lie algebra isomorphism $\Frad_{N} \cong \mathfrak{u}_N $ holds, \ie one can represent an arbitrary element $Y \in \Frad_{N} $ by an $N\times N $ matrix of the following shape: 
    	\begin{equation}\label{eq:Frad-mat-repr-full-uN}
    \omega \begin{pNiceMatrix}
    -2 \imath a_{1,1} 
    & \sqrt{2}( a_{1,2}+\imath a_{2,1} )& -\left(1+ \imath \right) a_{1,3}+\left(1- \imath \right) a_{3,1} 
    & \cdots 
    & -\left(1+\imath \right) a_{1,N}+\left(1- \imath \right) a_{N,1} 
\\
-\sqrt{2} ( a_{1,2} - \imath a_{2,1} )
 & -2 \imath a_{2,2} 
 & -\sqrt{2} (a_{3,2} - \imath a_{2,3}  ) 
 & \cdots 
 & -\sqrt{2} ( a_{N,2} -   \imath a_{2,N} )
\\
 \left(1- \imath \right) a_{1,3}-\left(1+ \imath \right) a_{3,1} 
 & \sqrt{2} (a_{3,2} + \imath a_{2,3}  ) & -2 \imath a_{3,3} 
 & 
 & -\left(1+ \imath \right) a_{3,N}+\left(1- \imath \right) a_{N,3} 
\\
  \vdots & \vdots   &  & \ddots 
   \\
 \left(1- \imath \right) a_{1,N}-\left(1+ \imath \right) a_{N,1} 
 &\sqrt{2}  (a_{N,2} + \imath a_{2,N}  ) 
 & \left(1- \imath \right) a_{3,N}-\left(1+ \imath \right) a_{N,3} 
 & 
  & -2 \imath a_{N,N} 
    \end{pNiceMatrix},  
    	\end{equation} 
   where $a_{1,1}, a_{2,2}, \dots, a_{N,N} $ are real coefficients. 
    \label{th:unfinal}
\end{theorem}
\begin{proof}
    Since $\u_N \cong  \su_N  \oplus  \R$, see for instance
    \cite{gilmoreLieGroupsLie2005}, the statement follows
    from~\Cref{th:LeviDecomp} and \Cref{th:iso-suN}.  The matrix \eqref{eq:Frad-mat-repr-full-uN}, in turn, can be obtained from \Cref{th:iso-suN} simply by intorducing new coefficient $a_{N,N} \in \R  $ and relaxing the traceless requirement. 
\end{proof} 

\begin{remark}\label{rmk:comrel-suN}
    We remark that given the Lie algebra isomorphism $ \frad_N  \cong \su_N $
    obtained in the proof of \Cref{th:iso-suN}, one can now use the
    commutation relations \eqref{eq:frad-N-closed-Lie-bracket} for the
    description of the special unitary Lie algebra $\su_N$ itself.
\end{remark}        

\section{Properties and isomorphism for the Lie algebra $\mFrad_{N}$}
\label{sec:Frad-minus}

In this Section we consider the $N^{2}$-dimensional Lie algebra $\mFrad_N=
\FradPlain_N (-\omega^2)$. We observe that in such a case the commutation
relations~\eqref{eq:pb-Fradkin} take the form: 
\begin{subequations}\label{eq:pb-Fradkin-minus}
    \begin{align}
        \bigl[L_{i,j},L_{k,l}  \bigr] &=
        L_{j,l}\delta_{i,k}+L_{k,j}\delta_{l,i}
        +L_{l,i}\delta_{j,k}+L_{i,k}\delta_{l,j},
        \label{eq:LijLlm-minus}
        \\
        \bigl[L_{i,j},F_{k,l}  \bigr] &=
        F_{j,l}\delta_{i,k}+ F_{k,j}\delta_{i,l}
        -F_{i,l}\delta_{j,k} 
        -F_{i,k}\delta_{j,l},
        \label{eq:LiFLlm-minus}
        \\
        \bigl[F_{i,j},F_{k,l}  \bigr] &=
        - \omega^2 \left(  
            L_{j,l}\delta_{i,k}+L_{j,k}\delta_{i,l}
            +L_{i,l}\delta_{j,k}+L_{i,k}\delta_{j,l}
        \right). 
        \label{eq:FijFlm-minus}
    \end{align}%
\end{subequations}
As we did in the case of $\Frad_N$, we present the Levi decomposition of $\mFrad_N $ in an
appropriate basis, the Killing form of its semisimple Levi
factor, and finally we construct an explicit isomorphism with $\gl_N (\R)$, thus
proving the point~\ref{th:main:2} of the Main Theorem.

\subsection{Change of the basis and Levi decomposition of $\mFrad_N$}

First, let us prove that the elements defined by
\eqref{eq:prop:more-general-transformation} form a Lie subalgebra of dimension
$N^{2}-1$ also in this case.

\begin{lemma} 
    The Lie algebra
\begin{align}
		\mfrad_N := 
	& \Span \bigl( 
	 \{ f_{i,i}  \}_{i=1}^{N-1} 
	 \cup  
	 \{ f_{i,j}  \}_{ {i,j=1},    {i \ne j}  }^{N} 
	\bigr)  
	\notag \\ 
	 = &  \Span ( f _{1,1},  f _{2,2},  \dots, f _{N-1, N-1}; f _{1,2}, \dots, f _{1,N},
	 f _{2,1}, \dots, f _{N-1, N} ),  
\end{align}
where $f _{i,j}$ are given by the transformation \eqref{eq:rotation}, is a semisimple $N^2-1 $-dimensional Lie subalgebra of $\mFrad_N$, whose Lie bracket is:   
\begin{equation}
	 	[f_{i,j} , f_{k,l} ]  =  
		2 \omega \, \bigl( 
		f_{k,j}  \delta_{i,l} 
		- f_{i,l}  \delta_{j,k}
		+\delta_{i,j} 
		( f_{N,l}  \delta_{N,k} - f_{k,N}  \delta_{l,N} )
		+ \delta_{k,l} 
		( f_{i,N}  \delta_{j,N} 
		- f_{N,j}  \delta_{N,i}
		)
	  \bigr). 
	 \label{eq:frad-N-minus-closed-Lie-bracket}  
\end{equation}
\end{lemma}
\begin{remark}
	Let us remark that despite the fact that $\alpha = -\omega^2 $ in this case, we perform the same change of basis as before. This has to be done in order to avoid complex transformations, which would make \Cref{prop:sign} not applicable anymore. Such a choice, in turn, causes some elements to be counted twice when computing the Lie bracket. 
\end{remark}

\begin{proof}
    By direct computation we have 
	\begin{align}
		[f_{i,j} , f_{k,l} ]_{\mfrad_N} &=
		\left[\F_{i,j} + \omega L_{i,j}- \delta_{i,j} F_{N,N} , 
		\F_{k,l} + \omega L_{k,l} - \delta_{k,l} F_{N,N} \right]_{\mFrad_N}
		\notag
		\\ 
		&= \left[ \F_{i,j}, \F_{k,l} \right]_{\mFrad_N}
		+\omega  \left[\F_{i,j}, L_{k,l} \right]_{\mFrad_N} 
		-\delta_{k,l} \left[ F_{i,j}, F_{N,N} \right]_{\Frad_N}
		 + \omega \left[ L_{i,j}, \F_{k,l} \right]_{\mFrad_N}
		 \notag \\ 
		&\phantom{=} +\omega^2   \left[ L_{i,j}, L_{k,l}\right]_{\mFrad_N}
		- \omega \delta_{k,l}  \left[ L_{i,j}, F_{N,N} \right]_{\mFrad_N}
		-\delta_{i,j}  \left[ F_{N,N}, F_{k,l} \right]_{\mFrad_N}
		-\omega \delta_{i,j}  \left[ F_{N,N}, L_{k,l} \right]_{\mFrad_N}.
		\label{eq:frad-minus-comm-relations-computation-in-new-basis-prop}
	\end{align} 
    Evaluating the commutators with the use of~\eqref{eq:pb-Fradkin-minus}, one
    arrives at the expression \eqref{eq:frad-N-minus-closed-Lie-bracket}.
    It is obviously closed, making $\mfrad_N$ a Lie subalgebra of $\mFrad_N$. 

    Next, using \eqref{eq:frad-N-minus-closed-Lie-bracket}, one can write down
    the structure constants of $\mfrad_N$ as: 
\begin{align}
	\tilde{c}^{(\alpha \beta)}_{(ij)(kl)} = 
	2 \omega \,
	\bigl[  
	\delta_{i,l} \delta\indices{^\alpha _k} \delta\indices{^\beta _j }
	-\delta_{j,k}\delta\indices{^\alpha _i} \delta\indices{^\beta _l }
	\notag % \\ 
		%& 
		&+\delta_{i,j} (
	 \delta_{N,k} \delta\indices{^\alpha _N} \delta\indices{^\beta _l }
	 - \delta_{l,N} \delta\indices{^\alpha _k} \delta\indices{^\beta _N}  
	 ) 
	\notag \\ 
	&+ \delta_{k,l} (
	 \delta_{j,N} \delta\indices{^\alpha _i} \delta\indices{^\beta _N }
	 -\delta_{N,i}\delta\indices{^\alpha _N } \delta\indices{^\beta _j } 
	 ) 
	 \bigr]. 
	 \label{eq:structure-constant-fradNminus}
	\end{align}
    The Killing form in dual basis reads as
    \eqref{eq:Killing-form-generalN-dual-basis} with the matrix element of the
    form:
    \begin{equation}\label{eq:Killing-matrix-el-fradNminus}
            \tilde{K}_{(ij) (kl)} = \sum_{\alpha, \beta, \gamma,\delta=1}^N 
            \tilde{c}^{(\gamma\delta)}_{(ij)(\alpha \beta)}
            \tilde{c}^{(\alpha \beta)}_{(kl)(\gamma \delta)}. 
    \end{equation}
    After simplifications, using \eqref{eq:structure-constant-fradNminus} and
    \eqref{eq:Killing-matrix-el-fradNminus}, we arrive at the following
    expression: 
    \begin{equation}\label{eq:Killing-form-generalN-matrix-el-frad-minus}
            \tilde{K}_{(ij) (kl)} = 8 N \omega^2 \bigl( 
            \delta_{k,l} \delta_{i,j}
            +\delta_{i,l}\delta_{j,k}
            -\delta_{N,k}\delta_{N,l} \delta_{i,j}
            -\delta_{N,i} \delta_{N,j} \delta_{k,l}
            \bigr). 
    \end{equation}
    We observe that we have non-zero entries if $i=j $ and $k=l $: 
    \begin{subequations}
    \begin{equation}\label{eq:kappa-fradN-minus}
            %\tilde \kappa_{i,j} \coloneqq 
             \tilde{K}_{(ii)(jj)}= 
            8 N \omega^2 
            \left(\delta_{i,j} +1 \right)
                    \mathop{=}^{\eqref{eq:block-matrix-elts-kappa} }  
                    - \kappa_{i,j}. 
    \end{equation}
    Also, one obtains non-zeros in the right-hand side of
    \eqref{eq:Killing-form-generalN-matrix-el-frad-minus} when $i=l $ and $j=k$
    (but $i \ne j $, otherwise we fall back to the particular case of
    \eqref{eq:kappa-fradN-minus}): 
            \begin{align}\label{eq:tilde-d-fradN-minus}
             \tilde{d}_{i,j} \coloneqq  
                    \tilde{K}_{(ij)(ji)} = 8N \omega^2 \delta_{i,j} . 
            \end{align}
    \end{subequations}
We can now write the matrix associated to the Killing form as a block
    matrix of the following shape: 
    \begin{equation}\label{eq:Killing-form-matrix-frad-minus} \renewcommand{\arraystretch}{1.5}
        \tilde{K} =    
                \left[\begin{array}{  c | c | c   }
        - \kappa   & \mathbb{0}  & \mathbb{0}   \\ \hline 
         \mathbb{0}   & \mathbb{0}  & \tilde{d}  \\ \hline 
        \mathbb{0}  & \tilde{d} & \mathbb{0}    
       \end{array}\right]_{N^2-1}
       =8 N \omega^2  \left[\begin{array}{  c | c | c    }
        \hat{\kappa}    & \mathbb{0}  & \mathbb{0}    \\ \hline 
         \mathbb{0}   & \mathbb{0}  &  \mathbb{1}_{{N(N-1)}/{2} }, \\ \hline 
         \mathbb{0}  &  \mathbb{1}_{ {N(N-1)}/{2}} & 0   
       \end{array}\right], 
    \end{equation}
    where $\tilde{d} = (\tilde{d}_{i,j})$ is a diagonal matrix of $N(N-1)/2$ size, defined by \eqref{eq:tilde-d-fradN-minus}; and $\mathbb{1}$ is an identity matrix of the same size. Using the known formula for the determinant of the block matrix~\cite{abadir2005matrix}: 
    \begin{equation}\label{eq:block-matrices-formula2}
            \det  \left[\begin{array}{  c | c   }
        A   & B \\ \hline 
         B  & A     
       \end{array}\right] = \det (A-B)\cdot \det (A+B),
    \end{equation}
    together with \eqref{eq:hatkappa-determinant-compitation}, we arrive at the following expression: 
    \begin{equation}
            \det ( \tilde{K}) = 
            (-1)^{\frac{N(N-1)}2} N^{N^2} 
            (8  \omega^2)^{N^2-1}.   
    \end{equation}
    Therefore, given $\omega \ne 0 $, the Killing form
    \eqref{eq:Killing-matrix-el-fradNminus} is non-degenerate and the Lie
    algebra $\mfrad_N $ is semisimple by Cartan's criterion. 
\end{proof}

\begin{remark}\label{rmk:ordering}
We remark that, in the previously considered case  $\frad_N$, the matrix $d$ in
\eqref{eq:Killing-form-matrix} is diagonal with all the elements being the
same, \ie proportional to the identity matrix. In this regard, the particular
ordering of the generators $\set{f_{i,j}}$ does not affect the shape of $d$. 

In the case of $\mfrad_N$  we have, in contrast, the anti-diagonal block: 
\begin{equation}\renewcommand{\arraystretch}{1.5}
	  \left[\begin{array}{  c | c   }
    \mathbb{0}    & \tilde{d}  \\ \hline 
     \tilde{d}  & \mathbb{0}     
   \end{array}\right],
   \label{eq:dblock}
\end{equation}  
instead of $d$, which is not diagonal itself. This matrix consists of the
coefficients of the elements of the form $\phi_{i,j} \phi_{j,i}$.
Therefore, to make its blocks $\tilde{d}$ diagonal, we must choose the
ordering of the generators to be: 
\begin{align}
	\{ f _{i,j} \} & = 	
	\{ f _{k,k} \}_{k=1}^{N-1} 
	\cup 
	\{ f _{i,j} \}_{1 \leq i<j \leq N}
	\cup 
	\{ f _{i,j} \}_{1 \leq j<i \leq N}  
	 \notag
	 \\ 
	 & =\left\{ f _{1,1}, f _{2,2}, \dots, f _{(N-1), (N-1)};  
	  f _{1,2}, f _{1,3}, \dots, f _{(N-1), N};
	  f _{2,1}, f _{3,1}, \dots, f _{N, (N-1)}
	 \right\}. 
\end{align}
\end{remark}
We close this Section with an analogue of~\Cref{th:LeviDecomp}:

\begin{theorem}\label{th:LeviDecomp-FradN-minus}
    The (trivial) Levi decomposition of $\mFrad_N$ is given by:
    \begin{equation}\label{eq:th-levi-direct-sum}
            \mFrad_N= \mfrad_N  \oplus   \R, 
    \end{equation}  
    where $\mfrad_N$ is semisimple, and the $1$-dimensional radical of
    ${\mFrad_N}$, spanned by $r_N \coloneqq  \F_{1,1} + \ldots + \F_{N,N} $. 
\end{theorem}

We omit the proof since it follows the same reasoning as for
\Cref{th:LeviDecomp} \mm.

\subsection{Properties of the Lie algebra $\mfrad_N$  }

In this Section, we investigate the properties of the semisimple Lie algebra
${\mfrad_N}$. We start with the following proposition:

\begin{prop}\label{th:eigenvalues-Killing-frad-minus}
    The eigenvalues of the matrix of the Killing form
    \eqref{eq:Killing-form-matrix-frad-minus} are given by: 
    \begin{align}
	\tilde{\lambda}_1 = 8N^2\omega^2 ,
	\quad  
        \tilde{\lambda}_2  = \ldots = \tilde{\lambda}_{N(N-1)/2} =8N \omega^2, 
	\quad 
	\tilde{\lambda}_{N(N-1)/2+1} = \ldots = \tilde{\lambda}_{N^2-1} =-8N \omega^2.
    \end{align}
\end{prop}
\begin{proof}
    Similarly to the proof of~\Cref{th:eigenvalues-Killing}, we have, for the
    set of eigenvalues of $ - \kappa $: 
    \begin{align}
            \tilde{\lambda}_1 = 8N^2\omega^2 , 
            \qquad
             \tilde{\lambda}_2 = \ldots = \tilde{\lambda}_{N-1} =8N \omega^2.
    \end{align}
    On the other hand, the eigenvalues of the block matrix~\eqref{eq:dblock}
    are obtained by scaling by $8N\omega^2$ the eigenvalues of the anti-diagonal block matrix:
    \begin{equation}
        \left[\begin{array}{  c | c   }
          \mathbb{0}   &   \mathbb{1}_{{N(N-1)}/{2} } \\ \hline 
           \mathbb{1}_{{N(N-1)}/{2} }  & \mathbb{0}    
       \end{array}\right].
        \label{eq:andiag}
    \end{equation}
    Using again formula~\eqref{eq:block-matrices-formula2}, one can easily see that:
    \begin{equation}
            \det  \left[\begin{array}{  c | c   }
        - \hat{\lambda} \mathbb{1}    &  \mathbb{1} \\ \hline 
          \mathbb{1}  &  - \hat{\lambda} \mathbb{1}   
       \end{array}\right] = 
      \det ( - \mathbb{1} - \hat{\lambda} \mathbb{1} )
       \cdot \det ( \mathbb{1}  - \hat{\lambda} \mathbb{1}), 
    \end{equation}	
    \ie the set of the eigenvalues of this matrix is given by the union of the
    sets of eigenvalues of the matrices $\mathbb{1} $  and $-\mathbb{1}$.
    Therefore, scaling by $8N\omega^2$  we obtain the statement. 
\end{proof}

\begin{corollary}
	The semisimple Lie algebra $\mfrad_N$ is  
	non-compact. 
\end{corollary}
\begin{proof}
    Given $\omega^2 >0 $, from \Cref{th:eigenvalues-Killing-frad-minus} one can
    observe that the Killing form of $\mfrad_N$ is indefinite because the
    eigenvalues of $\tilde{K}$ have different signs. Hence, the Lie algebra
    $\mfrad_N$ is non-compact. 
\end{proof} 

\begin{remark}
    Let us remark that the signature of $\mfrad_N$ is $(N(N-1)/2,N(N-1)/2-2,0)$.
    Differently from the case of $\frad_N$, this does not define uniquely the
    isomorphism class. Indeed, there are different non-compact real forms of
    $\sl_N(\mathbb{C})$ admitting the same signature, but non-isomorphic. This can
    be inferred by computing the simple roots of the associated Cartan
    decomposition and using them to build the associated Satake diagrams~\cite{Satake1960}.
    However, in the sequent Subsection, we present an explicit isomorphism and
    we will not need to use this kind of reasoning.
    \label{rem:satake}
\end{remark}

\subsection{Classification and explicit isomorphism map $\mfrad_N \to \sl_N(\mathbb{R})$}

As for the case of $\frad_N$, before giving the general result, we first present in all detail the simplest example, \ie the case $N=2$. 

\begin{example}[$\mfrad_2$] \label{eq:fraNminus-example-sl2}
	We have  
	\begin{equation}
		\mFrad_2 = \mfrad_2 \oplus \Span(r_2).
	\end{equation}
Following \eqref{eq:frad-N-minus-closed-Lie-bracket}, we may write down the commutation  relations  for $ \mfrad_2$ ({see \cref{table:mfrad2-comm-table}}). 

\begin{table}[h!]
	$$
		  \begin{array}{c||c|c|c}
          %  \toprule
           \mfrad_2 &  f_{1,1} & f_{1,2} & f_{2,1}
            \\
            \midrule  \midrule 
           f_{1,1} & 0 & - 4 \omega  f_{1,2} &  4 \omega f_{2,1}
            \\    \midrule 
           f_{1,2} & 4 \omega f_{1,2}  & 0 & -2 \omega f_{1,1}
            \\   \midrule 
            f_{2,1} & - 4 \omega f_{2,1}  & 2 \omega f_{1,1} &  0 
            \end{array}
	$$
\caption{$ \mfrad_2$ commutation table for $\omega>0 $. }	
	\label{table:mfrad2-comm-table} 
\end{table}

\noindent This algebra is isomorphic to $\sl_2(\mathbb{R}) = \Span \{ t_1, t_2, t_3 \} $, with the commutation relations summarized in~\cref{table:sl2-comm-table}. 

\begin{table}[h!]
		$$
		  \begin{array}{c||c|c|c}
          %  \toprule
           \sl_2(\mathbb{R}) &  t_1  & t_2 & t_3 
            \\
            \midrule  \midrule 
         t_1 & 0 & - 2 t_2  &  2 t_3 
            \\    \midrule 
          t_{2} & 2 t_2   & 0 & -t_1
            \\   \midrule 
          t_{3}  & -2 t_3  & t_1 &  0 
            \end{array}
	$$  
	\caption{$\sl_2(\mathbb{R})$ commutation table}
	\label{table:sl2-comm-table}
\end{table}
We choose the basis such that  
in a matrix representation we have: 
\begin{equation}\label{eq:fraNminus-example-sl2-basis}
	H \coloneqq  \Mat (t_1 ) =   \begin{pmatrix}
		1 & 0 \\ 0 & -1
	\end{pmatrix},
	\quad  
	e_{1,2}= \Mat (t_2 )
	= \begin{pmatrix}
		0 & 1 \\ 0 & 0 
	\end{pmatrix}, 
	\quad 
	e_{2,1} =  
 \Mat (t_3 )
	= \begin{pmatrix}
		0 & 0 \\ 1 & 0 
	\end{pmatrix}.
\end{equation}

The choice of the ordering of the basis is again dictated by the fact that we want the Cartan element to appear first.  

From the Tables~\ref{table:mfrad2-comm-table}~and~\ref{table:sl2-comm-table},
we see that these algebras are isomorphic. The explicit isomorphism is given
by 
\begin{equation}
	\varphi_2: \  \mfrad_2  \to \sl_2 (\mathbb{R})
\end{equation} 
such that: 
\begin{equation}
	\varphi_2(f_{1,1})= -2 \omega t_1,
	\qquad
	\varphi_2(f_{1,2})=- 2 \omega t_2, 
	\qquad
	\varphi_2(f_{2,1})=  -2 \omega t_3,
\end{equation}	
\ie\!: 
\begin{equation}
	\Mat (f_{1,1}) =- 2 \omega \begin{pmatrix}
		1 & 0 \\ 0 &  -1
	\end{pmatrix},
	\quad 
	\Mat ( f_{1,2} )= -2 \omega\begin{pmatrix}
		0 & 1 \\ 0 & 0 
	\end{pmatrix},
	\quad 
	\Mat (f_{2,1}) =-2 \omega  \begin{pmatrix}
		0 & 0 \\ 1& 0 
	\end{pmatrix}. 
\end{equation}
Therefore, one can write an arbitrary element $X \in \mfrad_2 $ as a $2\times 2 $ matrix:
\begin{equation}
	\Mat(X)= \Mat ( a_{1,1} f_{1,1} 
	+ a_{1,2} f_{1,2} 
	+ a_{2,1} f_{2,1} 
	) 
	= -2 \omega \begin{pmatrix}
		 a_{1,1} & a_{1,2}
		\\
		 a_{2,1} & - a_{1,1} 
	\end{pmatrix}. 
\end{equation}
\end{example}

Now, we have the general result:

\begin{theorem}
	\label{th:iso-slN}
	The Lie algebra isomorphism $\mfrad_N \cong \sl_N(\mathbb{R}) $ holds. 
\end{theorem}
\begin{proof}
	Let us prove the statement by constructing an explicit isomorphism 
	\begin{equation}\label{eq:isomorphism-to-slN}
		\varphi_N\colon \mfrad_N \to \sl_N (\mathbb{R}),  
	\end{equation}
\ie a representation of $\mfrad_N$ by $N \times N $ matrices. For $\sl_N (\mathbb{R})$ itself, we can introduce a basis as:   
\begin{equation}
	\widetilde{\mathcal{H}} \cup \mathcal{E}, 
\end{equation}
where $\widetilde{\mathcal{H}}$  is a Cartan subalgebra:
\begin{align}
	 	\widetilde{\mathcal{H}} \coloneqq \left\{
	 	\begin{pmatrix}
  	       1 
  	         &  0  
  	           &  \cdots 
  	           & 0  
  	         \\ 
  	            0  
  	            &
  	              0  
  	               &
  	              \cdots   
  	              & 
  	              0 
  	             \\ 
  	              \vdots 
  	            &
  	              \vdots   
  	               &
  	             \ddots   
  	             & \vdots 
  	             \\
  	             0     &     0  & \cdots & 
  	            - 1 
  \end{pmatrix}, 
      \begin{pmatrix}
      	0 & 0 & \cdots & 0 \\ 
      	0 & 1 & \cdots & 0 \\
      	\vdots & \vdots & \ddots & \vdots \\
      	0 & 0 & \cdots & -1 
      \end{pmatrix}, 
      	 	\cdots,  
      	 	\begin{pmatrix}
      	0 & \cdots  &  0 & 0 \\ 
          	\vdots & \ddots & \vdots   & \vdots \\
      	0 & \cdots  & 1 & 0  \\ 
      	0 & \cdots  & 0 & -1 
      \end{pmatrix} 
	 	\right\} 
	 	&= \bigl\{ 
		 (e_{kk} - e_{NN} )\bigr\}_{k=1}^{N-1}  
	 \notag	\\ & 	\eqqcolon \left\{ 
 	\widetilde{H}_1, \dots, 
	 	\widetilde{H}_{N-1} \right\} \label{eq:cartan-subset-slN}
	 \end{align}
and $\mathcal{E}$ is a set of all elementary matrices \eqref{eq:elementary-matrices-definition} excluding those with non-zero diagonal elements: 
\begin{equation}\label{eq:E-definiton}
	\mathcal{E} \coloneqq \{ e_{i,j}\}_{
	i \ne j  }^N. 
\end{equation}	
This clearly generalizes \eqref{eq:fraNminus-example-sl2-basis} from  \Cref{eq:fraNminus-example-sl2} for $N=2 $. 

The dimension check shows that: 
\begin{equation}
	|\widetilde{\mathcal{H}} |
	+ |\mathcal{E} | 
	= N-1 + N (N-1) = N^2 -1. 
\end{equation}
Let us now define the isomorphism~\eqref{eq:isomorphism-to-slN} via the
following rules:
\begin{subequations}\label{eq:iso-sl-rules}
	\begin{align}
		f_{i,i}&\mapsto -2 \omega \widetilde{H}_i, \qquad i=1,\dots,N-1 ; \label{it:iso-sl-rule1}
		\\ 
		f_{i,j}&\mapsto  -2 \omega e_{i,j}, 
		\qquad 
		1 \leq  i, j \leq N \  \& \  i \ne j. 
		\label{it:iso-sl-rule2} 
	\end{align}
\end{subequations}
It is possible to check that the matrices defined
    by~\eqref{eq:iso-sl-rules} generate a Lie algebra with commutation
    relations given by~\eqref{eq:frad-N-minus-closed-Lie-bracket}. Moreover, a
    function $\varphi_N$ is clearly injective and surjective. 
    One can write an arbitrary element $X \in \mfrad_N $ as the following $N \times N $ matrix: 
    \begin{equation}\label{eq:frad-minus-N-arbitrary-element-representation}
    \Mat (X) = - 2 \omega
    	 \begin{pNiceMatrix}
    	 	a_{1,1} & a_{1,2} & \Cdots & a_{1,N}
    	 	\\
    	 	a_{2,1} & a_{2,2} & \Cdots & a_{2,N}
    	 	\\
    	 	\Vdots & \Vdots & \Ddots & \Vdots 
    	 	\\
    	 	a_{N,1} & a_{N,2} & \Cdots & - \sum_{k=1}^{N-1} a_{k,k}
    	 \end{pNiceMatrix}
    \end{equation}
    where $a_{i,j} $ are real coefficients. 
    This concludes
    the proof of the statement.   
\end{proof}

We conclude this Section with the following example:

\begin{example}[$\mfrad_3$]
        Let us consider $\mfrad_3$ as a particular case of the
        \Cref{th:iso-slN}. The commutation relations are defined by
        \eqref{eq:frad-N-minus-closed-Lie-bracket} and can be summarized as in
        \Cref{table:frad3-minus-comm-table}. 		
\begin{table}[h!]
	\resizebox{1\textwidth}{!}{$
	\begin{array}{c||c|c|c|c|c|c|c|c}
   & f_{1,1} & f_{2,2} & f_{1,2} & f_{1,3} & f_{2,3} & f_{2,1} & f_{3,1} & f_{3,2} 
\\ \midrule \midrule 
 f_{1,1}  & 0 & 0 & -2 \omega f_{1,2} & -4 \omega f_{1,3} & -2 \omega f_{2,3} & 2 \omega f_{2,1} & 4 \omega f_{3,1} & 2 \omega f_{3,2} 
\\   \midrule
 f_{2,2}  & 0 & 0 & 2 \omega f_{1,2} & -2 \omega f_{1,3} & -4 \omega f_{2,3} & -2 \omega f_{2,1} & 2 \omega f_{3,1} & 4 \omega f_{3,2} 
\\  \midrule
 f_{1,2}  & 2 \omega f_{1,2} & -2 \omega f_{1,2} & 0 & 0 & -2 \omega f_{1,3} & 
2 \omega (f_{2,2}-f_{1,1}) & 2 \omega f_{3,2} & 0 
\\  \midrule
 f_{1,3}  & 4 \omega f_{1,3} & 2 \omega f_{1,3} & 0 & 0 & 0 & 2 \omega f_{2,3} & -2 \omega f_{1,1} & -2 \omega f_{1,2} 
\\   \midrule 
 f_{2,3}  & 2 \omega f_{2,3} & 4 \omega f_{2,3} & 2 \omega f_{1,3} & 0 & 0 & 0 & -2 \omega f_{2,1} & -2 \omega f_{2,2} 
\\   \midrule 
 f_{2,1}  & -2 \omega f_{2,1} & 2 \omega f_{2,1} & 2 \omega (f_{1,1}- f_{2,2}) & -2 \omega f_{2,3} & 0 & 0 & 0 & 2 \omega f_{3,1} 
\\   \midrule 
 f_{3,1}  & -4 \omega f_{3,1} & -2 \omega f_{3,1} & -2 \omega f_{3,2} & 2 \omega f_{1,1} & 2 \omega f_{2,1} & 0 & 0 & 0 
\\   \midrule
 f_{3,2}  & -2 \omega f_{3,2} & -4 \omega f_{3,2} & 0 & 2 \omega f_{1,2} & 2 \omega f_{2,2} & -2 \omega f_{3,1} & 0 & 0 
\end{array}
	$}
		\caption{$ \mfrad_3$ commutation table. }	
	\label{table:frad3-minus-comm-table} 
\end{table}

\begin{table}[h!]
$$ 
\begin{array}{c||c|c|c|c|c|c|c|c}
   & H_{1} & H_{2} & e_{1,2} & e_{1,3} & e_{2,3} & e_{2,1} & e_{3,1} & e_{3,2} 
\\ \midrule \midrule  
 H_{1}  & 0 & 0 & -e_{1,2} & -2 e_{1,3} & -e_{2,3} & e_{2,1} & 2 e_{3,1} & e_{3,2} 
\\  \midrule 
 H_{2}   & 0 & 0 & e_{1,2} & -e_{1,3} & -2 e_{2,3} & -e_{2,1} & e_{3,1} & 2 e_{3,2} 
\\  \midrule 
 e_{1,2}   & e_{1,2} & -e_{1,2} & 0 & 0 & e_{1,3} & -H_{1}+H_{2} & -e_{3,2} & 0 
\\  \midrule 
 e_{1,3}   & 2 e_{1,3} & e_{1,3} & 0 & 0 & 0 & -e_{2,3} & -H_{1} & e_{1,2} 
\\  \midrule 
 e_{2,3}  & e_{2,3} & 2 e_{2,3} & -e_{1,3} & 0 & 0 & 0 & e_{2,1} & -H_{2} 
\\  \midrule 
 e_{2,1}   & -e_{2,1} & e_{2,1} & H_{1}-H_{2} & e_{2,3} & 0 & 0 & 0 & -e_{3,1} 
\\  \midrule 
 e_{3,1}   & -2 e_{3,1} & -e_{3,1} & e_{3,2} & H_{1} & -e_{2,1} & 0 & 0 & 0 
\\  \midrule 
 e_{3,2}   & -e_{3,2} & -2 e_{3,2} & 0 & -e_{1,2} & H_{2} & e_{3,1} & 0 & 0 
\end{array}
$$
\caption{$ \sl_3(\mathbb{R})$ commutation table. }	
	\label{table:sl3-comm-table} 
\end{table} 

Therefore, by~\Cref{th:iso-slN}, the isomorphism can be choosen as follows:  
\begin{subequations}
\begin{align}
	\Mat (f_{1,1}) & = -2\omega H_1, 
	 & 
	\Mat  (f_{2,2}) & = -2\omega H_2;
	 \\ 
	 \Mat  (f_{1,2})& = -2\omega e_{1,2}, 
	  & 
	\Mat  ( f_{2,1})  & =  -2\omega e_{2,1}, 
	  \\ 
	 \Mat ( f_{1,3} )& = -2\omega e_{1,3}, 
	  & 
	\Mat  ( f_{3,1} )& = -2\omega e_{3,1}, 
	  \\ 
	\Mat   (f_{2,3}) & = -2\omega e_{2,3},
	  &
	 \Mat  (f_{3,2})  & = -2\omega e_{3,2}. 
\end{align}
\end{subequations}

An arbitrary element $X \in \mfrad_3$, in turn, can be represented as:
\begin{equation}
	\Mat(X) = -2 \omega \begin{pmatrix}
		a_{1,1} & a_{1,2} &  a_{1,3}
    	 	\\
    	 	a_{2,1} & a_{2,2} &  a_{2,3}
    	 	\\
    	 	a_{3,1} & a_{3,2} &  - a_{1,1} - a_{2,2} 
	\end{pmatrix}. 
\end{equation}
\end{example}

\subsection{Summary and Lie algebra isomorphism $\mFrad_{N}   
        \cong \mathfrak{gl}_N (\R)$ }

In the following Theorem we sum up our main finding of this Section, concluding
the proof of the case~\ref{th:main:2} of the Main Theorem:

\begin{theorem} \label{th:glNfinal}
    The Lie algebra isomorphism  $\mFrad_{N}  \cong \mathfrak{gl}_N (\R)   $
    holds, \ie one can represent an arbitrary element $Y \in \mFrad_{N} $ by an $N\times N $ matrix of the following shape: 
    	\begin{equation}\label{eq:mFrad-mat-repr-full-glN}
        - 2 \omega
    	 \begin{pNiceMatrix}
    	 	a_{1,1} & a_{1,2} & \Cdots & a_{1,N}
    	 	\\
    	 	a_{2,1} & a_{2,2} & \Cdots & a_{2,N}
    	 	\\
    	 	\Vdots & \Vdots & \Ddots & \Vdots 
    	 	\\
    	 	a_{N,1} & a_{N,2} & \Cdots &  a_{N,N}
    	 \end{pNiceMatrix}
    	\end{equation} 
   where $a_{i,j} $ are real coefficients. 
\end{theorem}

\begin{proof}
        It is well-know that $\gl_N (\R) \cong \sl_N (\R) \oplus \R$, see for
        instance~\cite{Humphreys1972.Introduction_Lie_Algebras_Representation_Theory}.
        Therefore, the isomorphism $\mFrad_{N}  \cong \mathfrak{gl}_N (\R)  $
        follows from~\Cref{th:LeviDecomp-FradN-minus} and~\Cref{th:iso-slN}.   Analogously to the previous case, the matrix \eqref{eq:mFrad-mat-repr-full-glN}, can be obtained from \Cref{th:iso-slN}  by introducing a new coefficient $a_{N,N} \in \R  $ and relaxing the traceless requirement. 
\end{proof}     
\newpage

\section{Properties and isomorphism for the Lie algebra $\FradZero_{N}$}
\label{sec:Frad-zero}

Let us finally consider the case $\FradZero_{N}$ by setting  $\alpha =0$
in~\eqref{eq:FijFlm}. The commutation relations~\eqref{eq:pb-Fradkin}
for $\FradZero_N$ can be rewritten as:
\begin{subequations}\label{eq:pb-Fradkin-zero}
    \begin{align}
        \bigl[L_{i,j},L_{k,l}  \bigr] &=
        L_{j,l}\delta_{i,k}+L_{k,j}\delta_{l,i}
        +L_{l,i}\delta_{j,k}+L_{i,k}\delta_{l,j},
        \label{eq:LijLlm-Fradkin-zero}
        \\
        \bigl[L_{i,j},F_{k,l}  \bigr] &=
        F_{j,l}\delta_{i,k}+ F_{k,j}\delta_{i,l}
        -F_{i,l}\delta_{j,k} 
        -F_{i,k}\delta_{j,l},
        \label{eq:LiFLlm-Fradkin-zero}
        \\
        \bigl[F_{i,j},F_{k,l}  \bigr] &=
       0.        \label{eq:FijFlm-Fradkin-zero}
    \end{align} 
\end{subequations}

From \eqref{eq:pb-Fradkin-zero}, one can notice that the Lie algebra
$\FradZero_N$ can be decomposed into a semidirect sum. Moreover, since in
this case the subalgebra generated by $ \{ L_{i,j} \}_{1\leq i, j\leq N}/
\mathcal{I} $ is semisimple, one can again consider it as a Levi
factor. The discussion here is also slightly different from the cases
considered above because, as one can observe from
\eqref{eq:pb-Fradkin-zero}, the basis of $\FradZero_N $ is already in
canonical form for the Levi decomposition. Therefore, unlike the previous
cases, there is no need to perform a change of basis.

We also mention that, of course, the symmetry algebra associated with the free motion (see \eqref{eq:cont-Hamiltonian} for $\alpha=0 $) is the Euclidean Lie algebra $\mathfrak{e}_N$.  We will see in the following how  the latter is connected to our parametric Lie algebra. 

In this Section, we prove that $\FradZero_N$ is a Lie subalgebra of the
universal enveloping algebra of the Euclidean algebra
$\mathfrak{e}_{N}$. Then, as mentioned above, we show
its Levi decomposition, and finally we build a matrix representation 
of $\FradZero_N$ by $(\mathcal{N}+1)\times (\mathcal{N}+1) $ matrices, where:
\begin{equation}
	\mathcal{N} \coloneqq \frac{N (N+1)}{2}.
\end{equation} 
Thus, we prove the point \ref{th:main:3} of the Main Theorem.

Again, before presenting the general results, let us first consider the known
example of $N=2$:

\begin{example}[$\FradZero_2$]
	The algebra $\FradZero_2= \Span (F_{1,1}, F_{2,2}, F_{1,2}, L_{1,2}) $ is 4-dimensional, with the Lie bracket satisfying~\cref{table:frad0-N=2-comm-table}. 
	\begin{table}[h!]
	\[
		  \begin{array}{c||c|c|c|c|}
          %  \toprule
          \FradZero_2 & F_{1,1} & F_{2,2} & F_{1,2} & L_{1,2} 
            \\
            \midrule  \midrule 
           F_{1,1}  & 0 & 0  &  0 & -2 F_{1,2} 
            \\    \midrule 
           F_{1,2}  & 0   & 0 & 0 & 2 F_{1,2} 
            \\   \midrule 
            F_{1,2}  & 0  & 0  &  0 & F_{1,1} - F_{2,2}
            \\ \midrule 
            L_{1,2} & 2 F_{1,2} & -2 F_{1,2} & 
            F_{2,2} - F_{1,1} & 0 
            \end{array}
        \]
\caption{$ \FradZero_2$ commutation table. }	
	\label{table:frad0-N=2-comm-table} 
\end{table} 	

\noindent One could notice that $\Span (F_{1,1}, F_{2,2}, F_{1,2})$ forms a
radical of $\FradZero_2$, which can be represented by $4 \times 4$
matrices in the following way: 
	\begin{subequations}
\begin{align}\label{eq:frad0-N=2-matricesF}
	\Mat (F_{1,1}) &= \begin{pmatrix}
		 0 & 0 & 0 & 1
		 \\ 
		 0 & 0 & 0 & 0 
		 \\ 
		  0 & 0 & 0 & 0  
		  \\
		   0 & 0 & 0 & 0 
	\end{pmatrix}, 
	&
	\Mat (F_{2,2}) &= \begin{pmatrix}
		 0 & 0 & 0 & 0
		 \\ 
		 0 & 0 & 0 & 1 
		 \\ 
		  0 & 0 & 0 & 0  
		  \\
		   0 & 0 & 0 & 0 
	\end{pmatrix},
	&
	\Mat (F_{1,2}) &= \begin{pmatrix}
		 0 & 0 & 0 & 0
		 \\ 
		 0 & 0 & 0 & 0
		 \\ 
		  0 & 0 & 0 & 1  
		  \\
		   0 & 0 & 0 & 0 
	\end{pmatrix}. 
\end{align}	
On the other hand, for the one-dimensional subalgebra $\Span \{L_{i,j} \} \cong
\so(2) \cong   \R$ we can choose the following $4\times 4$ matrix 
representation: 
\begin{equation}\label{eq:frad0-N=2-matrixL}
	\Mat (L_{1,2}) = \begin{pmatrix}
		0 & 0 & -1 & 0 
		\\
		0 & 0 & 1 & 0 
		\\ 
		2 & -2 & 0 & 0 
		\\
		0 & 0 & 0 & 0 
	\end{pmatrix}. 
\end{equation}
\end{subequations} 
It might be checked in a straightforward way that the
matrices~\eqref{eq:frad0-N=2-matricesF} and~\eqref{eq:frad0-N=2-matrixL}
satisfy the relations from \Cref{table:frad0-N=2-comm-table}.

 From \eqref{eq:frad0-N=2-matricesF} and \eqref{eq:frad0-N=2-matrixL}, it follows that an arbitrary element $X \in \FradZero_2 $ can be represented~as:
\begin{equation}
	\Mat(X)=  \Mat ( a_{1,1} F_{1,1} 
			 + a_{2,2} F_{2,2}
			 + a_{1,2} F_{1,2} 
			 + \ell_{1,2} L_{1,2}) 
			=  \begin{pmatrix}
		0 & 0 & - \ell_{1,2} & a_{1,1}  
		\\
		0 & 0 & \ell_{1,2}  & a_{2,2} 
		\\ 
		2\ell_{1,2}  & -2 \ell_{1,2} & 0 & a_{1,2} 
		\\
		0 & 0 & 0 & 0 
	\end{pmatrix}. 
\end{equation}
Where  $a_{1,1},  a_{2,2}, a_{1,2},  \ell_{1,2}   $
are real coefficients. We will comment on the isomorphism class of this
algebra in the next Subsection.
\label{ex:zero2}
\end{example}

 \subsection{The Lie algebra $\FradZero_N$ and the Euclidean
    Lie algebra}

In this Subsection we show that $\FradZero_N$ is connected with the
Euclidean algebra. First, recall that Euclidean Lie algebra
$\mathfrak{e}_{N}$ is a Lie algebra generated by the set:
\begin{equation}
\mathfrak{e}_{N} = \Span\Set{J_{i,j}}_{1\leq i<j\leq N}\cup\Set{\Pi_{i}}_{1\leq i\leq N}
    \label{eq:eNdef}
\end{equation}
with the following non-zero commutation relations:
\begin{equation}
    [J_{i,j},J_{k,l}]_{\mathfrak{e}_{N}} =
        J_{j,l}\delta_{i,k}+J_{k,j}\delta_{l,i}
        +J_{l,i}\delta_{j,k}+J_{i,k}\delta_{l,j},
        \quad
    [J_{i,j},\Pi_{k}]_{\mathfrak{e}_{N}} =
    \Pi_{j}\delta_{i,k}-\Pi_{i}\delta_{j,k},
    \label{eq:eNcomm}
\end{equation}
see, for example, \cite{Latini_2019}.
Let us now consider its universal enveloping algebra $U\mathfrak{e}_{N}$ as
the Lie algebra constructed by taking the tensor algebra and quotienting it
with respect to the commutation relations~\eqref{eq:eNcomm}, see~\cite[\S
V.17]{Humphreys1972.Introduction_Lie_Algebras_Representation_Theory}.  In
general, given $x,y\in U\mathfrak{e}_{N}$ then their Lie bracket is:
\begin{equation}
    [x,y]_{U\mathfrak{e}_{N}} = x y - y x.
    \label{eq:liebracketuea}
\end{equation}

\begin{prop}
    The Lie algebra $\FradZero_N$ is isomorphic to the Lie subalgebra
    \begin{equation}
        \mathfrak{q}_{N}
        \coloneqq
        \Span\Set{J_{i,j}}_{1\leq i<j\leq N}\cup
        \Set{\Pi_{i}\Pi_{j}}_{1\leq i\leq j\leq N}  
        \subset 
        U\mathfrak{e}_{N}
        \label{eq:qNdef}
    \end{equation}
    with the following isomorphism of Lie algebras $\upsilon\colon
    \FradZero_N\to\mathfrak{q}_N$:
    \begin{equation}
        L_{i,j} \mapsto J_{i,j},
        \qquad
        F_{i,j} \mapsto \Pi_{i}\Pi_{j}.
        \label{eq:upsiso}
    \end{equation}
\end{prop}

\begin{proof}
    We start by showing that the subset $\mathfrak{q}_{N}$ is indeed a Lie
    subalgebra of $U\mathfrak{e}_{N}$. Clearly $\mathfrak{q}_N$ is a closed
    linear space in $U\mathfrak{e}_{N}$. So, we have just to prove that the
    commutation relations of $\mathfrak{q}_{N}$ are closed in
    $U\mathfrak{e}_{N}$. From~\eqref{eq:eNcomm} we have
    $[J_{i,j}, J_{k,l}]_{U\mathfrak{e}_N} \in \Span\left\{ J_{i,j}
    \right\}_{1\leq i<j\leq N}\subset \mathfrak{q}_{N}$ and that the
    $\Pi_{i}$ are commutative, and so are the $\Pi_{i}\Pi_{j}$. This means
    that we have just to prove that
    $[J_{i,j},\Pi_{k}\Pi_{l}]\in\mathfrak{q}_{N}$. We have:
    \begin{equation}
        \begin{aligned}
            [J_{i,j},\Pi_{k}\Pi_{l}]_{U\mathfrak{e}_N}
            & = 
            (J_{i,j}\Pi_{k})\Pi_{l} - \Pi_{k}\Pi_{l}J_{i,j}
            \\
            &=(\Pi_{k}J_{i,j} + \Pi_{j}\delta_{i,k}-\Pi_{i}\delta_{j,k})\Pi_{l} - \Pi_{k}\Pi_{l}J_{i,j}
            \\
            &=\Pi_{k}(J_{i,j}\Pi_{l}) 
            + \Pi_{j}\Pi_{l}\delta_{i,k}-\Pi_{i}\Pi_{l}\delta_{j,k} - \Pi_{k}\Pi_{l}J_{i,j}
            \\
            &=\Pi_{k}(\Pi_{l}J_{i,j}
            +\Pi_{j}\delta_{i,l}-\Pi_{i}\delta_{j,l} ) 
            + \Pi_{j}\Pi_{l}\delta_{i,k}-\Pi_{i}\Pi_{l}\delta_{j,k} - \Pi_{k}\Pi_{l}J_{i,j}
            \\
            &=\Pi_{k}\Pi_{j}\delta_{i,l}-\Pi_{k}\Pi_{i}\delta_{j,l} 
            + \Pi_{j}\Pi_{l}\delta_{i,k}-\Pi_{i}\Pi_{l}\delta_{j,k},
        \end{aligned}
        \label{eq:qNcomm}
    \end{equation}
    which proves the statement. Now we observe that the commutation
    relations of $\mathfrak{q}_{N}$, namely\ \eqref{eq:eNcomm}
    and~\eqref{eq:qNcomm} are the same as $\FradZero_N$ upon identification
    of $J_{i,j}$ with $L_{i,j}$ and of $F_{i,j}$ with $\Pi_{i}\Pi_{j}$.
    This gives the isomorphism $\upsilon\colon \FradZero\to
    \mathfrak{q}_{N}$, and ends the proof.
\end{proof}

Let us now observe that we have $\dim \FradZero_N
=\dim\mathfrak{q}_{N}= N^{2}$, while $\dim \mathfrak{e}_{N} = N(N+1)/2$.
So, in general we have (for positive solutions) the equality:
\begin{equation}
    \dim \FradZero_N =\dim \mathfrak{e}_{N} +1 
    \quad
    \text{if and only if}
    \quad
    N=2.
    \label{eq:fradziso}
\end{equation}
In particular, we prove that for $N=2$ in fact $\FradZero_2$ and
$\mathfrak{e}_2\oplus \R$ are isomorphic, as noted for instance
in~\cite{Gonera_etal2021}. This is the content of the following lemma:

\begin{lemma}
    For $N=2$ we have the Lie algebra isomorphism
    $\FradZero_2\cong\mathfrak{q}_{2}\cong \mathfrak{e}_{2}\oplus \R$.
    \label{cor:e2}
\end{lemma}

\begin{proof}
    Let us consider the Lie algebra $\mathfrak{e}_{2}\oplus \R$, where 
    $\mathfrak{e}_2 = \Span (K_1, K_2, M )$ with commutation
    relations:
    \begin{align}
            [K_1, K_2]=0, \qquad 
            [K_1, M]= - K_2, \qquad 
            [K_2, M ] = K_1,
    \end{align}
    see also \eqref{eq:eNcomm}, and $\R=\Span(H)$, with $H$ being a
    central element.  Then, by explicit computation, we have that the
    desired Lie algebras isomorphism
    $\zeta_2\colon\FradZero_2\to\mathfrak{e}_2 \oplus \R $ is:
    \begin{align}
            F_{1,1} & \mapsto K_1-K_2+H,
            &
            F_{2,2} & \mapsto -K_1 +K_2 +H, 
            \\ 
            F_{1,2} & \mapsto -(K_1+K_2), 
            & 
            L_{1,2} & \mapsto -2 M.
    \end{align}
    This concludes the proof of the lemma.
\end{proof}

\subsection{Levi decomposition of $\FradZero_N$}
 
In this Subsection we  prove the first part of the
statement \ref{th:main:3} of the Main Theorem. In particular, we start from
the Levi decomposition. That is, we prove the following result:
\begin{theorem} \label{th:frad-zero}
A Levi decomposition of the Lie algebra $\FradZero_N$ is given by: 
\begin{equation}
	\FradZero_N  \cong \so_N (\R)     
		\mathop{\niplus}_\tau \R^\mathcal{N}, 
\end{equation} 
where $\R^\mathcal{N} = \Span  \{ F_{i,j} \}_{1\leq i \leq j\leq N}  $ is
an $\mathcal{N}$-dimensional abelian Lie algebra, and by notation
$\mathop{\niplus}\limits_\tau$ we emphasize the fact that the semidirect sum is
taken with respect to the map $\tau$, provided by an action of $\so_N (\R)$
onto $\R^\mathcal{N} $ by \eqref{eq:LiFLlm-Fradkin-zero}. 
\end{theorem} 
 
\begin{proof} 
    Observe that from~\eqref{eq:pb-Fradkin-zero}  it follows that $\FradZero_N$
    can be decomposed into the semidirect sum: 
    \begin{align}
                    \FradZero_N & \cong  
                            \Span \{L_{i,j}\}_{1\leq i <  j\leq N}       
                   \mathop{\niplus}_\tau 
                    \Span  \{ F_{i,j} \}_{1\leq i \leq  j\leq N} 
            %\notag \\ 
            %&
             \cong  \so_N (\R) \mathop{\niplus}_\tau \mathfrak{p}, 
    \end{align}
    with respect to the map $\tau $:
      \begin{align}\label{eq:Frad0-map-tau-first-def}
		\tau \colon  \so_N (\R) &\to \Der(\mathfrak{p} )
		\notag 
		\\ 
		 L_{i,j} & \mapsto \tau_{L_{i,j}},  
	\end{align}
	whose action by definition of a semidirect sum is provided by $ [\so_N (\R), \mathfrak{p}] $, \ie by  the commutator~\eqref{eq:LiFLlm-Fradkin-zero}: 
	\begin{equation}
		\tau_{L_{i,j}} F_{k,l} = [L_{i,j}, F_{k,l}].
	\end{equation} 
	
	 The Lie algebra $ \so_N (\R)$ is semisimple of  dimension
    $N(N-1)/2$. Given~\eqref{eq:FijFlm-Fradkin-zero}, it is clear that $
    \mathfrak{p}      $
    is a radical of $\FradZero_N$, and moreover it is an abelian Lie subalgebra of
    dimension $\mathcal{N}$. Since the commutation relations
    \eqref{eq:pb-Fradkin-zero} are already in the form: 
    \begin{align}
            [ \so_N (\R), \so_N (\R)] \subset  \so_N (\R), 
            \qquad 
            [\so_N (\R), \mathfrak{p}] \subset \mathfrak{p}, 
            \qquad 
            [\mathfrak{p}, \mathfrak{p} ] =0, 
    \end{align}
    one can conclude that  $\so_N (\R) \mathop{\niplus}\limits_\tau \mathfrak{p}$ is indeed a
    Levi decomposition, see~\cite[Chapter 6]{SnoblWinternitz2017book}.
   \end{proof} 

Let us notice that given $\Der (\mathfrak{p}) \cong  \Der(\R^\mathcal{N}) \cong \gl_\mathcal{N} (\R) $, one can rewrite \eqref{eq:Frad0-map-tau-first-def} as:  
      \begin{align}
		\tau \colon  \so_N (\R) &\to  \gl_{\mathcal{N}} (\R) 
		\notag 
		\\ 
		 L_{i,j} & \mapsto \tau_{L_{i,j}} =  \ad_{L_{i,j}} \big|_{ \R^\mathcal{N} },  
	\end{align}
	where $\ad_{L_{i,j}}  \big|_{ \R^\mathcal{N} }$ is an adjoint representation of $L_{i,j} $, restricted to the subalgebra $\R^\mathcal{N} \subset  \FradZero_N $. 
	
\subsection{Properties of the Lie algebra $\FradZero_N$ }	

As it was already discussed before, the original commutation relations \eqref{eq:pb-Fradkin-zero}  are not closed
with respect to the basis with $i \leq j $. Let us demonstrate it using a
particular example: 
\begin{example}
	Consider \eg $[L_{1,2}, F_{1,1} ] \mathop{=}\limits^{\eqref{eq:LiFLlm-Fradkin-zero} }
	F_{1,2} + F_{2,1} \mathop{=}\limits^\eqref{eq:ideal}  
	2F_{1,2} $. However, if we rewrite \eqref{eq:LiFLlm-Fradkin-zero} as: 
	\begin{equation}
		[L_{i,j},F_{k,l} ] = 
	\sum_{\alpha, \beta=1}^N 
	 s^{(\alpha \beta)}_{ (ij) (kl)} F_{\alpha, \beta} 
	\end{equation}
	and work only with $i\leq j $, we have that the  structure constant is:   
 	\begin{equation}\label{eq:Frad0-structure-constant-lemma-original}
		s^{(\alpha \beta)}_{ (ij) (kl)} =
		\delta_{i,l} \delta\indices{^\alpha _k}\delta\indices{^\beta _j}
		+\delta_{i,k} \delta\indices{^\alpha _j} \delta\indices{^\beta _l}
		-\delta_{j,k} \delta\indices{^\alpha _i} \delta\indices{^\beta _l}
		-\delta_{j,l} \delta\indices{^\alpha _k} \delta\indices{^\beta _i}. 
	\end{equation}
      However, to obtain a correct
        expression when acting on $\R^\mathcal{N} = \Span \{F_{i,j} \}_{1
        \leq i \leq j \leq N} $, we want to symmetrize the structure constant 
         $s^{(1,2)}_{(1,2)(1,1)}=1 $ 
        in  such a way to have  $\tilde{s}^{(1,2)}_{(1,2)(1,1)}=2 $ in this
        particular example, \ie to respect the condition \eqref{eq:ideal}. 
\end{example}
Now, let us find a general form of $\tilde s^{(\alpha \beta)}_{ (ij) (kl)}$. 
\begin{lemma} \label{lemma:frad0-new-symmetrized-structure-constant}
One can rewrite the commutation relations \eqref{eq:LiFLlm-Fradkin-zero} to
encode intrinsically the symmetry and antisymmetry conditions
\eqref{eq:ideal} as follows: 
\begin{equation}\label{eq:Frad0-structure-constant-lemma-new-sum}
			[L_{i,j},F_{k,l} ] = 
	\sum_{1 \leq \alpha \leq \beta \leq N }
	 \tilde{s}^{(\alpha \beta)}_{ (ij) (kl)} F_{\alpha, \beta},  
	\end{equation}
such that the sum runs now over $1\leq \alpha \leq \beta \leq N $ and the
symmetrized structure constant is given by:
	\begin{align}\label{eq:Frad0-structure-constant-lemma-modified}
		\tilde s^{(\alpha \beta)}_{ (ij) (kl)} \coloneqq  	(1- \delta\indices{^\alpha _\beta  })  
	&  \big( \delta_{i,l} \delta\indices{^\alpha _k}\delta\indices{^\beta _j}
		+\delta_{i,k} \delta\indices{^\alpha _j} \delta\indices{^\beta _l}
		-\delta_{j,k} \delta\indices{^\alpha _i} \delta\indices{^\beta _l}
		-\delta_{j,l} \delta\indices{^\alpha _k} \delta\indices{^\beta _i} \big)
		\notag 
	\\
	&+ \delta_{i,l} \delta\indices{^\alpha _j} \delta\indices{^\beta _k}
		+\delta_{i,k} \delta\indices{^\alpha _l} \delta\indices{^\beta _j} 
		-\delta_{j,k} \delta\indices{^\alpha _l} \delta\indices{^\beta _i} 
		-\delta_{j,l} \delta\indices{^\alpha _i}  \delta\indices{^\beta _k}. 
			\end{align}
\end{lemma}
\begin{proof}
 Essentially, we want to change the limits of the sum from $1 \leq \alpha,
 \beta \leq N $ to $1 \leq \alpha \leq \beta \leq N $, because it allows us
 to avoid the terms $F_{\beta, \alpha} $ with $ \beta > \alpha $ on the
 right-hand-side. Therefore, one can first rewrite the original sum as
 follows:
 	\begin{align}\label{eq:Frad0-lemma-structure-const-calc1}
 		 [L_{i,j},F_{k,l} ] &= 
	\sum_{\alpha, \beta=1}^N 
	 s^{(\alpha \beta)}_{ (ij) (kl)} F_{\alpha, \beta} 
	 \notag \\ 
	 & =  
	 \sum_{1\leq \alpha < \beta \leq N} 
	 s^{(\alpha \beta)}_{ (ij) (kl)} F_{\alpha, \beta}
	 +
	 \sum_{\gamma=1}^N 
	 s^{(\gamma \gamma )}_{ (ij) (kl)} F_{\gamma, \gamma}
	 + 
	 \sum_{1 \leq \beta < \alpha \leq N } 
	 s^{(\alpha \beta)}_{ (ij) (kl)} F_{\alpha, \beta}. 
 	\end{align}
In the last term of \eqref{eq:Frad0-lemma-structure-const-calc1}, let us rename the dummy summation indices as $ \alpha \leftrightarrow \beta $:
\begin{align}\label{eq:Frad0-lemma-structure-const-calc2}
		 [L_{i,j},F_{k,l} ] &=  \sum_{1\leq \alpha < \beta \leq N} 
	 s^{(\alpha \beta)}_{ (ij) (kl)} F_{\alpha, \beta}
	 +
	 \sum_{\gamma=1}^N 
	 s^{(\gamma \gamma )}_{ (ij) (kl)} F_{\gamma, \gamma}
	 + 
	 \sum_{1 \leq \alpha < \beta \leq N } 
	 s^{(\beta \alpha )}_{ (ij) (kl)} F_{ \beta, \alpha}
	 \notag \\
	 & =  \sum_{1 \leq \alpha < \beta \leq N }  
	 \left(  s^{(\alpha \beta)}_{ (ij) (kl)} +  s^{(\beta \alpha )}_{ (ij) (kl)} \right)  F_{\alpha, \beta} 
	 +  \sum_{\gamma=1}^N 
	 s^{(\gamma \gamma )}_{ (ij) (kl)} F_{\gamma, \gamma}, 
\end{align}
where we used a symmetry condition for $F_{\beta,\alpha} =F_{\alpha, \beta}$, given by \eqref{eq:ideal}. Finally, \eqref{eq:Frad0-lemma-structure-const-calc2} can be rewritten as a single sum if we take into account double counting of the term with $\alpha = \beta $, \ie\!:
\begin{align}
	\label{eq:Frad0-lemma-structure-const-calc3}
	 [L_{i,j},F_{k,l} ] &=  \sum_{1 \leq \alpha \leq  \beta \leq N}  
	 \left(  s^{(\alpha \beta)}_{ (ij) (kl)} +  s^{(\beta \alpha )}_{ (ij) (kl)} 
	 - \delta\indices{^\alpha _\beta  }  s^{(\alpha \beta)}_{ (ij) (kl)} 
	 \right)  F_{\alpha, \beta}   
	% \\
	% & 
	 = \sum_{1 \leq \alpha \leq  \beta \leq N}  
	 \left( (1- \delta\indices{^\alpha _\beta  }) s^{(\alpha \beta)}_{ (ij) (kl)} +  s^{(\beta \alpha )}_{ (ij) (kl)}  
	 \right)  F_{\alpha, \beta}.  
\end{align}
Hence, the symmetrized structure constants can be defined as:
\begin{equation}\label{eq:Frad0-lemma-structure-const-calc4}
		\tilde s^{(\alpha \beta)}_{ (ij) (kl)} \coloneqq (1- \delta\indices{^\alpha _\beta  }) s^{(\alpha \beta)}_{ (ij) (kl)} +  s^{(\beta \alpha )}_{ (ij) (kl)}. 
\end{equation}
Expanding this expression using \eqref{eq:Frad0-structure-constant-lemma-original}, one arrives at \eqref{eq:Frad0-structure-constant-lemma-modified}, thereby proving the lemma. 

\end{proof}

\subsection{Matrix representation of $\FradZero_N$}

Let us recall the following result from representation theory of Lie
algebras:

\begin{prop}[\cite{ghanamMinimalRepresentationsLie2017a}] \label{prop:ghanam}
    Suppose that the Lie algebra $\mathfrak{g} $ is a semidirect sum of a
    semisimple subalgebra $\mathfrak{a}$ and an $\mathcal{N}$-dimensional
    abelian ideal $\mathfrak{p}$ in the Levi decomposition. Then
    $\mathfrak{g}$ has a (faithful) representation as a subalgebra of
    $\gl_{\mathcal{N}+1} (\R)$. 
\end{prop}

In particular, \Cref{prop:ghanam} implies that $\FradZero_N$ has a
representation in $\gl_{\mathcal{N}+1}(\R)$. We now proceed to the
construction of such a matrix representation.  Let us first find a matrix
for an action of $\tau$:  

\begin{lemma}
    For an arbitrary $L \in \so_N (\R)$, an action  $\tau_L$ onto
    $\R^\mathcal{N} $ can be represented as an $\mathcal{N} \times
    \mathcal{N}$-matrix, whose matrix elements is given by:
    \begin{equation}\label{eq:frad0-matrix-element-of-action-lemma}
            \left( \tau_L \right)_{(\alpha, \beta) (k,l)} = 
            \sum_{1 \leq  i \leq j \leq N} \ell_{i,j}  \tilde s^{(\alpha, \beta )}_{ (ij) (k,l)},  
    \end{equation}
    where $\{ \ell_{i,j} \}_{1 \leq  i \leq j \leq N } $ are real coefficients representing the elements of $\so_N (\R) $, 
      and two-indices enumeration corresponds to taking elements $( F_{1,1},
     \ldots, F_{N,N} , F_{1,2},\ldots, F_{(N-1),N} ) \in~\R^\mathcal{N}$. 
\end{lemma}

\begin{proof}
	We have by definition: 
	\begin{equation}\label{eq:frad0-commutator-with-structure-const}
		\tau_{L_{i,j}} F_{k,l} = [L_{i,j},F_{k,l} ]. 
	\end{equation}
	
	A matrix of $ \tau_{L_{i,j}}$ 
	 can be obtained, as usual, by acting on $\R^\mathcal{N} $ with each $L_{i,j} $ one by one:
	\begin{equation}\label{eq:MatTauIntermediate}
		\Mat (\tau_{L_{i,j}}) = 
		\left(
			 \tau_{L_{1,2}}. \begin{bmatrix}
			 	F_{1,1} \\  \vdots \\ F_{(N-1),N}
			 	\end{bmatrix} \cdots 
			 \tau_{L_{(N-1),N}}. \begin{bmatrix}
			 	F_{1,1}\\ \vdots \\  F_{(N-1),N}
			 \end{bmatrix}
		\right). 
	\end{equation}
Therefore, a matrix element of \eqref{eq:MatTauIntermediate} is determined by a structure constants $ \tilde{s}^{(\alpha \beta)}_{ (ij) (kl)} $, calculated 
in \Cref{lemma:frad0-new-symmetrized-structure-constant}:
\begin{equation} \renewcommand{\arraystretch}{2}  
\Mat \left( \tau_{L_{i,j}}    \right)
= 
\left(
		\begin{array}{c}
			  \tilde s^{(1,1)}_{ (ij) (1,1)}
	 	  \\ 
	 	  \tilde s^{(2,2)}_{ (ij) (1,1)}
	 	  \\ 
	 	  \vdots 
	 	  \\ 
	 	  \tilde s^{(1,2)}_{ (ij) (1,1)}
	 	  \\ 
	 	  \vdots 
	 	  \\ 
	 	  \tilde s^{(N-1, N)}_{ (ij) (1,1)}
		\end{array}
		\begin{array}{c}
			  \tilde s^{(1,1)}_{ (ij) (2,2)}
	 	  \\ 
	 	  \tilde s^{(2,2)}_{ (ij) (2,2)}
	 	  \\ 
	 	  \vdots 
	 	  \\ 
	 	  \tilde s^{(1,2)}_{ (ij) (2,2)}
	 	  \\ 
	 	  \vdots 
	 	  \\ 
	 	  \tilde s^{(N-1, N)}_{ (ij) (2,2)}
		\end{array}
		\begin{array}{c}
			 \cdots 
	 	  \\ 
	 	   \cdots 
	 	  \\ 
	 	  \ddots 
	 	  \\ 
	 	  \cdots 
	 	  \\ 
	 	  \ddots 
	 	  \\ 
	 	  \cdots 
		\end{array}
		\begin{array}{c}
			  \tilde s^{(1,1)}_{ (ij) (N-1,N)}
	 	  \\ 
	 	  \tilde s^{(2,2)}_{ (ij) (N-1, N)}
	 	  \\ 
	 	  \vdots 
	 	  \\ 
	 	  \tilde s^{(1,2)}_{ (ij) (N-1, N)}
	 	  \\ 
	 	  \vdots 
	 	  \\ 
	 	  \tilde s^{(N-1, N)}_{ (ij) (N-1,N)}
		\end{array}
\right). 
\end{equation}

An arbitrary element  $L \in  \so_N (\R) $ is in the form:
\begin{equation}\label{eq:frad0-arb-so3}
	L = \sum_{1 \leq  i<j \leq N} \ell_{i,j} L_{i,j}, 
	\qquad 
	\ell_{i,j} \in \R. 
\end{equation} 
Therefore, an action:
\begin{equation}
	\tau_L = \sum_{1 \leq  i<j \leq N} \ell_{i,j} \tau_{L_{i,j}} 
\end{equation}
can be represented as: 
\begin{equation} \renewcommand{\arraystretch}{2}  \label{eq:frad0-matrix-element-of-action-full-matrix}
\Mat \left( \tau_{L}    \right)
= 
\left(
		\begin{array}{c}
			\sum \ell_{i,j}  \tilde s^{(1,1)}_{ (ij) (1,1)}
	 	  \\ 
	 	 \sum \ell_{i,j}  \tilde s^{(2,2)}_{ (ij) (1,1)}
	 	  \\ 
	 	  \vdots 
	 	  \\ 
	 	\sum \ell_{i,j}   \tilde s^{(1,2)}_{ (ij) (1,1)}
	 	  \\ 
	 	  \vdots 
	 	  \\ 
	 	\sum \ell_{i,j}   \tilde s^{(N-1, N)}_{ (ij) (1,1)}
		\end{array}
		\begin{array}{c}
		\sum \ell_{i,j} 	  \tilde s^{(1,1)}_{ (ij) (2,2)}
	 	  \\ 
	 	\sum \ell_{i,j}   \tilde s^{(2,2)}_{ (ij) (2,2)}
	 	  \\ 
	 	  \vdots 
	 	  \\ 
	 	\sum \ell_{i,j}   \tilde s^{(1,2)}_{ (ij) (2,2)}
	 	  \\ 
	 	  \vdots 
	 	  \\ 
	 	\sum \ell_{i,j}   \tilde s^{(N-1, N)}_{ (ij) (2,2)}
		\end{array}
		\begin{array}{c}
			 \cdots 
	 	  \\ 
	 	   \cdots 
	 	  \\ 
	 	  \ddots 
	 	  \\ 
	 	  \cdots 
	 	  \\ 
	 	  \ddots 
	 	  \\ 
	 	  \cdots 
		\end{array}
		\begin{array}{c}
		\sum \ell_{i,j} 	  \tilde s^{(1,1)}_{ (ij) (N-1,N)}
	 	  \\ 
	 	\sum \ell_{i,j}   \tilde s^{(2,2)}_{ (ij) (N-1, N)}
	 	  \\ 
	 	  \vdots 
	 	  \\ 
	 	\sum \ell_{i,j}   \tilde s^{(1,2)}_{ (ij) (N-1, N)}
	 	  \\ 
	 	  \vdots 
	 	  \\ 
	 	\sum \ell_{i,j}   \tilde s^{(N-1, N)}_{ (ij) (N-1,N)}
		\end{array}
\right),  
\end{equation}
where the sums are taken with respect to $1\leq i  < j \leq N $. A matrix \eqref{eq:frad0-matrix-element-of-action-full-matrix} can be indeed defined by its matrix element \eqref{eq:frad0-matrix-element-of-action-lemma}.  
\end{proof}
\begin{theorem}\label{th:Frad0-iso}
    An arbitrary element $X \in  \FradZero_N  $ can be represented in
    $\gl_{\mathcal{N}+1} (\R)  $ as follows: 
    \begin{equation}\renewcommand{\arraystretch}{1.5} \label{eq:frad0-block-representation}
    \Mat (X) = 
    \left[
            \begin{array}{c|c}
                    \Mat \left( \tau_{L}    \right) & \vec F 
                    \\ \hline 
                    \vec{0}^T & 0 
            \end{array}
    \right], 
    \end{equation}
    where an action $\Mat ( \tau_{L}) $ is given by
    \eqref{eq:frad0-matrix-element-of-action-full-matrix}, $\vec F $ is a
    coefficient vector representing $ \R^\mathcal{N} $ such that: 
\begin{equation}\label{eq:frad0-vector-F-def}
	\vec F^T \coloneqq 
	\begin{bmatrix}
		 a_{1,1} & \cdots   & a_{N,N} & a_{1,2} & \cdots &  a_{N-1,N}  
	\end{bmatrix} \in \R^\mathcal{N},  
\end{equation}
and $\vec 0$ is a zero vector of dimension $\mathcal{N}$. 
\end{theorem}
\begin{proof}
A block matrix \eqref{eq:frad0-block-representation} is constructed as a representation for a semidirect sum. The proof can be done using the rules of block matrices multiplication. 

First, it is trivial to check that $[\R^\mathcal{N}, \R^\mathcal{N}]=0$. Then, let us compute $[ \so_N(\R), \R^\mathcal{N}]$: 
\begin{equation}\renewcommand{\arraystretch}{1.5}
	\left[
		\begin{array}{c|c}
			 \Mat \left( \tau_{L_{i,j}}    \right)    & \vec 0   
			\\ \hline 
			\vec{0}^T & 0 
		\end{array}
	\right]
	. \left[
		\begin{array}{c|c}
			\mathbb{0}    & \vec F 
			\\ \hline 
			\vec{0}^T & 0 
		\end{array}
	\right] 
	- 
	\underbracket{ \left[
		\begin{array}{c|c}
			\mathbb{0}    & \vec F  
			\\ \hline 
			\vec{0}^T & 0 
		\end{array}
	\right].\left[
		\begin{array}{c|c}
			\Mat \left( \tau_{L_{i,j}} \right)     & \vec0 
			\\ \hline 
			\vec{0}^T & 0 
		\end{array}
	\right] }_{= \mathbb{0} }
	= \left[
		\begin{array}{c|c}
			\mathbb{0}    &  \Mat \left( \tau_{L_{i,j}} \right).  \vec F
			\\ \hline 
			\vec{0}^T & 0 
		\end{array}
	\right]. 
\end{equation}
Here  by construction  \eqref{eq:frad0-commutator-with-structure-const}, we have $ \Mat \left( \tau_{L_{i,j}} \right).  \vec F =\widetilde{\vec F}$, where  the components of the vector $\widetilde{\vec F} \in \R^\mathcal{N}$ are defined by the right-hand-side of~\eqref{eq:LiFLlm-Fradkin-zero}.

Finally, let us set all $a_{k,l}=0 $ to work in the subalgebra $[\so_N (\R), \so_N (\R) ] $. Consider the matrix of an adjoint representation $\Mat (\ad_{L_{i,j}} ) \in \gl_{N^2 }  (\R)  $. 
We remark that for each $L_{i,j} $ we have: 
\begin{equation}\renewcommand{\arraystretch}{1.8} 
	\Mat (\ad_{L_{i,j} } ) = 
	 \left[
		\begin{array}{c|c}
			 \Mat \left( \ad_{L_{i,j} } \big|_{\R^\mathcal{N}}  \right)  & \mathbb{0} 
			\\ \hline 
			\mathbb{0} & \Mat \left( \ad_{L_{i,j} } \big|_{\so_N (\R) } \right) 
		\end{array}
\right] 
	=\left[
		\begin{array}{c|c}
			\Mat  (\tau_{L_{i,j}} ) & \mathbb{0} 
			\\ \hline 
			\mathbb{0} & \Mat \left( \ad_{L_{i,j} } \big|_{\so_N (\R) } \right) 
		\end{array}
\right]. 
\end{equation}
Therefore, since for an adjoint representation the property: 
\begin{equation}
	\ad_{L_{i,j} } \ad_{L_{k,l} } - \ad_{L_{k,l} }\ad_{L_{i,j} } 
	= \ad_{[L_{i,j}, L_{k,l} ]}
\end{equation}
is satisfied automatically as a consequence of Jacobi identity, the same is guaranteed for $\Mat (\tau_{L_{i,j}})$, if we consider the commutator as a difference of products of two block-diagonal matrices. 
\end{proof} 
We also remark that the dimension of $\Mat(X)$ is $(\mathcal{N} +1 ) \times (\mathcal{N} +1 )   $, which agrees with~\Cref{prop:ghanam}.

\begin{example}[$N=3$]
	Let us discuss an example of applying the obtained formulas to $3$-dimensional case. Here $\mathcal{N} =6 $ and by~\Cref{th:frad-zero}:
	\begin{equation}
		 \FradZero_3  \cong \so_3 (\R)     
		\mathop{\niplus}_\tau \R^6.  
	\end{equation}
	The commutation relations for this algebra are summarized in~\Cref{table:frad0-N=3-comm-table}. 
	\begin{table}[h!]
	\[
		  \begin{array}{c||c|c|c|c|c|c|c|c|c|}
          %  \toprule
          \FradZero_3 & F_{1,1} & F_{2,2} & F_{3,3} & F_{1,2}  & F_{1,3} & F_{2,3} & L_{1,2}  & L_{1,3} & L_{2,3} 
            \\
            \midrule \midrule 
           F_{1,1}  & 0 & 0  &   0 & 0 & 0 & 0 & -2 F_{1,2} & -2 F_{1,3} & 0 
           \\ \midrule
           F_{2,2} & 0 & 0  &   0 & 0 & 0 & 0 & 2 F_{1,2} & 0 & -2 F_{2,3}
           \\ \midrule 
           F_{3,3} & 0 & 0  &   0 & 0 & 0 & 0 & 0 & 2 F_{1,3} & 2 F_{2,3}
           \\ \midrule 
           F_{1,2} & 0 & 0  &   0 & 0 & 0 & 0 & F_{1,1} - F_{2,2} & -F_{2,3} & -F_{1,3} 
           \\ \midrule 
           F_{1,3} & 0 & 0  &   0 & 0 & 0 & 0 & -F_{2,3} & F_{1,1} - F_{3,3} & F_{1,2}
           \\ \midrule 
           F_{2,3} & 0 & 0  &   0 & 0 & 0 & 0 & F_{1,3} & F_{1,2} & F_{2,2} - F_{3,3} 
           \\ \midrule 
           L_{1,2} & 2 F_{1,2} & - 2F_{1,2} & 0 &
            F_{2,2} - F_{1,1} & F_{2,3} & -F_{1,3} &  0 &  L_{2,3} & - L_{1,3}  
            \\ \midrule 
            L_{1,3} & 2F_{1,3} & 0 & -2F_{1,3} & F_{2,3} & F_{3,3} - F_{1,1} & - F_{1,2} & -L_{2,3} & 0 & L_{1,2}
             \\ \midrule  
             L_{2,3} & 0 & 2 F_{2,3} & -2 F_{2,3} & F_{1,3} & - F_{1,2} & F_{3,3}-F_{2,2} & L_{1,3} & - L_{1,2} & 0 
            \end{array}
        \]
\caption{$ \FradZero_3$ commutation table. }	
	\label{table:frad0-N=3-comm-table} 	
	\end{table}

Using the general formulas \eqref{eq:frad0-block-representation}, \eqref{eq:frad0-vector-F-def} and \eqref{eq:frad0-matrix-element-of-action-lemma}, 
one can explicitly represent an arbitrary element $X \in \FradZero_3 $ as follows: 
\begin{equation}
    \begin{aligned}
	\Mat (X) &= \sum_{i,j} a_{i,j} \Mat (F_{i,j})
	+ \sum_{i,j} \ell_{i,j} \Mat (L_{i,j})
        \\
	&= \begin{pmatrix}
0 & 0 & 0 & -\ell_{1,2} & -\ell_{1,3} & 0 & a_{1,1} 
\\
 0 & 0 & 0 & \ell_{1,2} & 0 & -\ell_{2,3} & a_{2,2} 
\\
 0 & 0 & 0 & 0 & \ell_{1,3} & \ell_{2,3} & a_{3,3} 
\\
 2 \ell_{1,2} & -2 \ell_{1,2} & 0 & 0 & -\ell_{2,3} & -\ell_{1,3} & a_{1,2} 
\\
 2 \ell_{1,3} & 0 & -2 \ell_{1,3} & \ell_{2,3} & 0 & -\ell_{1,2} & a_{1,3} 
\\
 0 & 2 \ell_{2,3} & -2 \ell_{2,3} & \ell_{1,3} & \ell_{1,2} & 0 & a_{2,3} 
\\
 0 & 0 & 0 & 0 & 0 & 0 & 0 
	\end{pmatrix}, 
    \end{aligned}
\end{equation}
where $\ell_{1,2}, \ell_{1,3}, \ell_{2,3},a_{1,1}, \dots, a_{2,3}$ are real coefficients. 
\end{example}

\newpage
\section{Conclusions and outlook}
\label{sec:concl}

In this paper, we investigated the parametric Lie algebra
${\FradPlain_{N}(\alpha)}$ that includes the symmetry algebras of
different oscillator models, both discrete and continuous. In particular,
it arises as the symmetry algebra of the  quadratic Hamiltonian
\eqref{eq:cont-Hamiltonian} (including IHO, repulsive IHO, free particle
model), the Darboux III oscillator \eqref{eq:defiho}, and a possible discretization
of the IHO, generalizing the Runge--Kutta one \eqref{eq:diho3}.

For the three cases $\alpha >0$, $\alpha <0 $ and $\alpha = 0$, we proved
the Lie algebra isomorphism maps to $\u_N$, $\gl_N(\R)$, and $\so_N (\R)
\niplus \mathbb{R}^\mathcal{N}$, respectively. For each case, we provided
the explicit isomorphism maps, summarized in \Cref{table:concl-summary}, along with explicit formulas for the Killing form derived from the corresponding Levi factor, obtained via the Levi decomposition. We emphasize that all formulas hold true for arbitrary $N$.  

\begin{table}[h!]  
	\renewcommand{\arraystretch}{1.5}  
    \centering
    \begin{tabular}{c|c|c|c|c}
    	 \toprule 
    	 Parameter & Lie algebra & Statement & Isomorphism &  Matrix representation 
    	 \\   \midrule  
    	 $\alpha >0 $ & $\Frad_N $ & \Cref{th:unfinal} &  $\Frad_{N} \cong \mathfrak{u}_N $ 
    	 &\Cref{eq:Frad-mat-repr-full-uN}
    	  \\   \midrule  
    	 $\alpha <0 $ & $\mFrad_N $ & \Cref{th:glNfinal} &  $\mFrad_{N}  \cong \mathfrak{gl}_N (\R) $ & \Cref{eq:mFrad-mat-repr-full-glN} 
    	  \\   \midrule  
    	 $\alpha =0 $ & $\FradZero_N $ & \Cref{th:Frad0-iso} &  $ \FradZero_N  \cong \so_N (\R)     
		\niplus  \R^\mathcal{N}$ & \Cref{eq:frad0-block-representation} 
		\\ \bottomrule 
     \end{tabular}
    \caption{Summary of the results, constituting the Main Theorem}
    \label{table:concl-summary}
\end{table}

 We  emphasize that, taking into account the low-dimensional isomorphisms
$\su_{1,1} \oplus  \R \cong \sl_{2} (\R) \oplus \R  \cong \gl_2 (\R)$, the
isomorphism with $ \gl_N (\R)$ in the non-compact case ($\alpha < 0$)
generalizes the one for $N=2$. This result was counterintuitive
because, in our view, it would have been more natural to expect a Lie algebra related to
an indefinite special unitary algebra $\su_{p,q}$.

From the technical point of view, we also mention that through the
aforementioned isomorphism for the compact case ($\alpha>0 $), the obtained commutation relations
\eqref{eq:frad-N-closed-Lie-bracket} might be convenient to describe the special
unitary algebra $\su_N$ itself for arbitrary $N$, see
\Cref{rmk:comrel-suN}.  In this regard, the new two-indices 
    basis $\{f_{i,j}\}_{1 \leq i,j \leq N} $, defined
    by~\eqref{eq:prop:more-general-transformation}, is analogous to
    introducing the spin variables, \eg in spin Calogero--Moser models, see
    the monograph~\cite[\S 7.1]{BabelonBernardTalon2003book} or the more
    recent papers~\cite{Reshetikhin2017}
    and references therein. In this particular case, adding the additional ``spin''
    degree of freedom allows us to calculate a closed-form expression for
    the structure constants \eqref{eq:structure-constant-in-fp}, and hence
    the unified commutation relations~\eqref{eq:frad-N-closed-Lie-bracket}.
      Although similar expressions for $\su_N$ were recently found in
\cite{bossionGeneralFormulasStructure2021} in another basis, as a byproduct of our results we obtained a new ``spin'' basis which might be more convenient to be used is some cases. This is because instead of all the many formulas for the structure constants $f_{ijk}$ and
    $d_{ijk}$ as appearing, for example, in \cite{bossionGeneralFormulasStructure2021}, we now have
    only one closed-form expression \eqref{eq:structure-constant-in-fp} for
    $c^{(\alpha\beta)}_{(ij)(kl)}$.  Together with the formula for the
    Killing form \eqref{eq:Killing-form-generalN-matrix-el-frad-minus}, it
    might be useful, in particular, for physical applications involving the
    Lie algebra $\su_N$ \cite{fuchsSymmetriesLieAlgebras2003}.  In the same
    fashion, we obtained the closed-form expressions for the structure constants
    \eqref{eq:structure-constant-fradNminus} and the Killing form
\eqref{eq:Killing-form-generalN-matrix-el-frad-minus} for the Lie  algebra 
$\sl_N (\R)$.

For the case $\alpha =0 $, we showed that the Lie algebra $ \FradZero_N $ is, in general, not isomorphic to the Euclidean Lie Algebra $\mathfrak{e}_N \cong \so_N (\R)  \niplus \R^N  $, but to the subalgebra of its universal enveloping algebra. The known isomorphism (see \Cref{cor:e2})  $\FradZero_2\cong \mathfrak{e}_{2}\oplus \R $ is  caused by the peculiar coincidence $N^2 = \mathcal{N}+1 $, which is true only for $N=2 $. We also remark that the method we used to construct an explicit matrix representation in \Cref{th:Frad0-iso} might be, in principle, generalized to  an arbitrary semidirect sum $\mathfrak{g} \niplus \R^k $.

Within this paper, we also showed that the IHO and its discretization admit a Nambu--Hamiltonian
structure. In principle, this can be extended to other MS systems
including, for example, the Darboux III model. This property could be
useful in finding MS discretizations of these models. 

\section*{Acknowledgements}

PD, GG, and DL  acknowledge the support of the research project Mathematical
Methods in NonLinear Physics (MMNLP), Gruppo 4-Fisica Teorica of INFN. The
authors' research was also partially supported by GNFM of the Istituto
Nazionale di Alta Matematica (INdAM).  PD also acknowledges the support from
the Ph.D. program of the Universit\`a degli Studi di Udine.  DL has been
partially funded by MUR - Dipartimento di Eccellenza 2023-2027, codice CUP
G43C22004580005 - codice progetto   DECC23$\_$012$\_$DIP.

\appendix

\section{Appendix: Some useful matrix identities}
\label{sec:appendix-matrices} 
\subsection{Matrices $\mathcal{A,B,C}$}
Let us define two families of $N \times N $ matrices, depending on parameters: 
\begin{align}\label{app:eq:A-B-def}
	\mathcal{A} (x, a, \vec v ) & \coloneqq  
	\begin{pmatrix}
		x & 0  & \cdots & 0 
		\\ 
		v_1 & a  & \cdots & 0 
		\\ 
		\vdots & \vdots   & \ddots 
		\\ 
		v_{N-1} & 0  & \cdots & a
	\end{pmatrix}, 
	& 
	\mathcal{B} (x, a, \vec v ) &\coloneqq  
	\begin{pmatrix}
	 0 & 0  & 0  & \cdots & 0 
	 \\ 
	0 & 	x & 0  & \cdots & 0 
		\\ 
	0& 	v_1 & a  & \cdots & 0 
		\\ 
		\vdots & \vdots  & \vdots  & \ddots & \vdots 
		\\ 
	0& 	v_{N-2} & 0  & \cdots & a
	\end{pmatrix}, 
\end{align}
where $x,a \in \R  $ and $\vec v$ is a vector with entries in $\R$:
\begin{equation}
	\vec v \coloneqq \begin{bmatrix}
		v_1 \\ \vdots \\ v_K
	\end{bmatrix} \in \mathbb{R}^K
\end{equation}
($K= N-1 $ for $\mathcal{A} $ and $ K =N-2 $ for $\mathcal{B}  $).  

Matrices of the shape \eqref{app:eq:A-B-def} form two subrings of the ring of lower-triangular matrices. Therefore, the linearity property holds: 
\begin{enumerate}[(i)]
	\item $\mathcal{A} (x, a, \vec{v}) 
	+ \mathcal{A} (y, b, \vec w) = 
      \mathcal{A} (x+y, a+b, \vec v + \vec w) $, 
    \item $  \mathcal{B} (x, a, \vec{v}) 
	+ \mathcal{B} (y, b, \vec w) = 
      \mathcal{B} (x+y, a+b, \vec v + \vec w) $. 
\end{enumerate}
Moreover, one can write down a simple multiplication rule: 
\begin{lemma}\label{app:lem-multiplication}
	Matrices $\mathcal{A}$ and $\mathcal{B}$ are subjects of the same multiplication rule:
	\begin{enumerate}[(i)]
		\item $\mathcal{A} (x, a, \vec v ). \mathcal{A}(y, b, \vec w) = 
		\mathcal{A}(xy, ab, a \vec w+ y \vec v)$, 
		 \item $\mathcal{B} (x, a, \vec v ). \mathcal{B}(y, b, \vec w) = 
		\mathcal{B}(xy, ab, a \vec w+ y \vec v)$. 
	\end{enumerate}
\end{lemma}
\begin{proof}
	Using the matrix multiplication rules, one can check in a straightforward manner that for the product of $\mathcal{A} (x, a, \vec v ) $ and $\mathcal{A}(y, b, \vec w) $: 
		\begin{align}
			\begin{pmatrix}
		x & 0  & \cdots & 0 
		\\ 
		v_1 & a  & \cdots & 0 
		\\ 
		\vdots & \vdots   & \ddots & \vdots 
		\\ 
		v_{N-1} & 0  & \cdots & a
	\end{pmatrix}.
		\begin{pmatrix}
		y & 0  & \cdots & 0 
		\\ 
		w_1 & b  & \cdots & 0 
		\\ 
		\vdots & \vdots   & \ddots & \vdots  
		\\ 
		w_{N-1} & 0  & \cdots & b
	\end{pmatrix}
	= 	\begin{pmatrix}
		xy & 0  & \cdots & 0 
		\\ 
		a w_1 + y v_1 & ab  & \cdots & 0 
		\\ 
		\vdots & \vdots   & \ddots & \vdots 
		\\ 
		a w_{N-1} + y v_{N-1} & 0  & \cdots & ab
	\end{pmatrix}. 
	\end{align}
Similarly for  $\mathcal{B} (x, a, \vec v )$ and $ \mathcal{B}(y, b, \vec w)$: 
\begin{align}
	\begin{pmatrix}
	 0 & 0  & 0  & \cdots & 0 
	 \\ 
	0 & 	x & 0  & \cdots & 0 
		\\ 
	0& 	v_1 & a  & \cdots & 0 
		\\ 
		\vdots & \vdots  & \vdots  & \ddots & \vdots 
		\\ 
	0& 	v_{N-2} & 0  & \cdots & a
	\end{pmatrix} .
	\begin{pmatrix}
	 0 & 0  & 0  & \cdots & 0 
	 \\ 
	0 & 	y & 0  & \cdots & 0 
		\\ 
	0& 	w_1 & b  & \cdots & 0 
		\\ 
		\vdots & \vdots  & \vdots  & \ddots & \vdots 
		\\ 
	0& 	w_{N-2} & 0  & \cdots & b 
	\end{pmatrix}
	= \begin{pmatrix}
	 0 & 0  & 0  & \cdots & 0 
	 \\ 
	0 & 	xy & 0  & \cdots & 0 
		\\ 
	0& y	v_1 + a w_1 & ab  & \cdots & 0 
		\\ 
		\vdots & \vdots  & \vdots  & \ddots & \vdots 
		\\ 
	0& 	y v_{N-2} + a w_{N-2} & 0  & \cdots & ab
	\end{pmatrix}.  
\end{align}
\end{proof}
Another problem might be the shape of the inverse. Matrix $\mathcal{B}$ is clearly degenerate, so its inverse does not exist, and $\Span \mathcal{B}  $ is not a division ring. However, for $\mathcal{A} $ we know that the inverse (if it exists) will be again lower triangular and, in fact, even a stronger statement holds: 
\begin{lemma}\label{app:lem-inverse}
	If $x,a \ne 0 $, the inverse of the matrix $\mathcal{A}$ exists and it is given by: 
	\begin{equation}
		\mathcal{A}^{-1} (x,a,v) = \mathcal{A} \left( \frac{1}{x}, \frac{1}{a}, - \frac{1}{a x} \vec v  \right). 
	\end{equation}
\end{lemma}
\begin{proof}
The diagonal elements of the inverse of a lower-triangular matrix are given by $d_i^{-1}$, where $d_i $ are the diagonal elements of the original matrix. 
	Moreover, one can easily see that: 
	\begin{equation}
			\begin{pmatrix}
		x & 0  & \dots & 0 
		\\ 
		v_1 & a  & \dots & 0 
		\\ 
		\vdots & \vdots   & \ddots 
		\\ 
		v_{N-1} & 0  & \dots & a
	\end{pmatrix}. 
	 	\begin{pmatrix}
		 \frac{1}{x} & 0  & \dots & 0 
		\\ 
		-\frac{v_1}{ax} & \frac{1}{a}  & \dots & 0 
		\\ 
		\vdots & \vdots   & \ddots 
		\\ 
		- \frac{v_{N-1}}{ax} & 0  & \dots & \frac{1}{a}
	\end{pmatrix}  = 
	\begin{pmatrix}
		1 & \dots & 0 
		\\ 
		\vdots & \ddots & \vdots 
		\\ 
		0 & \dots & 1
	\end{pmatrix}. 
	\end{equation}
\end{proof}
\begin{corollary}\label{app:lemm-Am1A}
	We have, by \Cref{app:lem-multiplication,app:lem-inverse}: 
	\begin{align}
		\mathcal{A}^{-1} (x,a,v) . 
		\mathcal{A} (y, b,w) 
		= \mathcal{A} \left( \frac{y}{x}, \frac{b}{a}, \frac{1}{a} \left( \vec w - \frac{y}{x} \vec v  \right)   \right). 
	\end{align}

\end{corollary}
One can see that a family $\mathcal{B}$ is the right ideal of a ring of lower-triangular matrices, so 
$\mathcal{A} . \mathcal{B} \in \Span \mathcal{B}$. However, to state a multiplication rule for $\mathcal{B}.\mathcal{A} $, we have to introduce the third family of matrices as: 
\begin{align}\label{app:eq:matrix-C-def}
	\mathcal{C} (a, \vec v, \vec w) 
	\coloneqq 
	\begin{pmatrix}
		v_1 & w_1 & 0 & \cdots  & 0 
		\\ 
		v_2 & w_2 & 0 & \cdots & 0 
		\\ 
		v_3 & w_3 & a & \cdots & 0 
		\\ 
		\vdots & \vdots & \vdots & \ddots & \vdots 
		\\ 
		v_N & w_N & 0 & \cdots & a
	\end{pmatrix}. 
\end{align}
We remark that for $w_1 \ne 0 $, $\mathcal{C}  $ is not lower-triangular. Note that, obviously, we again have:
\begin{equation}\label{app:eq:linearity-of-C}
	\mathcal{C} ( \lambda a, \lambda \vec v, \lambda \vec w) 
	= \lambda \mathcal{C} (a, \vec v, \vec w). 
\end{equation}
 Now we are ready to write down the rules of multiplication $\mathcal{A} $ and $\mathcal{B} $ both from the left and right: 
\begin{lemma}\label{app:lemm-prods-AB-BA}
	For the products between $\mathcal{A}$ and $\mathcal{B}$ we have the following expressions:
	\begin{enumerate}[(i)]
		\item $\mathcal{A}(x,a,\vec v) .\mathcal{B}(y,b,\vec w) = \mathcal{B} (ay, ab, a \vec w )  $, \ie $\Span \mathcal{B} $ form an ideal, 
		\label{app:lemm-prods-AB-BA-i}
		\item $\mathcal{B}(x,a,\vec v).
		\mathcal{A}(y,b,\vec w) =  
		\mathcal{C} 
		\left( ab, 
		 \vec V, 
		\vec W
		\right) 
		   $, 
		 \label{app:lemm-prods-AB-BA-ii}
	\end{enumerate}
	where 
	\begin{align}
		\vec V  =  
		\begin{bmatrix}
			 0 \\ x w_1 \\ a w_2 + w_1 v_1 
			 \\ 
			 a w_3 + w_1 v_2 
			 \\
			 \vdots 
			 \\ 
			 a w_{N-1} + w_1 v_{N-2}
		\end{bmatrix}, 
		\quad
		\vec W   =  
		\begin{bmatrix}
			0 \\ bx  \\
			b v_1 \\ b v_2 \\
			 \vdots 
			\\ b v_{N-2}
		\end{bmatrix}.  
	\end{align} 
\end{lemma}
\begin{proof}
		Indeed, one checks directly that the left-hand-side of \ref{app:lemm-prods-AB-BA-i} leads to: 
	\begin{align}
			\begin{pmatrix}
		x & 0  & \dots & 0 
		\\ 
		v_1 & a  & \dots & 0 
		\\ 
		\vdots & \vdots   & \ddots 
		\\ 
		v_{N-1} & 0  & \dots & a
	\end{pmatrix} .  
	\begin{pmatrix}
	 0 & 0  & 0  & \dots & 0 
	 \\ 
	0 & 	y & 0  & \dots & 0 
		\\ 
	0& 	w_1 & b  & \dots & 0 
		\\ 
		\vdots & \vdots  & \vdots  & \ddots & \vdots 
		\\ 
	0& 	w_{N-2} & 0  & \dots & b 
	\end{pmatrix} 
	= \begin{pmatrix}
	 0 & 0  & 0  & \dots & 0 
	 \\ 
	0 & 	a y & 0  & \dots & 0 
		\\ 
	0& 	a w_1 & a b  & \dots & 0 
		\\ 
		\vdots & \vdots  & \vdots  & \ddots & \vdots 
		\\ 
	0& a 	w_{N-2} & 0  & \dots & a b 
	\end{pmatrix}. 
	\end{align}
	In turn, the left-hand-side of \ref{app:lemm-prods-AB-BA-ii} gives us: 
	\begin{align}
	\begin{pmatrix}
	 0 & 0  & 0  & \dots & 0 
	 \\ 
	0 & 	x & 0  & \dots & 0 
		\\ 
	0& 	v_1 & a  & \dots & 0 
		\\ 
		\vdots & \vdots  & \vdots  & \ddots & \vdots 
		\\ 
	0& 	v_{N-2} & 0  & \dots & a 
	\end{pmatrix}  
	\begin{pmatrix}
		y & 0  & \dots & 0 
		\\ 
		w_1 & b  & \dots & 0 
		\\ 
		\vdots & \vdots   & \ddots 
		\\ 
		w_{N-1} & 0  & \dots & b
	\end{pmatrix}
	=
		\begin{pmatrix}
		0 & 0 & 0 & \cdots  & 0 
		\\ 
		x w_1 & bx & 0 & \cdots & 0 
		\\ 
		a w_2 + w_1 v_1 & b v_1 & ab & \cdots & 0 
		\\
		a w_3 + w_1 v_2 & b v_2 & 0 & \dots & 0
		\\ 
		\vdots & \vdots & \vdots & \ddots & \vdots 
		\\ 
		a w_{N-1} + w_1 v_{N-2} & b v_{N-2} & 0 & \cdots & ab
	\end{pmatrix}. 
	\end{align}
\end{proof}
Let us now find the determinants of all three matrices defined above:
\begin{lemma}\label{app:lemma-dets-summarized}
	The determinants of matrices $\mathcal{A,B}$, and $\mathcal{C} $  can be calculated as follows:
	\begin{enumerate}[(i)]
		\item $\det \mathcal{A} (x,a,v) = x \cdot a^{N-1} $,  
		\label{app:lemma-dets-summarized-i}
		\item $\det \mathcal{B} (x,a,v)=0 $,
		\label{app:lemma-dets-summarized-ii}
		\item $\det \mathcal{C} (a, \vec v, \vec w)  = 
		a^{N-2} (v_1 w_2 - w_1 v_2)$.
	    \label{app:lemma-dets-summarized-iii} 
	\end{enumerate} 
\end{lemma}
\begin{proof}
  The formulas \ref{app:lemma-dets-summarized-i} and \ref{app:lemma-dets-summarized-ii} follow from the well-known fact that the determinant of a triangular matrix is a product of its diagonal elements. 
  
  To calculate the determinant of $\mathcal{C}$, one can use the general formula for the determinant of the block matrix~\cite{abadir2005matrix} (given $A$ invertible): 
\begin{equation}\label{eq:block-matrices-formula-general-app}
	\det  \left[\begin{array}{  c | c   }
    A   & B \\ \hline 
     C  & D     
   \end{array}\right] = \det (A) \det ( D - C A^{-1} B). 
\end{equation}
In our case: 

\begin{equation}
	\renewcommand{\arraystretch}{1.5} 
	 \det 	\mathcal{C} (a, \vec v, \vec w)= 
	\left[\begin{array}{  c | c   }
    \begin{matrix}
    	v_1 & w_1 
    	\\ 
    	v_2 & w_2 
    \end{matrix}    & 
    \begin{matrix}
    	0 & \cdots & 0
    	\\
    	0 & \cdots & 0 
    \end{matrix}  \\ \hline 
     \begin{matrix}
     	v_3 & w_3 
     	\\ 
     	\vdots & \vdots 
     	\\
     	v_N & w_N
     \end{matrix}   & 
     \begin{matrix}
     	a & \cdots & 0  
     	\\ 
     	\vdots & \ddots & \vdots 
     	\\ 
     	0 & \cdots & a
     \end{matrix}     
   \end{array}\right].
\end{equation}
Here $B$ is the zero block, therefore: 
\begin{align}
	 \det 	\mathcal{C} (a, \vec v, \vec w) 
	 = \det (A) \det  (D) 
	 = a^{N-2}  (v_1 w_2 - w_1 v_2),
\end{align} 
which concludes the proof of the statement. 
\end{proof}  

\subsection{Matrices $\breve{\mathcal{B}}_i$ and $\breve{\mathcal{C}}_i$}
By construction, the first row of the matrix $\mathcal{B}$ defined in \eqref{app:eq:A-B-def}, consists of all zeros. We want to place the one ``floating'' non-zero element to the $ i$th column, let us set it equal to $1$. For this purpose, one can just sum it with an elementary matrix $e_{1,i}$ (see~\eqref{eq:elementary-matrices-definition}): 
\begin{align}\label{app:eq-breveBi-def}
	\breve{\mathcal{B}}_i(x, a, \vec v )    \coloneqq \mathcal{B} (x, a, \vec v ) 
	+ e_{1,i}
	= \begin{pNiceMatrix}
	 0 & 0  & 0  & \Cdots & 1 & \Cdots  & 0 
	 \\ 
	0 & 	x & 0  & \Cdots & 0 & \Cdots & 0 
		\\ 
	0& 	v_1 & a  & \Cdots & 0  & \Cdots & 0 
		\\ 
		\Vdots & \Vdots  & \Vdots  & \Ddots &  &  & \Vdots 
		\\
			\Vdots & \Vdots  &   &   & a &  &\Vdots 
				\\
			\Vdots & \Vdots  & \Vdots  &   &  & \Ddots  & \Vdots 
		\\ 
	0& 	v_{N-2} & 0  & \Cdots & 0   & \Cdots  &  a
	\end{pNiceMatrix}.
\end{align}
Similarly, despite the fact that the matrix $\mathcal{C}$ (see \eqref{app:eq:matrix-C-def}) contains two (in general) non-zero elements $v_1 $ and $w_1 $ in the first row, we may want to generalize it in the same fashion (assuming $i>2 $, otherwise we fall back to the previous case): 
\begin{align}\label{app:eq-breveCi-def}
	\breve{\mathcal{C}}_i(x, a, \vec v )    \coloneqq \mathcal{C} (x, a, \vec v ) 
	+ e_{1,i}
	= \begin{pNiceMatrix}
	 v_1 & w_1  & 0  & \Cdots & 1 & \Cdots  & 0 
	 \\ 
	v_2 & 	w_2 & 0  & \Cdots & 0 & \Cdots & 0 
		\\ 
	v_3 & 	w_3 & a  & \Cdots & 0  & \Cdots & 0 
		\\ 
		\Vdots & \Vdots  & \Vdots  & \Ddots &  &  & \Vdots 
		\\
			\Vdots & \Vdots  & \Vdots  &   & a &  &\Vdots 
				\\
			\Vdots & \Vdots  & \Vdots  &   &  & \Ddots  & \Vdots 
		\\ 
	v_N & 	w_{N} & 0  & \Cdots & 0   & \Cdots  &  a
	\end{pNiceMatrix},
	\qquad i>2. 
\end{align}
Let us calculate its determinant:
\begin{lemma}\label{app:lemma-detCi}
For $i>2 $: 
\begin{align}
		\det \breve{\mathcal{C}}_i (a, \vec v, \vec w) &=
		 a^{N-2} \left( v_1 w_2 - w_1 v_2 \right)
	+ a^{N-3} \left( 
	 w_i v_2 - v_i w_2 \right). 
\end{align}
\end{lemma}
\begin{proof}
One can perform a Gaussian elimination to remove $1$ from the first row. For this purpose, one has to multiply the $i$th row by $ - \dfrac{1}{a} $ (assuming $a$ non-zero) and add it to the first one:  
\begin{align}
	\det \breve{\mathcal{C}}_i (a, \vec v, \vec w) &= 
	\det \begin{pmatrix}
		v_1 & w_1 & 0 & \cdots & 1 & \cdots  & 0 
		\\ 
		v_2 & w_2 & 0 & \cdots & 0  & \cdots  & 0
		\\ 
		v_3 & w_3 & a & \cdots & 0  & \cdots  & 0
		\\ 
		\vdots & \vdots & \vdots & \ddots & \vdots & \ddots & \vdots 
		\\
		v_i & w_i & 0 & \cdots & a 
		& \cdots & 0 
		\\ 
		\vdots & \vdots & \vdots & \ddots & \vdots  & \ddots & \vdots 
		\\ 
		v_N & w_N & 0 & \cdots & 0 & \cdots & a
	\end{pmatrix}_{N \times N}  
	= \det 
	\begin{pmatrix}
		v_1 - \frac{v_i}{a} & w_1 - \frac{w_i}{a} & 0 & \cdots & 0 & \cdots  & 0 
		\\ 
		v_2 & w_2 & 0 & \cdots & 0  & \cdots  & 0
		\\ 
		v_3 & w_3 & a & \cdots & 0  & \cdots  & 0
		\\ 
		\vdots & \vdots & \vdots & \ddots & \vdots & \ddots & \vdots 
		\\
		v_i & w_i & 0 & \cdots & a 
		& \cdots & 0 
		\\ 
		\vdots & \vdots & \vdots & \ddots & \vdots  & \ddots & \vdots 
		\\ 
		v_N & w_N & 0 & \cdots & 0 & \cdots & a
	\end{pmatrix}_{N \times N}  
	\notag 
	\\
	&= \det \mathcal{C} \left( 
	a, 
	\begin{bmatrix}
		v_1 - \frac{v_i}{a} 
		\\ 
		v_2 
		\\
		\vdots 
		\\ 
		v_N
	\end{bmatrix}, 
	\begin{bmatrix}
		w_1 - \frac{w_i}{a} 
		\\ 
		w_2 
		\\
		\vdots 
		\\ 
		w_N
	\end{bmatrix}
	\right) 
	\notag
	= a^{N-2} \cdot  \left( \left(v_1 - \frac{v_i}{a} \right) w_2 
	- \left(w_1 - \frac{w_i}{a} \right) v_2   \right) 
	\notag 
	\\ 
	&= a^{N-2} \left( v_1 w_2 - w_1 v_2 \right)
	+ a^{N-3} \left( 
	 w_i v_2 - v_i w_2 \right),
\end{align}
where we used \Cref{app:lemma-dets-summarized}, \ref{app:lemma-dets-summarized-iii}.  
\end{proof} 
We will now prove a series of lemmas involving matrices $\breve{\mathcal{B}}_i $ and $\breve{\mathcal{C}}_i $, but first let us formulate an auxiliary statement for matrix $\mathcal{A}$: 
\begin{lemma}\label{app:lemm:e1iA}
	Consider elementary matrix $e_{1,i}$. Then the following product rules hold: 
	\begin{enumerate}[(i)]
		\item $e_{1,1} \mathcal{A}(x,a,v) = x \cdot e_{1,1}= \mathcal{A} (x, 0,\vec 0) $, 
		\item $e_{1,j} \mathcal{A}(x,a,v) = v_{j-1} e_{1,1} + a e_{1,j} $ for $j>1 $. 
	\end{enumerate}
	Note that if we set $v_0 \coloneqq x -a $, this can be written as a single property: 
	\begin{align}
		e_{1,j} \mathcal{A}(x,a,v) = v_{j-1} e_{1,1} + a e_{1,j}
	\end{align}
\end{lemma}
\begin{proof}
	A direct computation shows that: 
\begin{align}
	\begin{pmatrix}
		1 & \cdots & 0 & \cdots & 0 
		\\ 
		\vdots & \ddots & \vdots & \cdots  & 0 
		\\ 
		0 & \cdots & 0 & \cdots & 0
	\end{pmatrix}. \begin{pmatrix}
		x & 0  & \dots & 0 
		\\ 
		v_1 & a  & \dots & 0 
		%\\
	%	v_2 & 0 & a & \dots  & 0 
		\\ 
		\vdots & \vdots   & \ddots 
		\\ 
		v_{N-1} & 0  & \dots & a
	\end{pmatrix}
	&= \begin{pmatrix}
		x & 0  & \dots & 0 
		\\ 
		0 & 0  & \dots & 0 
		%\\
	%	v_2 & 0 & a & \dots  & 0 
		\\ 
		\vdots & \vdots   & \ddots 
		\\ 
		0 & 0  & \dots & 0
	\end{pmatrix}, 
	\\ 
	\begin{pmatrix}
		0 & \cdots & 1 & \cdots & 0 
		\\ 
		\vdots & \ddots & \vdots & \cdots  & 0 
		\\ 
		0 & \cdots & 0 & \cdots & 0
	\end{pmatrix}. \begin{pmatrix}
		x & 0  & \dots & 0 
		\\ 
		v_1 & a  & \dots & 0 
		\\ 
		\vdots & \vdots   & \ddots 
		\\ 
		v_{N-1} & 0  & \dots & a
	\end{pmatrix}
	& = 
	\begin{pmatrix}
		v_{j-1} & \cdots & a & \cdots & 0 
		\\ 
		\vdots & \ddots & \vdots & \cdots  & 0 
		\\ 
		0 & \cdots & 0 & \cdots & 0
	\end{pmatrix}. 
\end{align}
\end{proof}
Then we proceed with:
\begin{lemma}
	The following product rule holds:
	\begin{equation}\label{app:eq:lemm-BreveBiA}
		\breve{\mathcal{B}}_i (z, c, \vec u) 
		 .\mathcal{A} (t, d, \vec f ) = d \cdot  \breve{\mathcal{C}}_i 
		\left( c, 
		 \vec F , 
		\vec U
		\right),
	\end{equation}
	where 
	\begin{align}
		\vec F  = 
		\frac{1}{d}
		\begin{bmatrix}
			 0 \\ z f_1 \\ c f_2 + f_1 u_1 
			 \\ 
			 c f_3 + f_1 u_2 
			 \\
			 \vdots 
			 \\ 
			 c f_{N-1} + f_1 u_{N-2}
		\end{bmatrix}, 
		\quad 
		\vec U  = \begin{bmatrix}
			1 \\  z  \\
			  u_1 \\    u_2 \\
			 \vdots 
			\\   u_{N-2}
		\end{bmatrix}.  
	\end{align}
\end{lemma}

\begin{proof}
We expand the product by definition \eqref{app:eq-breveBi-def}: 
\begin{align}
	\breve{\mathcal{B}}_i (z, c, \vec u) 
		 .\mathcal{A} (t, d, \vec f )& = 
		 \left[ \mathcal{B} (z, c, \vec u ) 
	+ e_{1,i}
		 \right]  .\mathcal{A} (t, d, \vec f ) 
		 =  \mathcal{B} (z, c, \vec u ) 
		 \mathcal{A} (t, d, \vec f ) 
		 + e_{1,i}  \mathcal{A} (t, d, \vec f ).  
\end{align}
Then, by \cref{app:lemm-prods-AB-BA}, \ref{app:lemm-prods-AB-BA-ii} and \cref{app:lemm:e1iA}: 
\begin{align}
	\breve{\mathcal{B}}_i (z, c, \vec u) 
		 .\mathcal{A} (t, d, \vec f )& = 
		  \mathcal{C} 
		\left( cd, 
		\begin{bmatrix}
			 0 \\ z f_1 \\ c f_2 + f_1 u_1 
			 \\ 
			 c f_3 + f_1 u_2 
			 \\
			 \vdots 
			 \\ 
			 c f_{N-1} + f_1 u_{N-2}
		\end{bmatrix}, 
		\begin{bmatrix}
			0 \\ d\cdot z  \\
			d\cdot  u_1 \\ d \cdot  u_2 \\
			 \vdots 
			\\ d \cdot  u_{N-2}
		\end{bmatrix}
		\right)
		+ f_{i-1} e_{1,1} + d \cdot e_{1,i}.
\end{align}
Using eqs.~\eqref{app:eq:linearity-of-C} and \eqref{app:eq-breveCi-def}, one can finally write:  
\begin{align}
	\breve{\mathcal{B}}_i (z, c, \vec u) 
		 .\mathcal{A} (t, d, \vec f )& 
	= d \cdot  \breve{\mathcal{C}}_i 
		\left( c, 
		\frac{1}{d}
		\begin{bmatrix}
			 f_{i-1} \\ z f_1 \\ c f_2 + f_1 u_1 
			 \\ 
			 c f_3 + f_1 u_2 
			 \\
			 \vdots 
			 \\ 
			 c f_{N-1} + f_1 u_{N-2}
		\end{bmatrix}, 
		\begin{bmatrix}
			0 \\  z  \\
			  u_1 \\    u_2 \\
			 \vdots 
			\\   u_{N-2}
		\end{bmatrix}
		\right), 
\end{align}
which coincides with \eqref{app:eq:lemm-BreveBiA}. 
\end{proof}
If we take difference (or sum) of matrices $ {\mathcal{B}} $ and $\breve{\mathcal{C}}_i $, we will obtain again a matrix of type $\breve{\mathcal{C}}_i $, with respect to the following arguments: 
\begin{lemma}\label{app:lemm:Ci+B}
 \begin{align}
 	 \mathcal{B} ( y, c, \vec u) - \breve{\mathcal{C}}_i (a, \vec v, \vec w)  =  \breve{\mathcal{C}}_i \left(-a+c, 
	 -\vec v, - \vec Y \right),  
 \end{align}	
 where 
 \begin{align}
 	\vec Y = \begin{bmatrix}
	 	-w_1 \\ -w_2 + y \\ -w_3+ u_1 \\ \vdots \\ -w_N + u_{N-2}
	 \end{bmatrix}. 
 \end{align}
\end{lemma}
\begin{proof}
	Directly: 
	\begin{align}
	 \mathcal{B} ( y, c, \vec u) - \breve{\mathcal{C}}_i (a, \vec v, \vec w) & = 
	  -\mathcal{C} (a, \vec v, \vec w) + \mathcal{B} ( y, c, \vec u) - e_{1,i} 
	  \notag \\
	  & =-\begin{pmatrix}
		v_1 & w_1 & 0 & \cdots  & 0 
		\\ 
		v_2 & w_2 & 0 & \cdots & 0 
		\\ 
		v_3 & w_3 & a & \cdots & 0 
		\\ 
		\vdots & \vdots & \vdots & \ddots & \vdots 
		\\ 
		v_N & w_N & 0 & \cdots & a
	\end{pmatrix} 
	+ 
	\begin{pmatrix}
	 0 & 0  & 0  & \dots & 0 
	 \\ 
	0 & 	y & 0  & \dots & 0 
		\\ 
	0& 	u_1 & c  & \dots & 0 
		\\ 
		\vdots & \vdots  & \vdots  & \ddots & \vdots 
		\\ 
	0& 	u_{N-2} & 0  & \dots & c
	\end{pmatrix} - e_{1,i}  
	\notag
	\\
	& = 
	\begin{pmatrix}
		-v_1 & -w_1 & 0 & \cdots  & 0 
		\\ 
		-v_2 & -w_2 + y & 0 & \cdots & 0 
		\\ 
		-v_3 & -w_3 + u_1 & -a+c & \cdots & 0 
		\\ 
		\vdots & \vdots & \vdots & \ddots & \vdots 
		\\ 
		-v_N &- w_N + u_{N-2} & 0 & \cdots &- a+ c 
	\end{pmatrix}
	 + e_{1,i}  
	 \notag  \\
	 &= \breve{\mathcal{C}}_i \left(-a+c, 
	 -\vec v, \begin{bmatrix}
	 	-w_1 \\ -w_2 + y \\ -w_3+ u_1 \\ \vdots \\ -w_N + u_{N-2}
	 \end{bmatrix}
	 \right),
\end{align}
where we used the definition of $\breve{\mathcal{C}}_i$, \ie eq.~\eqref{app:eq-breveCi-def}.  
\end{proof}
\subsection{Final lemma}
Using all the results we formulated in \Cref{sec:appendix-matrices}, let us now present a result, useful for the proof of \Cref{th:Nambu-mechanics-f1-fN}: 
\begin{lemma}[Final determinant]\label{app:lemm:final-det}
\begin{equation}\label{app:lemm:final-det-eq} 
\resizebox{1\textwidth}{!}{$
\begin{aligned}
\det \left[ \mathcal{B} (t, d, \vec r) -  \breve{\mathcal{B}}_i (z, c, \vec u )
		\mathcal{A}^{-1} ( x,a,\vec v ) 
		\mathcal{A} (y,b,\vec w) \right] =
		&    - \left(- \frac{b}{a} \right)^N \cdot  \bigg[ 
 \left(c-\frac{a}{b}\cdot d\right)^{N-2}  \left(w_{i-1}-\frac{y}{x}\cdot v_{i-1}\right) \left(z-\frac{a}{b}\cdot t\right)
% \notag 
 \\ 
 &  
  + \left(c-\frac{a}{b}\cdot d\right)^{N-3} \cdot 
 \bigg\{
 z \left(u_{i-2}-\frac{a}{b}\cdot r_{i-2}\right)\cdot \left(w_{1}-\frac{y}{x}\cdot v_{1}\right)
% \notag 
 \\  
 & -\left(c\cdot \left(w_{i-1}-\frac{y}{x}\cdot v_{i-1}\right)+\left(w_{1}-\frac{y}{x}\cdot v_{1}\right)\cdot  {u_{i-2}}\right)\cdot \left(z-\frac{a}{b}\cdot t\right)
 \bigg\}  
  \bigg]. 
\end{aligned}  
$
}
\end{equation}
\end{lemma}
\begin{proof}
		Using \Cref{app:lemm-Am1A} and~\Cref{app:lemm-prods-AB-BA}, we rewrite 
	\begin{align}
			  \breve{\mathcal{B}}_i (z, c, u )
		\mathcal{A}^{-1} ( x,a,v ) 
		\mathcal{A} (y,b,w) 
		&=   
		 \breve{\mathcal{B}}_i (z, c, u ). 
		 \mathcal{A} \left( \frac{y}{x}, \frac{b}{a}, \frac{1}{a} \left( \vec w - \frac{y}{x} \vec v  \right)   \right)  
		\notag
		\\
		&=  \frac{b}{a}  \breve{\mathcal{C}}_i 
		\left( c , 
		\begin{bmatrix}
			\frac{1}{b} \left(w_{i-1} - \frac{y}{x} v_{i-1} \right)
			   \\ \frac{z}{b} \left(w_1 - \frac{y}{x} v_1 \right) 
			\\ 
			\frac{c}{b} \left(w_2 - \frac{y}{x} v_2 \right)
			+  \left(w_1 - \frac{y}{x} v_1 \right)
			\frac{u_1}{b}
			\\ 
			\frac{c}{b} \left(w_3- \frac{y}{x} v_3 \right)
			+  \left(w_1 - \frac{y}{x} v_1 \right)
			\frac{u_2}{b}
			\\
			\vdots 
			\\ 
			\frac{c}{b} \left(w_{N-1} - \frac{y}{x} v_{N-1} \right)
			+  \left(w_1 - \frac{y}{x} v_1 \right)
			\frac{u_{N-2} }{b}
		\end{bmatrix}, 
		\begin{bmatrix}
			0 \\   z  \\   u_1 \\ \vdots \\  u_{N-2} 
		\end{bmatrix}
		\right). 
	\end{align}
Then by Lemma~\ref{app:lemm:Ci+B}: 
\begin{equation}
\resizebox{1\textwidth}{!}{$
\begin{aligned}
	\mathcal{B} (t, d, \vec r) -  \breve{\mathcal{B}}_i (z, c, \vec u )
		\mathcal{A}^{-1} ( x,a,\vec v ) 
		\mathcal{A} (y,b,\vec w) 
		&= 
		\frac{b}{a} \left( \frac{a}{b} \mathcal{B} (t, d, \vec r) 
		-   \breve{\mathcal{C}}_i 
		\left( c , 
		\begin{bmatrix}
			\frac{1}{b} \left(w_{i-1} - \frac{y}{x} v_{i-1} \right)
			   \\ \frac{z}{b} \left(w_1 - \frac{y}{x} v_1 \right) 
			\\ 
			\frac{c}{b} \left(w_2 - \frac{y}{x} v_2 \right)
			+  \left(w_1 - \frac{y}{x} v_1 \right)
			\frac{u_1}{b}
			\\ 
			\frac{c}{b} \left(w_3- \frac{y}{x} v_3 \right)
			+  \left(w_1 - \frac{y}{x} v_1 \right)
			\frac{u_2}{b}
			\\
			\vdots 
			\\ 
			\frac{c}{b} \left(w_{N-1} - \frac{y}{x} v_{N-1} \right)
			+  \left(w_1 - \frac{y}{x} v_1 \right)
			\frac{u_{N-2} }{b}
		\end{bmatrix}, 
		\begin{bmatrix}
			0 \\   z  \\   u_1 \\ \vdots \\  u_{N-2} 
		\end{bmatrix}
		\right) \right)
		%\notag
		\\
		&= -\frac{b}{a} \breve{\mathcal{C}}_i 
		\left( c - \frac{a}{b}d ,   
		\begin{bmatrix}
			\frac{1}{b} \left(w_{i-1} - \frac{y}{x} v_{i-1} \right)
			   \\ \frac{z}{b} \left(w_1 - \frac{y}{x} v_1 \right) 
			\\ 
			\frac{c}{b} \left(w_2 - \frac{y}{x} v_2 \right)
			+  \left(w_1 - \frac{y}{x} v_1 \right)
			\frac{u_1}{b}
			\\ 
			\frac{c}{b} \left(w_3- \frac{y}{x} v_3 \right)
			+  \left(w_1 - \frac{y}{x} v_1 \right)
			\frac{u_2}{b}
			\\
			\vdots 
			\\ 
			\frac{c}{b} \left(w_{N-1} - \frac{y}{x} v_{N-1} \right)
			+  \left(w_1 - \frac{y}{x} v_1 \right)
			\frac{u_{N-2} }{b}
		\end{bmatrix}, 
		\begin{bmatrix}
			0
			\\
			z -\frac{a}{b} t
			\\ 
			u_1 -\frac{a}{b} r_1
			\\
			\vdots 
			\\
			u_{N-2} - \frac{a}{b} r_{N-2}  
		\end{bmatrix}
		\right).
\end{aligned}
$
}
\end{equation}
So far: 
\begin{align}
	\det \left[ \mathcal{B} (t, d, \vec r) -  \breve{\mathcal{B}}_i (z, c, \vec u )
		\mathcal{A}^{-1} ( x,a,\vec v ) 
		\mathcal{A} (y,b,\vec w) \right] 
		&= \left(- \frac{b}{a} \right)^N \det 
		\breve{\mathcal{C}}_i 
		\left( c - \frac{a}{b}d ,  
		\vec{\widetilde V}, 
		\vec{\widetilde W}
		\right)
		\notag
		\\
		&=- \left(- \frac{b}{a} \right)^N \cdot 
		\left[ 
			\left( c- \frac{a}{b}d \right)^{N-2}    
			\widetilde V_1   \widetilde W_2   
	+ 	\left( c- \frac{a}{b}d \right)^{N-3} \left( 
	  \widetilde W_i   \widetilde V_2 -  \widetilde V_i \widetilde W_2 \right)\right],  
	  \label{app:lemm:final-det-eq-intermediate}
\end{align}
where: 
\begin{align}
	\vec {\widetilde V} & = \begin{bmatrix}
			\frac{1}{b} \left(w_{i-1} - \frac{y}{x} v_{i-1} \right)
			   \\ \frac{z}{b} \left(w_1 - \frac{y}{x} v_1 \right) 
			\\ 
			\frac{c}{b} \left(w_2 - \frac{y}{x} v_2 \right)
			+  \left(w_1 - \frac{y}{x} v_1 \right)
			\frac{u_1}{b}
			\\ 
			\frac{c}{b} \left(w_3- \frac{y}{x} v_3 \right)
			+  \left(w_1 - \frac{y}{x} v_1 \right)
			\frac{u_2}{b}
			\\
			\vdots 
			\\ 
			\frac{c}{b} \left(w_{N-1} - \frac{y}{x} v_{N-1} \right)
			+  \left(w_1 - \frac{y}{x} v_1 \right)
			\frac{u_{N-2} }{b}
		\end{bmatrix}, 
		& 
	\vec {\widetilde W} & =  \begin{bmatrix}
			0
			\\
			z -\frac{a}{b} t
			\\ 
			u_1 -\frac{a}{b} r_1
			\\
			\vdots 
			\\
			u_{N-2} - \frac{a}{b} r_{N-2}  
		\end{bmatrix},
\end{align}
\ie 
\begin{subequations}
\begin{align}
	\widetilde V_1 & = \frac{1}{b} \left(w_{i-1} - \frac{y}{x} v_{i-1} \right), 
	&
	\widetilde W_1 & = 0,
	\\  
	\widetilde V_2 & =  \frac{z}{b} \left(w_1 - \frac{y}{x} v_1 \right),
	&
	\widetilde W_2 & = z -\frac{a}{b} t,
	\\ 
	\widetilde V_i & = \frac{c}{b} \left(w_{i-1} - \frac{y}{x} v_{i-1} \right)
			+  \left(w_1 - \frac{y}{x} v_1 \right)
			\frac{u_{i-2}}{b},
	&
	\widetilde W_i & = u_{i-2} - \frac{a}{b} r_{i-2}.
\end{align}	
\end{subequations}
Expanding \eqref{app:lemm:final-det-eq-intermediate}, we obtain the full expression \eqref{app:lemm:final-det-eq} in terms of original variables.  

\end{proof}

\printbibliography
%\bibliographystyle{alphaurl}
%\bibliography{bibliography.bib}
\end{document}